\documentclass[12pt]{article}
\usepackage{comment}
\usepackage[utf8]{inputenc}
\usepackage[T1]{fontenc}
\DeclareMathAlphabet{\mathpzc}{OT1}{pzc}{m}{it}
\usepackage{geometry}

\usepackage[dvipsnames]{xcolor}

\newcommand{\revision}[1]{#1}

\usepackage{amsfonts}
\usepackage{amssymb}
\usepackage{amsthm}
\usepackage{amsmath}
\usepackage{amscd}
\usepackage{authblk}
\usepackage[shortlabels]{enumitem}
\usepackage{mathrsfs}
\usepackage{tikz}
\usetikzlibrary{calc,arrows,decorations.pathreplacing}
\usepackage{nicefrac, xfrac}
\usepackage{mathtools,xparse}
\usepackage{bm,bbm}
\usepackage{algorithm,algpseudocode}
\usepackage{stackengine}

\newcommand{\bbGamma}{\reflectbox{\rotatebox[origin=c]{180}{$\mathbb L$}}}
\newcommand{\bi}{\mathbbm{i}}
\newcommand{\Bgm}{\bbGamma}

\usepackage[pagebackref, colorlinks = true,
linkcolor = red,
urlcolor  = blue,
citecolor = blue,
anchorcolor = blue]{hyperref}
\hypersetup{pdftitle={\@title},pdfauthor={\@author}}

\theoremstyle{plain}

\newtheorem{theorem}{Theorem}[section]

\newtheorem{proposition}[theorem]{Proposition}
\newtheorem{conjecture}[theorem]{Conjecture}

\theoremstyle{definition}

\newtheorem*{theorem*}{Theorem}

\renewcommand{\emph}[1]{\textit{#1}} 

\renewcommand{\epsilon}{\varepsilon}
\newcommand{\R}{\mathbb{R}}

\newcommand{\spn}{\ensuremath{\mathrm{span}}}

\newcommand{\sub}{\ensuremath{ \subseteq}} 

\newcommand{\vertiii}[1]{{\left\vert\kern-0.25ex\left\vert\kern-0.25ex\left\vert #1
		\right\vert\kern-0.25ex\right\vert\kern-0.25ex\right\vert}}

\newcommand{\Lm}{\ensuremath{\Lambda}}
\newcommand{\Gm}{\ensuremath{\Gamma}}

\newcommand{\Dl}{\ensuremath{\Delta}}

\newcommand{\lm}{\ensuremath{\lambda}}

\newcommand{\al}{\ensuremath{\alpha}}

\newcommand{\ep}{\ensuremath{\epsilon}}

\DeclareSymbolFont{bbold}{U}{bbold}{m}{n}
\DeclareSymbolFontAlphabet{\mathbbold}{bbold}

\newcommand{\cO}{\ensuremath{\mathcal{O}}}
\newcommand{\cP}{\ensuremath{\mathcal{P}}}

\newcommand{\AAA}{\ensuremath{\mathbb A}}%
\newcommand{\MM}{\ensuremath{\mathbb M}}

\newcommand{\RR}{\ensuremath{\mathbb R}}

\def\dir{\text{{\rm dir}}}

\def\spat{\text{{\rm spat}}}

\providecommand{\phantomsection}{}
\AtBeginDocument{\let\textlabel\label}
\makeatletter
\newcommand{\mylabel}[2]{\raisebox{.7\normalbaselineskip}{\phantomsection}(#1)%
	\def\@currentlabel{#1}\textlabel{#2}}
\makeatother

\makeatletter
\newcommand\xlabel[2][]{\phantomsection\def\@currentlabelname{#1}\label{#2}}
\makeatother

\ExplSyntaxOn
\NewDocumentCommand{\mathlist}{ O{,} m m }
 {
  \egreg_mathlist:nnn { #1 } { #2 } { #3 }
 }

\ExplSyntaxOff

\allowdisplaybreaks

\numberwithin{equation}{section}

\graphicspath{{graphics_final/}{graphics_old/}}

\title{An inflated dynamic Laplacian to track the emergence and disappearance of semi-material coherent sets}

\author[1]{Jason Atnip}
\author[1]{Gary Froyland}
\author[2]{Peter Koltai}

\affil[1]{School of Mathematics and Statistics, University of New South Wales, Sydney, NSW 2052, Australia \protect\\  {\small \texttt{j.atnip@unsw.edu.au, g.froyland@unsw.edu.au}}}
\affil[2]{Department of Mathematics, University of Bayreuth, Universitätsstraße 30, 95440 Bayreuth, Germany \protect\\  {\small \texttt{peter.koltai@uni-bayreuth.de}}}

\begin{document}

\maketitle

\begin{abstract}
Lagrangian methods continue to stand at the forefront of the analysis of time-dependent dynamical systems.
Most Lagrangian methods have criteria that must be fulfilled by trajectories as they are followed throughout a \emph{given finite flow duration}.
This key strength of Lagrangian methods can also be a limitation in more complex evolving environments.
It places a high importance on selecting a time window that produces useful results, and these results may vary significantly with changes in the flow duration.
We show how to overcome this drawback in the tracking of coherent flow features.
Finite-time coherent sets (FTCS) are material objects that strongly resist mixing in complicated nonlinear flows.
Like other materially coherent objects, by definition they must retain their coherence properties throughout the specified flow duration.
Recent work [Froyland and Koltai, \emph{CPAM}, 2023] introduced the notion of \emph{semi-material} FTCS, whereby a balance is struck between the material nature and the coherence properties of  FTCS.
This balance provides the flexibility for FTCS to come and go, merge and separate, or undergo other changes as the governing unsteady flow experiences dramatic shifts.
The purpose of this work is to illustrate the utility of the inflated dynamic Laplacian introduced in [Froyland and Koltai, \emph{CPAM}, 2023] in a range of dynamical systems that are challenging to analyse by standard Lagrangian means, and to provide an efficient meshfree numerical approach for the discretisation of the inflated dynamic Laplacian.
\end{abstract}

\section{Introduction}

Lagrangian methods for complex fluid flows have been at the forefront of fluid flow analysis for the last two decades.
One of the key descriptors of fluid flow is the separation of coherent flow regions from more turbulent or chaotic regions. 
Early approaches to the study of transport \cite{MMP84,romkedar90,pierrehumbert91,pierrehumbertyang93,pojehaller98} were followed by a wide variety of Lagrangian techniques focusing on various Lagrangian flow properties; a very brief sample includes \cite{shadden-etal,rypina2011investigating,allshouse2012detecting,budivsic2012geometry,haller2013coherent,ma2014differential,AlPe15}. 
Concurrently, growing out of techniques to find almost-invariant sets through the eigenfunctions of Perron--Frobenius operators \cite{DJ99, F05, FP09, FrJuKo13}, the concept of coherent sets \cite{FSM10,F13} arose, enabling transfer operator methods to be applied to general time-varying dynamics.
These probabilistic transfer operator constructions led to new perspectives and techniques from spectral geometry, and to a dynamic Laplace operator~\cite{Fro15}.
It is this viewpoint that we expand upon in the present work.

In unsteady flows with explicit time dependence, these coherent and turbulent regions are themselves time dependent and move about in the flow domain or the phase space of the dynamical system. 
\revision{Applications include mapping and tracking the Antarctic polar vortex \cite{FSM10}, transport of heat and salinity in ocean eddies \cite{FrEtAl12,FJ18,FRS19}, molecular dynamics subject to non-stationary forcing~\cite{KoCiSch16,KWNS18}, kinetically tracking the gulf stream \cite{FRS19}, mapping deep-ocean coherence at the basin and sub-basin scale \cite{ABFS22}, transporting heat in turbulent convection~\cite{vieweg2024lagrangian},  air parcels feeding atmospheric blocks~\cite{SchoeEtAl25}, and  diagnosing ocean fronts \cite{DKF25}.}

Traditionally, to properly pose a question concerning coherence, a time window is specified over which coherence is sought. 
In some physical settings, there may be a natural \revision{time interval} of interest, but in many situations \revision{coherent objects have varying lifetimes} across the domain or phase space. 
\revision{This is the case, for example in turbulent fluid flow, where coherent vortices or other structures appear and disappear, and geophysical flows in the ocean or atmosphere that see continual birth and death of coherent eddies or vortices.}
\revision{The time-expanded paradigm we will describe has also been developed in a related form to find transient or quasi-stationary almost-invariant sets, which are constrained to  remain approximately stationary in space.
This restriction is particularly suitable for applications such as diagnosing atmospheric blocks 
\cite{BadzaFroyland24} and capturing turbulent superstructures \cite{Badzaetal26} in convection-driven fluids.}

With unknown or variable \revision{coherence timescales and lifetimes}, traditional Lagrangian methods are difficult to apply.
Various (typically time-window-based) approaches have been proposed \cite{FrEtAl12,macmillan20,AnKaBV20,blachut2020tale,ElA21,DFK22, Blachut2023}, but they either require \emph{a priori} knowledge of timescales or are  computationally expensive. 
This motivated the approach of \cite{FrKo23}, which performs a single computation over a wide time window to identify multiple regions of coherence with varying and possibly coexisting individual lifetimes.
A relaxed notion of semi-material coherent sets was introduced in \cite{FrKo23} to enable coherent regions to appear and disappear over time.
The central mathematical object in \cite{FrKo23} is an inflated dynamic Laplace operator, whose dominant eigenfunctions encode semi-material coherent sets.
The present work further develops the theory and numerics of the inflated dynamic Laplacian methodology, enabling the analysis of more complicated dynamics than that treated in~\cite{FrKo23}.

While \cite{FrKo23} treated the inflated dynamic Laplacian with Neumann boundary conditions, here we provide corresponding theory for the spectrum, and  associated geometric and functional quantities with Dirichlet boundary conditions; \revision{see Theorem~\ref{thm:eigbounds}}.
Dirichlet boundary conditions handle the situations where there is no natural boundary for the available Lagrangian trajectory data, or where coherent regions are not sought near the boundary of the domain~\cite{FJ18}.
To construct an estimate of the inflated dynamic Laplacian from trajectory data, \cite{FrKo23} extended the finite-element approach of~\cite{FJ18}.
In the present work we employ \revision{a bespoke version of integral operator approximation, based on} diffusion maps~\cite{CoLa06}, and put considerable effort into reducing the computational load of diffusion maps through operator-splitting techniques~\cite{Str68}, separating the space and time components.
\revision{This construction is detailed in Section~\ref{sec:dmapsDL}.}
To overcome orthogonality restrictions of the eigenfunctions of the inflated dynamic Laplacian with Neumann boundary conditions, we introduce an ``augmentation'' of the eigenfunction collection to fully extract all useful spatiotemporal information, even with multiple coexisting coherent regions.

These advances are illustrated in three case studies.
The first is a double-gyre system that rapidly switches the relative sizes of the two gyres.
Our semi-material coherent sets are able to follow this rapid change in volume in the incompressible flow, which is a highly non-material change.

The second case study concerns the breakup of the Antarctic polar vortex during Sept--Oct~2002. 
Our semi-material coherent sets identify the initial polar vortex, which then shrinks dramatically during the breakup period, and then grows again to a smaller vortex afterward.
Our semi-material approach avoids the filamentation of the vortex post-breakup that occurs with traditional Lagrangian techniques that identify strictly material objects.

Our final case study combines three regimes:  regular autonomous double-gyre dynamics, followed by mixing non-autonomous double-gyre dynamics, followed again by regular autonomous double-gyre dynamics.
In this model, the ``cores'' of the two gyres persist throughout all three phases of the dynamics, while the fluid surrounding the cores undergoes changes from regular to mixing.
Using the eigenfunction augmentation mentioned above, our inflated dynamic Laplacian correctly identifies the beginning and end of the three regimes, and the coherent regions within each regime. 
The central ``cores'' are correctly identified as persisting throughout all three phases.

An outline of the paper is as follows.
In Section~\ref{sec:FTCS} we recap the salient properties of the dynamic Laplacian \cite{Fro15}.
Sections~\ref{ssec:timexpansion}--\ref{ssec:iDL} detail the inflated dynamic Laplacian\revision{, and review material from}~\cite{FrKo23}, before Section~\ref{ssec:dirichlet} introduces the definitions and theory required for Dirichlet boundary conditions in space.
Section~\ref{sec:practical} describes some practical considerations and how to use the eigenfunctions of the inflated dynamic Laplacian to distinguish coherent and mixing regions.
Section~\ref{sec:dmapsDL} introduces our specialised diffusion map approach to estimate the inflated dynamic Laplacian.
Our numerical case studies are presented in Section~\ref{sec:numericalexamples}, and Section~\ref{sec:efunstruc} provides a more detailed analysis of the structure of the eigenfunctions of the inflated dynamic Laplacian and further explanations and illustrations of the eigenfunction augmentation procedure.

\section{Finite-time coherent sets, the dynamic Laplacian, spectrum and eigenfunctions}
\label{sec:FTCS}

In this section we briefly set up notation and background for the dynamic Laplacian \cite{Fro15,FrKw17} and related averaged quantities before introducing the inflated dynamic Laplacian in the following section.
Let $M\sub\RR^d$ be a $d$-dimensional, connected, compact manifold, and consider a smooth time-dependent velocity field $v:[0,\tau]\times M\to\RR^d$, over a finite time duration $[0,\tau]$.
For the sake of simplicity, we assume $v(t,\cdot)$ is divergence free for all~$t\in[0,\tau]$.
Let $\phi_t:M\to\phi_t(M)$ be the flow map generated by $v$ from time $0$ to time $t$;
for each $t\in[0,\tau]$, set~$M_t=\phi_t(M)$.

Let $g$ denote the usual Euclidean metric on $\R^d$.
For each $t\in [0,\tau]$, define the pullback metric $g_t := \phi_t^*(g_{|\phi_t(M)})$ on~$M$. 
In global coordinates in $\mathbb{R}^d$, for $x\in M$ and $u,w\in \mathbb{R}^d$ one has $g_t(x)(u,w) = \langle D\phi_t(x)u, D\phi_t(x)w \rangle$, where $D\phi_t(x)$ denotes the Jacobian of $\phi_t$ in~$x$. 
The natural coordinate representation of the $g_t$ at $x\in M$ is thus $g_t(x)=D\phi_t(x)^\top D\phi_t(x)$, namely the Cauchy--Green tensor.

Suppose $\Gm\sub M$ is a codimension-1 $C^\infty$ surface disconnecting $M$ into the disjoint union $M=A_1\cup\Gm\cup A_2$ where $A_1$ and $A_2$ are connected submanifolds.
We are particularly interested in finding those disconnectors whose codimension-1 volume does not significantly grow when evolved by the dynamics $\phi_t$ because such distinguished $\Gamma$ correspond to a boundary of a coherent set;  see Figure \ref{fig:FTCSschematic}, which shows two possible disconnectors~$\Gamma$.
\begin{figure}
    \centering
\includegraphics[width=0.45\textwidth]{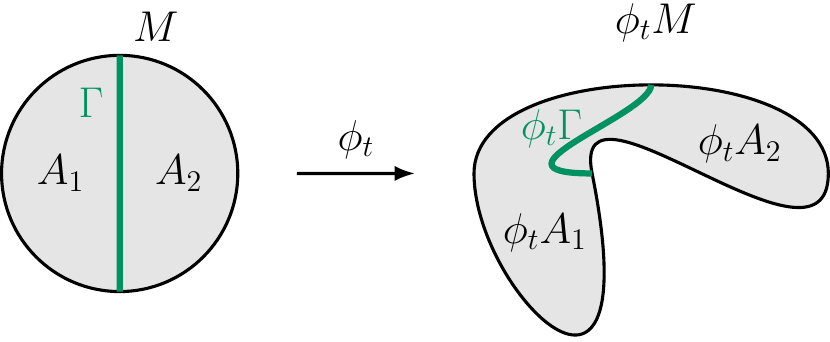}\hfill
\includegraphics[width=0.45\textwidth]{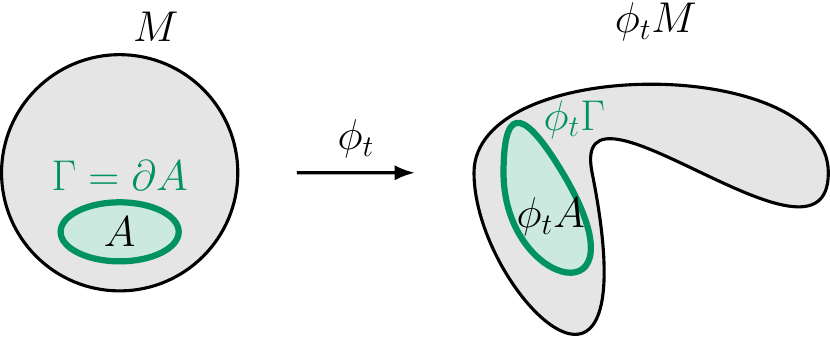}
    \caption{Left: The curve $\Gamma$ disconnects the manifold $M$; $M$ and $\Gamma$ are evolved by $\phi_t$.  Right: The curve $\Gamma$ again disconnects $M$, now without touching the boundary of~$M$. In this case, $A$ could be $A_1$ and the complement of $A$ (minus $\Gamma$) could be~$A_2$.
    }
    \label{fig:FTCSschematic}
\end{figure}
Let $\ell$ denote Lebesgue measure on $M$, $\iota:\Gm\xhookrightarrow{} M$ denote the inclusion map, $\iota^* g_t$ the induced metric on~$\Gm$, and $V_{\iota^*g_t}$ the corresponding volume form on~$\Gm$. Recall the dynamic Cheeger constant (see \cite[equation (20)]{Fro15} or \cite[equations (3.6)--(3.7)]{FrKw17}) is given by
\begin{align}\label{eq: nm cheeger const}
    h^D:=\inf_\Gm\frac{\frac{1}{\tau}\int_0^\tau V_{\iota^*g_t}(\Gm)\, dt}
    {\min\{\ell(A_1),\ell(A_2)\}}.
\end{align}
The quantity $h^D$ describes the average {codimension-1} volume of the least-growing evolving disconnector~$\Gm$, relative to the {volume of the} submanifolds {that} $\Gm$ disconnects {$M$ into}.

We also recall that the dynamic Sobolev constant (see \cite[equation (21)]{Fro15} or \cite[Section 3.2]{FrKw17}) is given by\footnote{{Recall $\| u \|_{g_t}^2:= g_t(u,u)$ for $u\in\R^d$. Further, recall} $\nabla_{g}f$ is the unique vector field on $M$ that satisfies
$g(\nabla_gf, w)=w(f)$
for all vector fields $w:M\to\RR^d$, and $w(f)$ is the Lie derivative of $f:M\to\RR^d$ with respect to $w$.}
\begin{align}\label{eq: dyn Sobolev const}
    s^D:=\inf_{f\in C^\infty(M)}\frac{\frac{1}{\tau}\int_0^\tau\|\nabla_{g_t}f\|_{g_t}\, d\ell dt}{\inf_{\al\in\RR}\int_M|f-\al|\, d\ell}.
\end{align}
The quantity $s^D$ is equal to $h^D$; see \cite{Fro15} and \cite{FrKw17}.
This dynamic Federer--Fleming Theorem connects the geometric minimisation (\ref{eq: nm cheeger const}) with the functional minimisation (\ref{eq: dyn Sobolev const}), and in fact the minimising $f$ in (\ref{eq: dyn Sobolev const}) will take on  positive values on one side of the minimising disconnector $\Gamma$ in (\ref{eq: nm cheeger const}) and negative values on the other side.
Solving the minimisation problem (\ref{eq: dyn Sobolev const}) is computationally challenging, but a relaxation of (\ref{eq: dyn Sobolev const}) is easily solved as the second eigenfunction of the dynamic Laplace operator;
 see Figure~\ref{fig:DLexample}.
\begin{figure}[h]
    \centering
    \includegraphics[width = 0.49\textwidth]{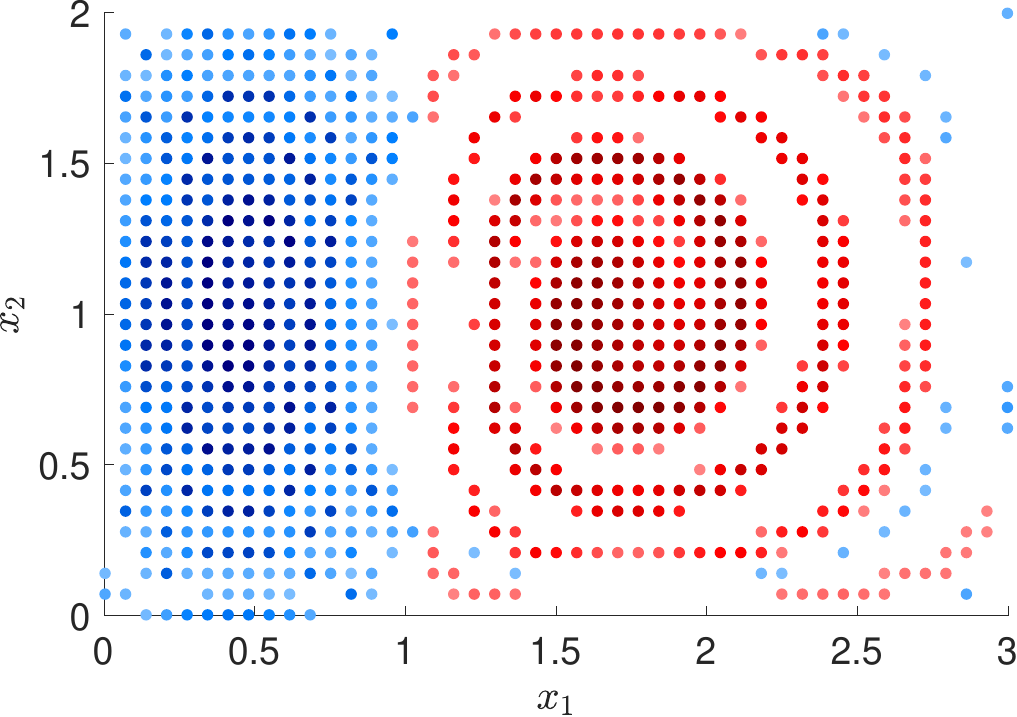}
    \hfill
\includegraphics[width=0.49\textwidth]{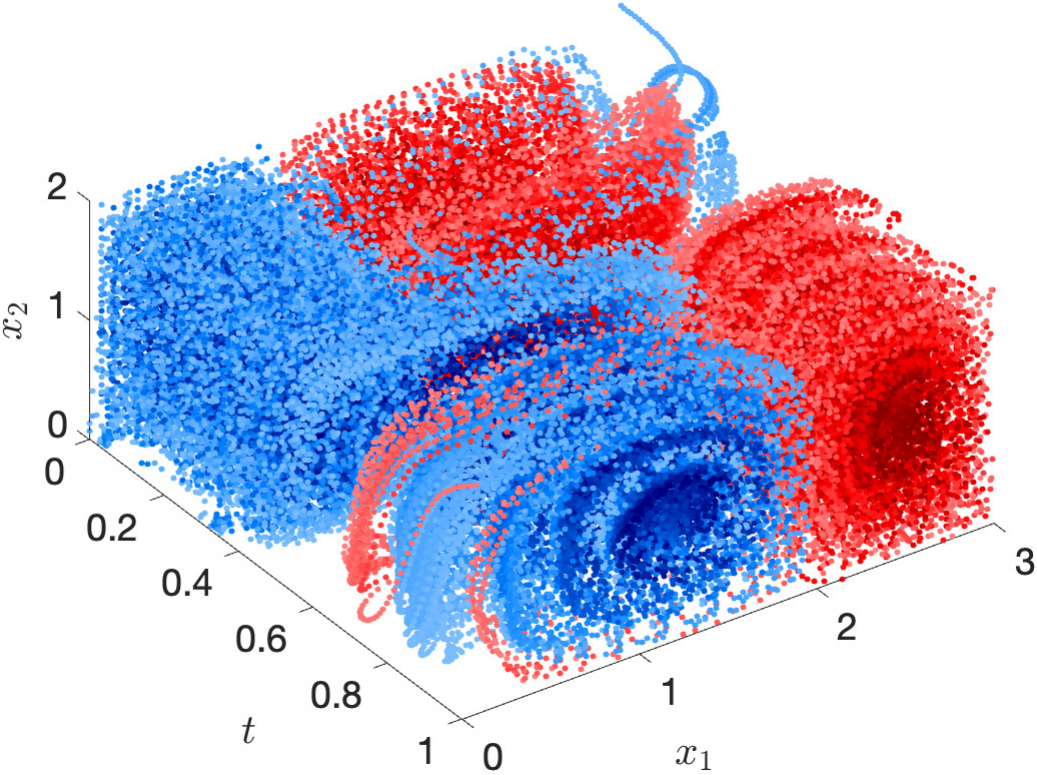}
    \caption{Left: The eigenfunction $f_2$ of $\Delta^D$ corresponding to $\lambda_2^D$;  here the extreme values of $f_2$ -- shown in red and blue -- clearly highlight the two coherent sets.
    Right: The same eigenfunction evolved forward and plotted in spacetime. This plot demonstrates that the red and blue regions of the two-dimensional phase space remain largely dynamically disconnected and with relatively small evolving one-dimensional boundary.  Both images are displayed with an absolute value cutoff of~0.25. {The underlying flow is defined and further analysed in Section~\ref{sec:switching}.}
    }
    \label{fig:DLexample}
\end{figure}

The dynamic Laplace operator \cite{Fro15} $\Delta^D: \mathrm{dom}(\Delta^D) \subset L^2(M)\to L^2(M)$ is given by a time-average of Laplace--Beltrami operators $\Delta_{g_t}$ for metrics $g_t$ on $M$:
\begin{align}\label{eq: dyn Laplacian}
    \Dl^D:=\frac{1}{\tau}\int_0^\tau \Dl_{g_t}\, dt,
\end{align}
{understood as a Bochner integral pointwise on its domain. 
Recall that the metrics $g_t$ are the pullbacks of the Euclidean metric from the manifold $\phi_t(M)$ to the manifold $M$ via~$\phi_t$.
The dynamic Laplacian is a symmetric, elliptic operator with countable discrete spectrum diverging to~$-\infty$.
Our primary interest is in the eigenvalues close to zero and their corresponding eigenfunctions.

The dynamic Cheeger inequality \cite{Fro15, FrKw17} provides a link between the dynamic geometry and the spectrum of the dynamic Laplacian:
\begin{align}\label{eq dyn cheeger ineq}
    h^D\leq 2\sqrt{-\lm_2^D},
\end{align}
where $\lm_2^D$ is the first (non-trivial) eigenvalue of the dynamic Laplacian $\Dl^D$ given in~\eqref{eq: dyn Laplacian}.
In particular, if $\lambda_2^D$ is close to zero, this forces the existence of a dynamically small disconnector.
\revision{Note that~\eqref{eq: nm cheeger const} and the definition of $g_t$ implies that coherence is invariant under the addition of time-dependent rigid-body transformations to $\phi_t$.
For the same reason, with such an addition of motion, the eigenfunctions of $\Delta^D$ are themselves simply transformed by the same rigid motion (see \cite[Theorem 4.3]{Fro15} for a formal statement).
This property is often called \emph{frame invariance}.}

\section{Semi-material finite-time coherent sets and the inflated dynamic Laplacian}
\label{sec:semimaterialCS}

The dynamic Laplacian performs very well when the coherent sets remain relatively coherent throughout the flow interval $[0,\tau]$.
Our main interest in the present work is to address situations where coherent sets come and go, or merge and split.

\subsection{Time-expanded domains and semi-material coherent sets}
\label{ssec:timexpansion}
We wish to know when coherent sets are present and to detect regime changes when they occur.
In order to do this, we append a time-coordinate to our phase space to create a time-expanded domain $\MM_0$, in which each manifold $(M,g_t)$ is considered as a $t$-fibre, given by
\begin{align}\label{eq: M_0 def}
    \MM_0:=\bigcup_{t\in[0,\tau]}\{t\}\times M = [0,\tau]\times M.
\end{align}
We will also consider the co-evolved spacetime manifold $\MM_1$, given by
\begin{align}\label{eq: M_1 def}
    \MM_1:=\bigcup_{t\in[0,\tau]}\{t\}\times \phi_t(M),
\end{align}
which is the standard time-expansion used in the dynamical systems community.

The canonical mapping $\Phi:\MM_0\to\MM_1$ associating initial conditions with their trajectories is given by
\begin{align}\label{eq: canon map}
    \Phi(t,x):=(t,\phi_t(x)).
\end{align}

To codify boundaries of  finite-time coherent sets in time-augmented space we now consider disconnectors $\Bgm$ of $\MM_0$ and their co-evolved counterparts $\Phi(\Bgm)$, which disconnect $\MM_1$.
A very special way to construct a $\Bgm$ is to take a disconnector $\Gamma$ of $M$ and copy it across time;  i.e.\ $\Bgm_{\rm copy}=[0,\tau]\times\Gamma$;  see the upper left panel in Figure~\ref{fig:timeextension}.
\begin{figure}
    \centering
    \includegraphics[width = .9\textwidth]{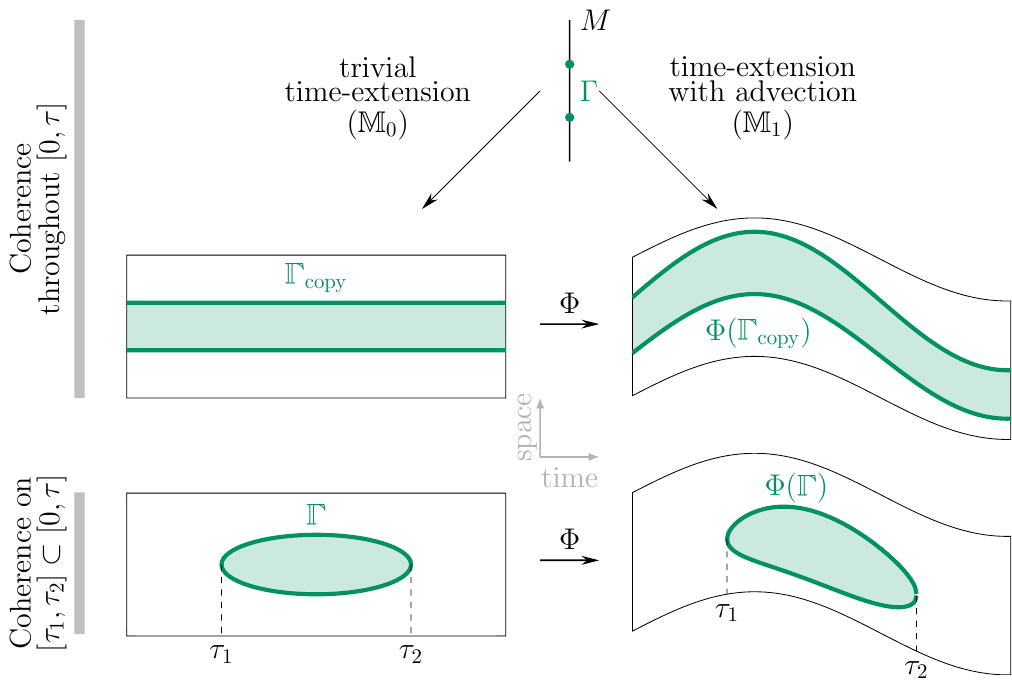}
    \caption{
    If coherence in the system persists throughout the entire time interval $[0,\tau]$, a material coherent set provides a small dynamic Cheeger constant~$h^D$.
    The top of the diagram shows a single disconnector $\Gamma\subset M$
    (green dots at the very top of the figure). Upper left: by trivial copying in time we obtain the disconnector $\Bgm_{\rm copy} = [0,\tau] \times \Gm \subset\MM_0$. We visualise one of the sets it disconnects $\MM_0$ into in pale green. Upper right: by evolving $\Bgm_{\rm copy}$ forward in time with the dynamics from time 0 to time $\tau$ we trace out the disconnector $\bigcup_{t\in[0,\tau]}\{t\}\times\phi_t(\Gamma) \subset \MM_1$.
    The lower row of the diagram concerns the situation where there is a coherent set present \emph{only for part of the time interval}, say a subinterval $[\tau_1,\tau_2]\subset [0,\tau]$.  Lower right: following the dynamics, a coherent set appears at $\tau_1$ from a small expanding core, exists for a while, and then shrinks and dissipates completely at~$\tau_2$. Lower left:  We pull back the lower right image to time $t=0$ using the inverse of~$\Phi$ to obtain the disconnector~$\Bgm \subset \MM_0$.
    }
    \label{fig:timeextension}
\end{figure}
The co-evolved version would be $\Phi(\Bgm_{\rm copy})=\bigcup_{t\in[0,\tau]}\{t\}\times\phi_t(\Gamma)$.
If $\Gamma$ is the boundary of a finite-time coherent set (i.e.\ a set whose boundary remains relatively small \revision{relative to its volume} \emph{throughout} the time duration $[0,\tau]$) then the codimension-1 volume of $\Bgm_{\rm copy}$ and $\Phi(\Bgm_{\rm copy})$ should also be low.

We now point out an important difference between $\MM_0$ and $\MM_1$, which illustrates why $\MM_0$ is the more natural space to work in.
Consider a disconnector $\Gamma$ of $M$ subjected to identity dynamics in time;  that is,~$\phi_t(x)=x$.
The disconnector $\Phi(\Bgm_{\rm copy})$ is identical to~$\Bgm_{\rm copy}$.
Now consider the same disconnector $\Gamma$ of $M$, but instead subjected to dynamics $\phi_t(x)=x+\alpha t$ for some nonzero~$\alpha\in\mathbb{R}^d$.
The new disconnector $\Bgm':=\Phi(\Bgm_{\rm copy})$ of $\MM_1$ will now have a larger size in $\MM_1$ than the one we previously considered, due to the translation;  see Figure~\ref{fig:M0vsM1}.
\begin{figure}[htb]
    \centering
    \includegraphics[width = 0.75\textwidth]{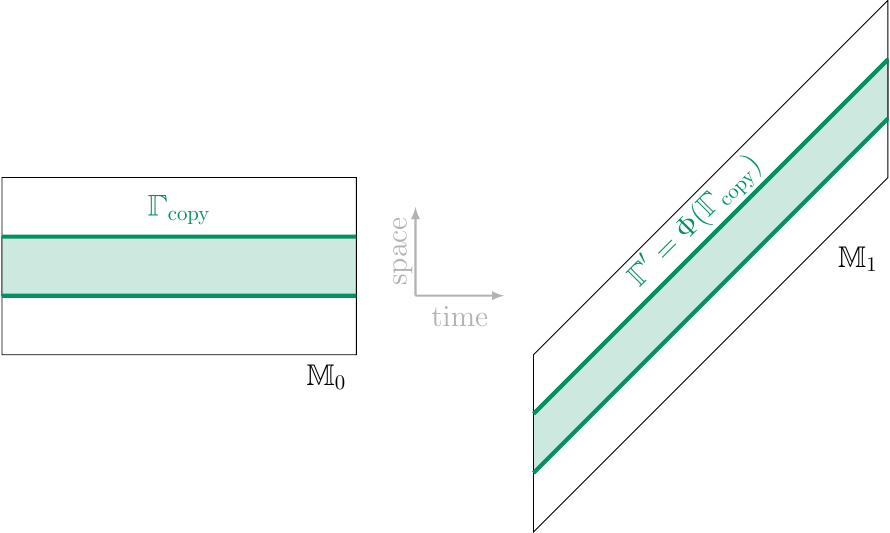}
\caption{ {Illustration of $\MM_0$ vs $\MM_1$ and associated disconnectors  for steady translational dynamics~$\phi_tx = x+\alpha t$, {with $t\in [0,\tau]$.} {Left: the disconnector $\bbGamma_{\rm copy}$ in $\MM_0$ is independent of $\alpha$.} {Right: } the velocity $\alpha$ of {the} translation {$\phi_t$} changes the measure of the spacetime disconnector~$\bbGamma'=\Phi(\bbGamma_{\rm copy})$ {in $\MM_1$, even though the translational dynamics has no effect on} the coherence of the set.}} 
    \label{fig:M0vsM1}
\end{figure}
In terms of coherence, translational dynamics should make no difference.
This is reflected by the facts that (i)  the disconnector $\Gamma$ of $M$ generates the \emph{same} disconnector $\Bgm_{\rm copy}$ of $\MM_0$ for any $\alpha$ and that (ii) the metrics $g_t$ on each time fibre remain independent of~$\alpha$.
Thus, it is the size of the disconnector in $\MM_0$ that is the true indicator of coherence, not the size of the disconnector in~$\MM_1$.

If there are \emph{no finite-time coherent sets that live throughout $[0,\tau]$}, it will be \emph{impossible} to find a disconnector of the form $\Bgm_{\rm copy}$ or $\Phi(\Bgm_{\rm copy})$ whose boundary has small volume because the evolved disconnectors $\phi_t(\Gamma)$ will be large for some non-trivial amount of time.
Nevertheless, coherence may be present for \emph{some} time, and this motivates the additional freedom we gain by using \emph{general} disconnectors $\Bgm$.
Such disconnectors are not necessarily material, in the sense that the $t$-fibre of the disconnector $\Bgm$ is not the $t$-evolution of the $0$-fibre of $\Bgm$;  in symbols, $\Bgm\cap (\{t\}\times M_t)\neq\phi_t(\Bgm\cap (\{0\}\times M)).$
This is in contrast to material disconnectors of the form $\Bgm_{\rm copy}$, where such an equality holds.
If the total codimension-1 Euclidean volume of these {time-varying} disconnectors is low, then they represent boundaries of \emph{semi-material coherent sets};  see the lower left and lower right panels of
Figure~\ref{fig:timeextension}.
What we lose in materiality, we gain in flexibility.
When a coherent set appears or disappears, the disconnectors $\Bgm$ are, in principle, able to track this emergence or vanishing by breaking materiality, as shown in the lower panels of Figure~\ref{fig:timeextension}.
Breaking materiality comes with a penalty of being less regular in time.
Disconnectors $\Bgm$ that follow coherent sets will be both regular in time and regular in space in order to maintain a small space-time volume for~$\Bgm$.

We use a parameter $a>0$ to relatively weight one-dimensional volume in time and $d$-dimensional volume in space.
Recall from \cite[Figure 2.1 and also eq.~(1.6)]{FrKo23} that at a point $(t,x)\in \MM_0$ we define local distances by the metric $G_{0,a}$ with coordinate representation
\begin{equation}
\label{eq:G0a}
    \begin{pmatrix}
        1/a^2&0\\
        0& D\phi_t(x)^\top D\phi_t(x)
    \end{pmatrix}.
\end{equation}
The lower right block is the local matrix representation of the metric $g_t$ we impose on the $t$-fibre $\{t\}\times M\subset\MM_0$.
When there is deviation of the disconnector $\Bgm$ from the ``material'' form of $\Bgm_{\rm copy}$, the parameter $a$ will penalise this deviation.
In this way, we will be able to control how much deviation from materiality we are willing to allow.
Given a codimension-1 $C^\infty$ surface $\Bgm\sub \MM_0$ which disconnects $\MM_0=\AAA_1\cup\Bgm\cup\AAA_2$ into a disjoint union with connected submanifolds  $\AAA_1$ and~$\AAA_2$, we require a way of measuring the codimension-1 volume of~$\Bgm$.
This is achieved by ``restricting'' the metric $G_{0,a}$ to~$\Bgm$.
More formally, let $\bi:\Bgm\xhookrightarrow{}\MM_0$ denote the inclusion map, $\bi^*G_{0,a}$ the induced metric on $\Bgm$,  and $V_{\bi^*G_{0,a}}$ the corresponding volume form on~$\Bgm$. 
We set
\begin{align}\label{eq: nm a Cheeger const }
    H_a:=\inf_{\footnotesize\Bgm}\frac{V_{\bi^*G_{0,a}}(\Bgm)}{\min\{V_{G_{0,a}}(\AAA_1),V_{G_{0,a}}(\AAA_2)\}},
\end{align}
where $V_{G_{0,a}}$ denotes the volume measure on $\MM_0$ with respect to~$G_{0,a}$.

For each $a\ge 0$ we define the Sobolev constant for the Riemannian manifold $(\MM_0, G_{0,a})$ by
\begin{align}\label{eq: def a Sobolev const}
    S_a:=\inf_{F\in C^\infty(\MM_0)}\frac{\int_{\MM_0}\|\nabla_{G_{0,a}}F\|_{G_{0,a}}\, dV_{G_{0,a}}}{\inf_{\al\in\RR}\int_{\MM_0}|F-\al|\, dV_{G_{0,a}}}.
\end{align}
Note that the representation \eqref{eq:G0a} of the metric $G_{0,a}$ and the assumed volume preservation of the flow implies that~$dV_{G_{0,a}}=\frac{1}{a}\,dt\, d\ell$.

Proposition~\ref{prop:lima} relates the geometric and functional minimisation problems \eqref{eq: nm cheeger const} and \eqref{eq: dyn Sobolev const} on $M$ to those on $\MM_0$, respectively (\ref{eq: nm a Cheeger const }) and~\eqref{eq: def a Sobolev const}.
\begin{proposition}
\label{prop:lima}
    We have that 
    \begin{enumerate}
        \item (\cite[Proposition 3.1]{FrKo23}) $H_a=S_a\leq s^D=h^D$ for all $a\geq 0$,
        \item (\cite[Proposition 3.1]{FrKo23}) $H_a$ and $S_a$ are nondecreasing for $a\geq 0$.
            \end{enumerate}
\end{proposition}
Part 1 says that the minimisation problems on $\MM_0$ are relaxations of the minimisation problem on $M$.
The greater flexibility of the spacetime minimisation leads to lower minima.
Part 2 says that as we link time fibres together more tightly by increasing the factor $a$, the Cheeger and Sobolev constants will grow.
We conjecture that the geometric and functional solutions on $\MM_0$ converge to the ``time-copied'' solutions on $M$ in the limit of linking time fibres infinitely tightly.
Thus, increasing the parameter $a$ makes the optimal disconnectors $\Bgm$ increasingly material until in the limit as $a\to\infty$ they become completely material. 
\begin{conjecture}
\label{conj1}
    $\lim_{a\to\infty} S_a=s^D$ and $\lim_{a\to\infty} H_a=h^D$. 
    Moreover, let  $F_a:\mathbb{M}_0\to\mathbb{R}$ minimise $S_a$, $f:M\to\mathbb{R}$ minimise $s^D$, $\Bgm_a$ minimise $H_a$, and $\Gamma$ minimise~$h^D$. 
    Then $\lim_{a\to\infty} F_a(t,x)=f(x)$ for all $t\in[0,\tau]$ and $x\in M$, and $\lim_{a\to\infty} \Bgm_a=[0,\tau]\times \Gamma$\revision{, in some appropriate sense of convergence of sets}. 
\end{conjecture}

\revision{In the degenerate limit $a\to 0$ one has completely uncontrolled deviation from material motion of coherent sets because there is no dynamical connection between time fibres.
This limit emphasises infinitesimally time-local patterns, which only capture domain geometry, not dynamics.}

\subsection{The inflated dynamic Laplacian}
\label{ssec:iDL}

The \emph{inflated dynamic Laplacian} $\Delta_{G_{0,a}} : \mathrm{dom}(\Delta_{G_{0,a}}) \subset L^2(\MM_0,G_{0,a})\to L^2(\MM_0,G_{0,a})$ is defined by
\begin{equation}
    \label{iDLdefn}
    \Delta_{G_{0,a}}F(t,\cdot)=a^2\partial_{tt}F(t,\cdot)+\Delta_{g_t}F(t,\cdot).
\end{equation}
That is, at a point $(t,x)\in\MM_0$, a one-dimensional Laplacian of strength $a^2$ is applied along the time axis, and on each $t$-fibre the Laplace--Beltrami operator $\Delta_{g_t}$ is applied in space.
The operator $\Delta_{G_{0,a}}$ is a standard Laplace--Beltrami operator with metric $G_{0,a}$ and therefore inherits all of the usual properties:  symmetry, ellipticity, and discrete spectrum diverging to~$-\infty$.
Let $\mathbb{S}_0=L^2(\mathbb{M}_0)$ and for $k\ge 1$ let $\mathbb{S}_{k}=\{F\in L^2(\mathbb{M}_0):\langle F,F_i\rangle=0, 1\le i\le k\}$.
For $k\ge 1$, one has the standard variational characterisation of eigenvalues of Laplace--Beltrami operators (recall that the volume form $V_{G_{0,a}}$ is given by $dV_{G_{0,a}} = \frac{1}{a}dt\,d\ell$)
\begin{equation}
\label{eq:LB-lamk}
\Lambda_{k,a}=-\inf_{ F\in \mathbb{S}_{k-1} } \frac{\int_{\mathbb{M}_0} \left\|\nabla_{G_{0,a}}F\right\|^2_{G_{0,a}}\ dV_{G_{0,a}}}{\int_{\mathbb{M}_0} F^2\ dV_{G_{0,a}}}.
\end{equation}
One also immediately obtains the standard Cheeger inequality (see e.g.\ \cite[Remark VI.2.4]{chavel_isoperimetry}),
\begin{align}\label{a strd Cheeger ineq}
    H_a\leq 2\sqrt{-\Lm_{2,a}},
\end{align}
where $\Lm_{2,a}$ is the first nontrivial eigenvalue of~$\Dl_{G_{0,a}}$ {with homogeneous Neumann boundary conditions}.
In our dynamic context, this relates the volume of the optimal space-time disconnector $\Bgm$ to the second eigenvalue of~$\Dl_{G_{0,a}}$.

In the same way that the second eigenfunction of $\Delta^D$ approximated the minimiser of \eqref{eq: dyn Sobolev const}, the second eigenfunction of $\Delta_{G_{0,a}}$ approximates the minimiser of~\eqref{eq: def a Sobolev const}.
This is illustrated in Figure~\ref{fig:average-vs-inflated-dg-efuns}.
\begin{figure}[htb]
    \centering
    \includegraphics[width = 0.49\textwidth]{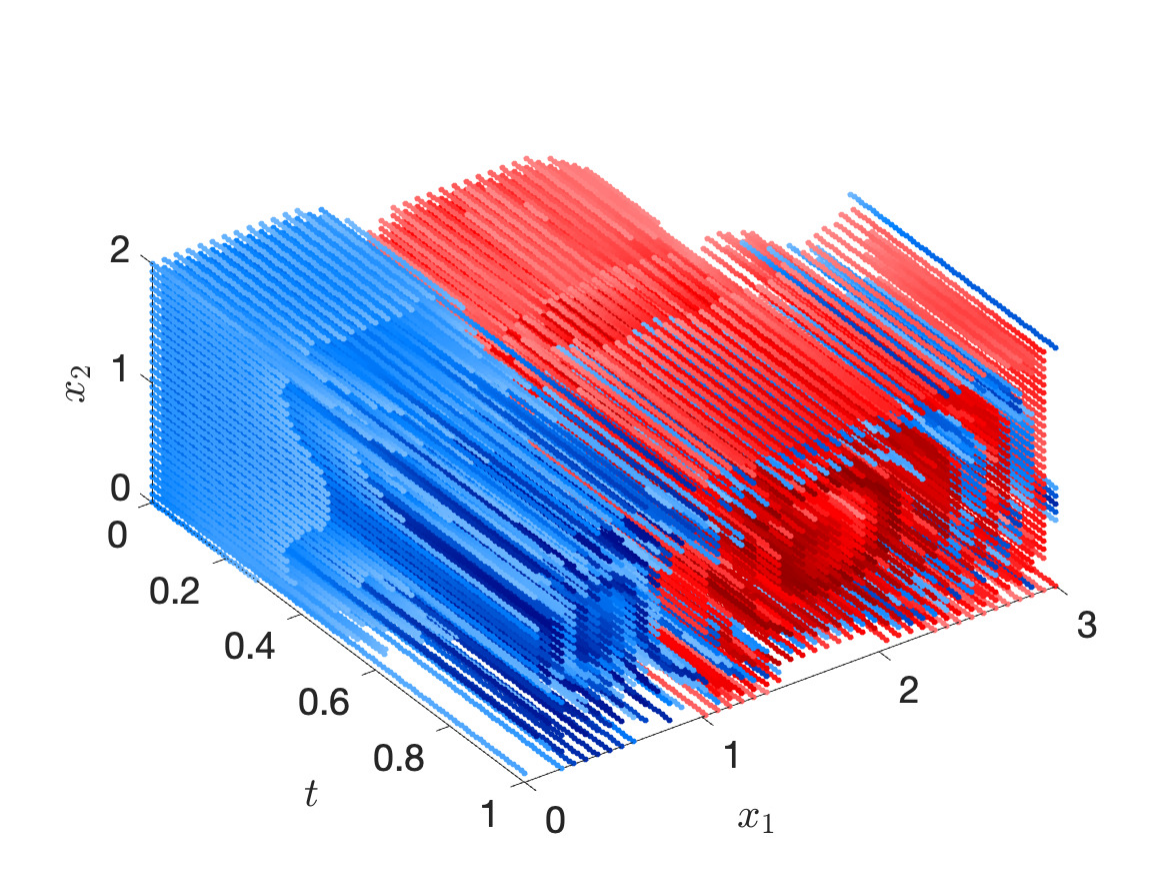}
    \hfill
    \includegraphics[width=0.49\textwidth]{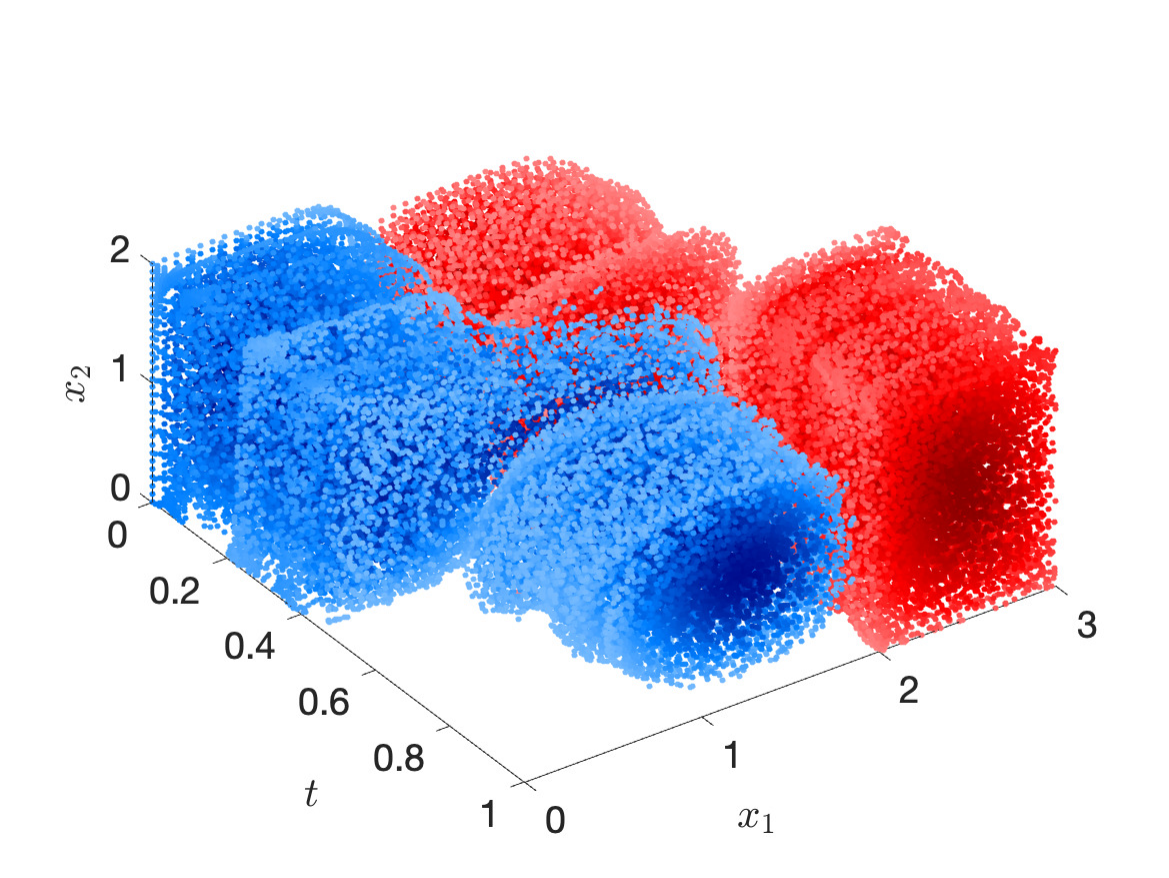}
    \caption{Third (second spatial) eigenfunction of the inflated dynamic Laplacian $\Delta_{G_{0,a}}$ with $a=0.3340, \epsilon=0.0044$ {(see Section~\ref{sec:dmapsDL} for the role of the parameter $\epsilon$ in the numerical discretisation of~$\Delta_{G_{0,a}}$)}. An absolute value cutoff of 0.25 has been applied. The extreme values of the eigenfunction -- deep red is $+1$ and deep blue is $-1$ -- 
    clearly highlight the two semi-material coherent sets.
    Left:~$\MM_0$. Material trajectories (straight lines in $\MM_0$) tend to have constant eigenfunction values. The strong change in colour of some trajectories indicate that the sets are only almost-material.
    Right:~$\MM_1$. The red and blue regions in the co-evolved time-expanded phase space remain largely disconnected, indicating strong coherence.
    Figure~\ref{fig:DGswitch_tempvar} highlights trajectories along which there is a nonmaterial change in the eigenfunction.}
    \label{fig:average-vs-inflated-dg-efuns}
\end{figure}
Compared to Figure~\ref{fig:DLexample} (right), the blue and red objects in Figure~\ref{fig:average-vs-inflated-dg-efuns} are smoother, but deviate slightly from representing a material family.
We analyse this example further in Section~\ref{sec:switching}.

\subsection{Dirichlet boundary conditions in space}
\label{ssec:dirichlet}

Sometimes $M$ is a subset of a larger ambient space and we would like to exclude disconnectors that touch the boundary of $M$ because in such situations, the boundary of $M$ is chosen somewhat arbitrarily.
For example, in flows in the ocean or atmosphere, one may choose $M$ to be a domain containing coherent objects of interest, but the precise choice of $M$ is not particularly important, and therefore should not play a strong role in selecting disconnectors.
In other situations, we may wish to explicitly exclude finite-time coherent sets that touch the boundary of the domain $M$.
For either of these reasons one can apply Dirichlet boundary conditions on $M$ to the eigenproblem for $\Delta^D$ so that the eigenfunctions are zero at the boundary of $M$ and therefore cannot take extreme values nearby.
For example, the leftmost red slice in Figure \ref{fig:PVav} illustrates the polar vortex at nominal day 0.
\begin{figure}[htb]
    \centering
    \includegraphics[width=\textwidth]{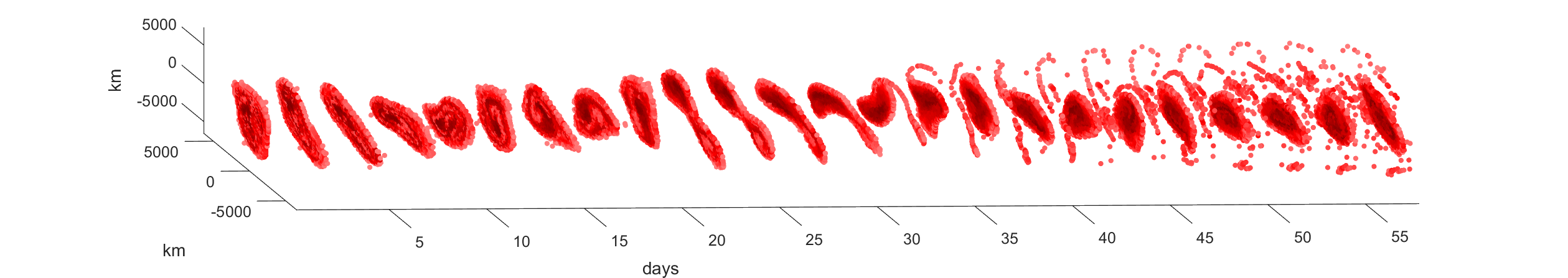}
    \caption{The leading eigenvector of the dynamic Laplacian advected through time.  The initial time slice corresponds to the initial manifold $M$ on which the dynamic Laplacian is defined and the values displayed range from 0 (white) to 1 (deep red). A cutoff of 0.25 has been applied to focus on the dominant feature (the vortex).  The slices to the right correspond to advecting the eigenvector forward in time. See Section~\ref{ssec:examplePV} for further details {and the supplementary material for movies displaying all trajectories, and with the boundary trajectories highlighted}.}
    \label{fig:PVav}
\end{figure}
This leftmost slice is a superlevel set of the leading eigenfunction of $\Delta^D$ with Dirichlet boundary conditions placed on a larger disk-like domain containing the vortex.
The vortex structure does not touch the boundary of the spatial domain. 
The subsequent slices show the first slice propagated through time.
At around day 20, the vortex begins to undergo a breakup event, and filamentation of the vortex is seen in the red slices to the right of the image, after the vortex has undergone breakup.
{For further details please refer to the supplementary movie \texttt{mov\_PV\_DynLapMode\_1.mp4}.}

On our time-expanded space $\MM_0$, in such a situation we apply Dirichlet boundary conditions in space, but retain Neumann boundary conditions on the two ``time-faces'' $\{0\}\times M$ and $\{\tau\}\times M$ of $\MM_0$ for greater flexibility.
Figure~\ref{fig:PVinfl} shows a spacetime SEBA\footnote{See Section~\ref{ssec:examplePV} for further details on the construction of this vector.} vector \cite{FRS19} obtained by inserting the leading two eigenvectors of $\Delta_{G_{0,a}}$ with Dirichlet boundary conditions in space.
\begin{figure}[htb]
    \centering
\includegraphics[width=\textwidth]{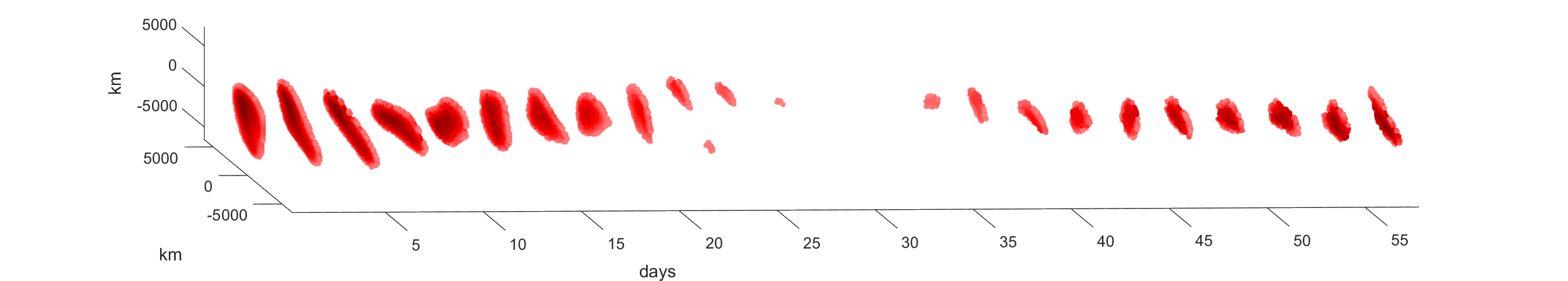}
    \caption{Maximum of two sparse vectors spanning the leading two-dimensional eigenspace of the inflated dynamic Laplacian, constructed by the SEBA \cite{FRS19} algorithm.  The values displayed range from 0 (white) to 1 (deep red) and a cutoff of 0.25 has been applied to focus on the dominant features.  See Section~\ref{ssec:examplePV} for further details.}
    \label{fig:PVinfl}
\end{figure}
Figure \ref{fig:PVav} captures the breakup of the vortex over the breakup interval from day 20 onward, through a ``pinching'', but strictly forcing a material object across the entire 60 days leads to significant filamentation toward the end of the evolution.
In contrast, Figure~\ref{fig:PVinfl} clearly shows the breakup and dissipation of the vortex between days 25 and 35, and when part of the vortex reforms after day 35, a correct reformed subvortex is captured with no filamentation.
Thus, again we see the motivation for relaxing materiality of coherent sets in situations where there are \emph{no fully coherent objects over the full flow duration}.
Further discussion of this example appears in Section~\ref{ssec:examplePV}.

We now briefly recap quantities on $M$ arising from Dirichlet boundary conditions from~\cite{FJ18}, and then extend these to $\MM_0$.
We recall the dynamic Cheeger constant for Dirichlet boundary conditions~\cite{FJ18}, 
\begin{equation}
\label{dynche}
h_\dir^D:=\inf_{A\subset M}\frac{\frac{1}{\tau}\int_0^\tau V_{\iota^*g_t}(\partial A)\ dt}{\ell(A)},
\end{equation}
where $A\subset M$ is an open submanifold of the interior of $M$ with compact closure and $C^\infty$ boundary.
Similarly, we recall the dynamic Sobolev constant \cite{FJ18}:
\begin{equation}
\label{dynsob}
s_\dir^D:=\inf_{f\in C^\infty_c(M,\mathbb{R})} \frac{\frac{1}{\tau}\int_0^\tau \int_M \|\nabla_{g_t}f\|_{g_t}\ d\ell dt}{\int_M |f|\ d\ell},
\end{equation}
where $C^\infty_c(M,\mathbb{R})$ denotes the set of non-identically vanishing, compactly supported, $C^\infty$ real-valued functions on the interior of $M$.
As in Section~\ref{sec:FTCS}, one has a dynamic Federer--Fleming Theorem \cite{FJ18}: $h^D_{\rm dir}=s^D_{\rm dir}$.
A dynamic Cheeger inequality holds \cite{FJ18} $h^D_{\rm dir}\leq 2\sqrt{-\lm_1^D}$,
where $\lm_1^D$ is the first  eigenvalue of the dynamic Laplacian $\Dl^D$ with Dirichlet boundary conditions.

We now turn to $\MM_0$.
Let $\mathbb{A}\subset\mathbb{M}_0$ be a submanifold of $[0,\tau]\times \mathrm{Int}(M)$ with compact closure in $[0,\tau]\times \mathrm{Int}(M)$
and $C^\infty$ boundary except where $\mathbb{A}$ intersects $\{0,\tau\}\times M$ if the latter is nonempty.
Let $\bi:\partial \mathbb{A}\xhookrightarrow{}\mathbb{M}_0$ denote the inclusion map, $\bi^*G_{0,a}$ the induced metric on $\partial \mathbb{A}$, and $V_{\bi^*G_{0,a}}$ the corresponding volume form on~$\partial \mathbb{A}$. Set
\begin{equation}
\label{augche}
H_{\dir, a}:=\inf_{\mathbb{A}\subset\mathbb{M}_0}\frac{V_{\bi^*G_{0,a}}(\partial\mathbb{A})}{V_{G_{0,a}}(\mathbb{A})}.
\end{equation}
We may also define a Sobolev constant for the Riemannian manifold $(\mathbb{M}_0,G_{0,a})$:
\begin{equation}
\label{augsob}
S_{\dir,a}:=\inf_{F\in C_c^\infty(\mathbb{M}_0)} \frac{\int_{\mathbb{M}_0}\|\nabla_{G_{0,a}}F\|_{G_{0,a}}\ dV_{G_{0,a}}}{\int_{\mathbb{M}_0} |F|\ dV_{G_{0,a}}},
\end{equation}
where $C^\infty_c(\mathbb{M}_0,\mathbb{R})$ denotes the set of nonidentically vanishing, $C^\infty$ real-valued functions that vanish on~$[0,\tau]\times \partial M$.
We consider the inflated dynamic Laplacian (as usual in weak form) on a Sobolev space built\footnote{Let $H^1_\dir(\MM_0)$ denote the Sobolev space of weakly differentiable $L^2(\MM_0)$ functions whose derivatives are in $L^2(\MM_0)$ and satisfy homogeneous Dirichlet boundary conditions on $[0,\tau] \times \partial M$ (i.e., ``in space''). This space can be obtained as the closure of the set $C^\infty_c(\mathbb{M}_0,\mathbb{R})$ of smooth functions that vanish on $[0,\tau] \times \partial M$ with respect to the $H^1(\MM_0)$ norm.} from these smooth functions, say ``with Dirichlet boundary conditions'' to mean Dirichlet boundary conditions on the spatial boundary.
The variational expression for the eigenvalues $\Lambda_{k,a}$ is as given in (\ref{eq:LB-lamk}), with $\mathbb{S}_k$, $k\ge 0$ restricted to functions satisfying the Dirichlet-in-space boundary conditions.
One obtains a Federer--Fleming Theorem and Cheeger inequality for these mixed boundary conditions on~$\mathbb{M}_0$.
\begin{proposition}
\label{mixedFFCheeger}
Imposing the above mixed Dirichlet-in-space and Neumann-in-time boundary conditions on the manifold $(\MM_0,G_{0,a})$, one has $H_{\rm dir,a}=S_{\rm dir,a}$ and a Cheeger inequality $H_a\leq 2\sqrt{-\Lm_{1,a}}$, where $\Lm_{1,a}<0$ is the first  eigenvalue of $\Dl_{G_{0,a}}$ with Dirichlet boundary conditions {in space}.
\end{proposition}
\begin{proof}
    For the first claim, for $\mathbb{A}$ described prior to the Proposition we set 
    $$F_\epsilon(t,x)=\left\{
  \begin{array}{ll}
    1, & \hbox{$(t,x)\in \mathbb{A}$;} \\
    1-d((t,x),\partial\mathbb{A})/\epsilon, & \hbox{$(t,x)\in \mathbb{M}_0\setminus\mathbb{A}, \ d((t,x), \partial\mathbb{A})<\epsilon$;} \\
    0, & \hbox{$(t,x)\in \mathbb{M}_0\setminus\mathbb{A}, \ d((t,x),\partial\mathbb{A})\ge\epsilon$.}
  \end{array}
\right.$$
Following the argument of \cite[Theorem II.2.1]{chavel_isoperimetry} one has that
\[
S_{{\rm dir},a}\le \lim_{\epsilon\to 0}\frac{\|\nabla F_\epsilon\|_1}{\|F_\epsilon\|_1}=\frac{V_{\bi^*G_{0,a}}(\partial\mathbb{A})}{V_{G_{0,a}}(\mathbb{A})}.
\]
Since $\mathbb{A}$ was arbitrary, we conclude $S_{{\rm dir},a}\le H_{{\rm dir},a}$.
To show $S_{{\rm dir},a}\ge H_{{\rm dir},a}$ one follows the argument in \cite[Theorem II.2.1]{chavel_isoperimetry}, almost verbatim.

For the second claim concerning the Cheeger inequality, our mixed boundary conditions fall under what Chavel \cite[p.~15]{chavel_eigenvalues} refers to as a ``mixed eigenvalue problem''. 
Our domain $\MM_0$ is a ``normal domain'' (connected manifold with compact closure and nonempty piecewise $C^\infty$ boundary) in the terminology of Chavel~\cite{chavel_eigenvalues}.
In our setting, Chavel's ``space of admissible functions'' -- denoted $\mathfrak{H}$ in \cite{chavel_eigenvalues} -- is the $H^1$-completion of $C^\infty$ functions compactly supported on the interior of our $\MM_0$ union the Neumann part of the boundary of~$\MM_0$. In such a setting, a Rayleigh Theorem~\cite[p.~16]{chavel_eigenvalues} is proven.
One may then follow the proof of the Cheeger inequality in \cite[\S IV.3 Theorem~3]{chavel_eigenvalues} almost verbatim to obtain the Cheeger inequality in our mixed boundary condition setting.
\end{proof}

Proposition~\ref{isoperprop} addresses the relationship between $h_\dir^D$ and $H_{\dir,a}$, and the behaviour of $H_{\dir,a}$ with increasing $a$;  one has the same interpretations as those discussed immediately after Proposition \ref{prop:lima}.
\begin{proposition}
\label{isoperprop}
In the Dirichlet boundary condition setting above, one has 
\begin{enumerate}
\item $H_{\dir,a}=S_{\dir,a}\le s_\dir^D=h_\dir^D$ for all $a\ge 0$.
\item $H_{\dir,a}$ and $S_{\dir,a}$ are nondecreasing in $a\ge 0$.
\end{enumerate}
\end{proposition}
\begin{proof}
\

{Part 1:} The fact that $H_{\dir,a}=S_{\dir,a}$ follows from Proposition \ref{mixedFFCheeger}.
The fact that $s_\dir^D=h_\dir^D$ follows from the Dirichlet version of the dynamic Federer--Fleming Theorem 
\cite[Theorem 1]{FJ18}.
The inequality $S_{\dir,a}\le s_\dir^D$ follows similarly to part~1 of the proof of~\cite[Proposition 3.1]{FrKo23}.
One chooses an $f_\epsilon\in C_c^\infty(M)$ and sets $F_\epsilon\in C_c^\infty(\mathbb{M}_0)$ to be~$F_\epsilon(t,x)=f_\epsilon(x)$.
One then shows that the infimand in \eqref{augsob} using $F=F_\epsilon$ is equal to the infimand of \eqref{dynsob} using $f=f_\epsilon$, with the proof proceeding as in  \cite[Proposition 3.1]{FrKo23}.

{Part 2:} This follows as in the proof of part~2 of~\cite[Proposition~3.1]{FrKo23}.
\end{proof}
We have a conjecture analogous to Conjecture \ref{conj1}.
\begin{conjecture}
    $\lim_{a\to\infty} S_{\dir,a}=s_{\dir}^D$ and $\lim_{a\to\infty} H_{\dir,a}=h_{\dir}^D$. 
    Moreover, let  $F_a:\mathbb{M}_0\to\mathbb{R}$ minimise $S_{\dir,a}$, $f:M\to\mathbb{R}$ minimise $s_{\dir}^D$, $\mathbb{A}_a$ minimise $H_{\dir,a}$, and $A$ minimise $h_{\dir}^D$. 
    Then $\lim_{a\to\infty} F_a(t,x)=f(x)$ for all $t\in[0,\tau]$ and $x\in M$, and $\lim_{a\to\infty} \mathbb{A}_a=[0,\tau]\times A$. 
\end{conjecture}

Theorem~\ref{thm:eigbounds} is the Dirichlet-boundary-condition version of~\cite[Theorem 3.3]{FrKo23}. 
Let $\Delta^D$ with Dirichlet boundary conditions have eigenvalues $0>\lambda_1^D>\lambda_2^D>\cdots$.
We show that $\Delta_{G_{0,a}}$ with Dirichlet boundary conditions in space has eigenvalues $0>\Lambda_{1,a}>\Lambda_{2,a}>\cdots$, and that the eigenvalues of these two operators satisfy certain relationships.
\begin{theorem}
\label{thm:eigbounds}
The operator $\Delta_{G_{0,a}}$ with Dirichlet-in-space boundary conditions has discrete spectrum, each eigenvalue is negative and has finite multiplicity, and the eigenvalues only accumulate at~$-\infty$. Furthermore: 
\begin{enumerate}
\item For each $k\ge 1$ and $a > 0$ one has $\lambda_k^D \le \Lambda_{k,a}$.
\item For each $k\ge 1$, $\Lambda_{k,a}$ is nonincreasing in~$a\ge 0$.
\item For each $k\ge 1$, $\lim_{a\to\infty} \Lambda_{k,a} = \lambda_k^D$.
\end{enumerate}
\end{theorem}
\begin{proof}
    See Appendix~\ref{app:eigenbounds}.
\end{proof}
One may interpret the eigenvalues of $\Delta^D$ (resp.\ $\Delta_{G_{0,a}}$) as decay rates of the corresponding eigenfunctions under the semiflow 
$\exp(\Delta^D t)$ (resp.\ $\exp(\Delta_{G_{0,a}} t)$).
Greater irregularity in an eigenfunction generally leads to greater decay.
Part 1 states that the $k^{\rm th}$ eigenvalue of $\Delta_{G_{0,a}}$ is not farther from 0 than the $k^{\rm th}$ eigenvalue of $\Delta^D$, which is consistent with the intuition that we expect the $k^{\rm th}$ eigenfunction of $\Delta_{G_{0,a}}$ to be more regular than the $k^{\rm th}$ eigenfunction of~$\Delta^D$.
Part 2 states that as distinct time slices are tied more tightly together by increased temporal diffusion, the eigenvalues $\Lambda_{k,a}$ move farther from zero, indicating faster decay and more irregular eigenfunctions.
Part 3 states that in the limit of tying different time slices together infinitely tightly (thus achieving full materiality), the eigenvalues of $\Delta_{G_{0,a}}$ converge to those of~$\Delta^D$.

\section{Practical considerations}
\label{sec:practical}

\subsection{Distinguishing spatial and temporal eigenfunctions}
\label{ssec:spatial}
For Neumann boundary conditions, in \cite{FrKo23} we called eigenfunctions of $\Delta_{G_{0,a}}$ of the form $f(t)\cdot\mathbf{1}(x)$ \emph{temporal} because they carried only temporal information.  Eigenfunctions in the $L^2$ orthogonal complement of $L^2([0,\tau])\times \spn\{\mathbf{1}\}$ were called \emph{spatial}.
{Spatial eigenfunctions can be identified by computing the spatial means on each time fibre and checking for constancy across time fibres;  see the short \cite[section~4.2]{FrKo23} for a specific implementation.}

With Dirichlet boundary conditions on the spatial boundaries of $\MM_0$, there are no temporal eigenvalues to clutter the leading eigenfunctions.
To see this, suppose there is an eigenfunction of the separable form $f(t)g(x)$.
Because of the separable structure of $\Delta_{G_{0,a}}$, this would imply that $f$ is an eigenfunction of the one-dimensional Laplacian with Neumann boundary conditions and that $g$ is an eigenfunction of $\Delta_{g_t}$
for \emph{all}~$t\in[0,\tau]$.
Thus generically there are no eigenfunctions of this separable form and the Dirichlet spectrum consists purely of informative eigenfunctions of non-separable form.

\subsection{Choosing the temporal diffusion parameter~$a$}
\label{subsec:Dirichlet_choose_a}

\revision{The parameter $a$  has a physical interpretation of a time scale: the temporal duration of semi-material coherent sets in the time interval $[0,\tau]$ as obtained from~\eqref{eq: nm a Cheeger const } are commensurate to~$a$ (see \cite[Section~4.4]{FrKo23} for more details).}
In \cite[Section 4.1]{FrKo23} we discussed an \emph{a priori} heuristic to select a minimum value of the temporal diffusion strength parameter $a$.
This heuristic was based on matching the leading nontrivial spatial and temporal eigenvalues of the inflated dynamic Laplacian, and ensured that the leading eigenfunctions were not cluttered with temporal eigenfunctions, which provided no information on spatial coherence.
Roughly speaking with such a choice of $a$ we expect a single separation into two temporal zones (given by the leading nontrivial temporal eigenfunction) to be as likely as a single separation into two spatial zones (given by the leading nontrivial spatial eigenfunction).
More generally, if we wish to encourage identifying sets of temporal extent approximately $\tau/k$ and spatial volume approximately $\ell(M)/l$ then in order that the approximate ``cost'' of the temporal and spatial variations are commensurate in the numerator of \eqref{eq:LB-lamk}, we suggest matching the $(k-1)^{\mathrm{th}}$ temporal eigenvalue with the $(l-1)^{\mathrm{th}}$ nontrivial spatial eigenvalue.

To select a suitable $a$ in the Dirichlet boundary condition case, we modify the ideas in~\cite[Section 4.4]{FrKo23}.
With Neumann boundary conditions on both temporal and spatial boundaries of the domain $[0,\tau]\times M$, we calculated that the leading temporal eigenvalue was $\lambda^{\rm temp}_1=-a^2\pi^2/\tau^2$.
Because we continue to use Neumann boundary conditions on the temporal boundaries $\{0\}\times M$ and $\{\tau\}\times M$, we keep this eigenvalue as a heuristic comparator for the leading Dirichlet eigenvalue of $\Delta$ on the domain~$M$.
If $M$ is a rectangle with side lengths $s_1$ and $s_2$, then this leading Dirichlet eigenvalue is~$\pi^2(1/s_1^2+1/s_2^2)$.
Thus, matching these two eigenvalues we obtain~$a=\tau\sqrt{1/s_1^2+1/s_2^2}$.
Because we wish the spatial information to slightly dominate, we treat this as a heuristic \emph{minimum} value for $a$, namely~$\smash{ a_{\rm min}=\tau\sqrt{1/s_1^2+1/s_2^2} }$.
More generally, if $M$ were not a rectangle, we would have $a_{\rm min}=\tau\sqrt{-\lambda_1(M)} / \pi$, where $\lambda_1(M)$ is the leading eigenvalue of $\Delta$ on $M$ with Dirichlet boundary conditions.

\subsection{Distinguishing coherent
flow regimes from mixing regimes}
\label{ssec:distinguish_coh_mix}

In \cite[Section 4.3]{FrKo23} we argued that the $L^2$ time-fibre norms are key indicators of the presence of coherence vs complex phase space mixing.
While that analysis was for Neumann boundary conditions, the same arguments hold for Dirichlet boundary conditions.
If $F:\mathbb{M}_0\to \mathbb{R}$ is an eigenfunction of the inflated dynamic Laplacian with Dirichlet boundary conditions, we therefore also expect low values of $\|F(t,\cdot)\|_{L^2(M)}$ in time intervals of strong mixing.
Exactly how low will depend on the time-diffusion parameter $a$;  smaller $a$ will allow more variation of $\|F(t,\cdot)\|_{L^2(M)}$ in time $t$.
Figure~\ref{fig:seba-L2-norms} illustrates the behaviour of the $L^2$ time-fibre norms for the leading two eigenvectors and corresponding leading two SEBA vectors~\cite{FRS19} of the inflated dynamic Laplacian for the polar vortex experiment. 
\begin{figure}[htb]
    \centering
\includegraphics[width=\textwidth]{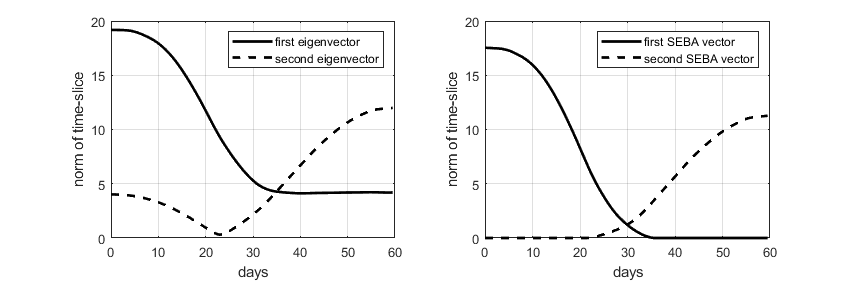}
    \caption{Left: $\ell_2$ norm of discrete eigenvectors estimating $F_1(t,\cdot)$ and $F_2(t,\cdot)$ vs $t$.  These slices are visualised in Figures~\ref{fig:PVevec1} and~\ref{fig:PVevec2} respectively.  One sees that $F_1$ captures the vortex presence prior to breakup, while $F_2$ is supported both before and after the breakup (but with opposite sign, see Figure~\ref{fig:PVevec2}, leading to the ``bump'' in norm around day 22). 
    Right: $\ell_2$ norm of SEBA vectors -- formed from the leading two eigenvectors $F_1, F_2$ of $\Delta_{G_{0,a}}$ -- vs time. The corresponding slices of these leading two SEBA vectors are superimposed in Figure~\ref{fig:PVinfl}. One clearly sees the first SEBA vector has initial strong support, when the initial vortex is present, but then decays as the vortex breaks up.  The second SEBA vector represents the reformation of part of the initial vortex, but some portion of the initial vortex is filamented and lost (see Figure~\ref{fig:PVav}) and so the norm of the dashed line does not reach the same maximum as the solid line.}
    \label{fig:seba-L2-norms}
\end{figure}
In this experiment there is only a brief period of partial breakup around 22 days where the second eigenvector has zero time-fibre norm (see Figure~\ref{fig:seba-L2-norms} (left)), while before this time the second eigenvector has only near-zero time-fibre norms.  
Similarly, the first eigenvector post-breakup also has only near-zero time-fibre norms. 
Selecting  two sparse basis vectors for the span of the two leading eigenvectors via the SEBA algorithm cleanly separates the pre- and post-breakup objects;  see Figure~\ref{fig:seba-L2-norms} (right).

\subsection{Selecting the number of eigenvalues}
\label{ssec:gap}

The standard heuristic for selecting the number of eigenvalues to consider in spectral methods is to look for the first gap in the spectrum -- the so-called eigengap heuristic.
For example, if the distance $|\lambda_{K+1}-\lambda_K|$ is much larger than the distances $|\lambda_{k+1}-\lambda_k|$ for $1\le k< K$ then we say there is a gap between $\lambda_{K}$ and $\lambda_{K+1}$.
One may also rescale eigenvalues to take into account the asymptotic spread of eigenvalues of Laplace--Beltrami operators due to Weyl's law {or directly use properties of the eigenfunctions;  for both options} see~\cite{FRS19}.
Once the number of eigenvalues $K$ has been fixed, and the spatial eigenfunctions identified as above, the nontrivial spatial eigenfunctions (i.e.\ not including the leading constant eigenfunction) undergo an augmentation procedure, described in the following subsection.

\subsection{Eigenfunction 
augmentation}
\label{ssec:augment}
{When Neumann boundary conditions are used, the inflated dynamic Laplacian spectrum contains temporal eigenvalues with corresponding temporal eigenfunctions, as discussed above.
The existence of these temporal eigenfunctions artificially restricts the geometry of the spatial eigenfunctions by constraining \emph{all} distinct eigenfunctions (including temporal) to be orthogonal.
To reinject the necessary degrees of freedom to capture all important spatial variability, we augment the set of eigenfunctions as follows.
Having selected a suitable value for the number of eigenvalues $K$ as above, let $\lambda^{\spat}_2,\ldots,\lambda^{\spat}_J$ denote the leading $J\le K$ \revision{spatial} eigenvalues \revision{contained in the collection $\{\lambda_2,\ldots,\lambda_K\}$}, \emph{not including} the leading eigenvalue\footnote{Recall we are discussing the case of Neumann boundary conditions.} $\lambda^{\spat}_1=0$.
For each spatial eigenfunction $F^\spat_j$, $j=2,\ldots,J$, we define a companion function $\hat{F}_j(t,x):=\|F_j^\spat(t,\cdot)\|_{L^2(M)}\mathbf{1}_M(x)$, which is constant in space on each time fibre, taking the constant value given by the $L^2$ norm of the time fibre. Further rationale and a detailed analysis is given in 
Section~\ref{sec:efunstruc}.}

\subsection{Extracting individual semi-material finite-time coherent sets}

{Algorithm \ref{alg:workflow} sketches the workflow for extracting individual semi-material finite-time coherent sets.}

{
\begin{algorithm}[htb]
    \caption{Workflow to obtain individual semi-material finite-time coherent sets}
    \label{alg:workflow}
    \begin{algorithmic}[1]
        \State \textsc{Input}: Trajectory data $\mathbb{D}_1$, temporal diffusion strength $a$, number $K$ of dominant eigenvalues to be computed.
        \State If Neumann boundary conditions:
        \begin{enumerate}[(a)]
            \item Compute eigenfunctions $F_k$, $k=1,\ldots,K$, using Algorithm~\ref{algo:Strang}.
            \item Identify spatial eigenfunctions $F^\spat_1,\ldots,F^\spat_J$ for $J\le K$;  see subsection~\ref{ssec:spatial}.
            \item Create an augmented collection of $2(J-1)$ functions $F_2^\spat,\ldots,F_J^\spat,\hat{F}_2,\ldots,\hat{F}_J$ as described in Section~\ref{ssec:augment}.
            \item Apply the SEBA algorithm \cite[Algorithm 3.1]{FRS19} to this collection, to give $2(J-1)$ SEBA vectors.  See also code in \cite[Appendix A.1]{FRS19}.
            \item Remove ``residual'' SEBA vectors (those that are constant on each time slice) to leave SEBA vectors supported on individual finite-time coherent sets.  See Section~\ref{sssec:Neumann_augment} and Figure~\ref{fig:cmcmvars}.
        \end{enumerate}
         \State If Dirichlet boundary conditions:
         \begin{enumerate}[(a)]
             \item Compute eigenfunctions $F_k$, $k=1,\ldots,K$, using Algorithm~\ref{algo:Dirichlet}.
            \item Apply the SEBA algorithm \cite[Algorithm 3.1]{FRS19} to this collection, to give $K$ SEBA vectors.
         \end{enumerate}
        \State \textsc{Output}: SEBA vectors supported on semi-material finite-time coherent sets.
    \end{algorithmic}
\end{algorithm}
}

\section{The diffusion-map based approach for the inflated Dynamic Laplacian}
\label{sec:dmapsDL}

In \cite[Section~6]{FrKo23} a discretisation scheme of the inflated dynamic Laplacian was devised, based on a finite element method introduced in~\cite{FJ18}.
The FEM approach has many advantages: preserving the symmetric structure of the dynamic Laplacian; ability to handle general domains; simple treatment of boundary conditions; automatic handling of data that is nonuniformly sampled, sparse, or missing;  no free parameters such as neighbourhood radius or cutoff values;  and estimates provided at all points in the domain.
Because the dominant eigenvalues of $\Delta_{G_{0,a}}$ are those that are smallest in absolute value, in the FEM approach they are computed with an inverse iteration method. 
Another advantage of the FEM approach is that the matrices involved are sparse, but one downside of directly approximating a Laplacian is that for very large (sparse) matrices, computing the smallest eigenvalues requires one to solve a large system of linear equations at each inverse iteration step.

In the current work we detail and explore a diffusion-map based method to  approximate $\mathcal{P}(\ep) = e^{\ep \Delta}$, instead of approximating $\Delta$ directly as with FEM.
The dominant spectrum of $e^{\ep \Delta}$ is now around the eigenvalue~1, and these spectral values can be obtained by forward iteration, which only requires multiplying vectors with a sparse matrix. 
We will thus use such an approach -- with some important modifications -- to approximate the dominant spectrum of the inflated dynamic Laplacian.

\subsection{Diffusion-map approximation of the  Laplacian}
\label{ssec:dmaps}

As a brief primer for the numerical approximation of the inflated dynamic Laplacian, we recall a standard kernel-based method, the diffusion maps approach, to estimating the standard Laplace operator $\Delta$ on a domain or a submanifold~$M\subset\mathbb{R}^d$.
One first formally exponentiates $\Delta$ to obtain the operator $\mathcal{P}(\ep) = e^{\ep\Delta}$, representing the heat-flow operator on $M$ for auxiliary heat-flow time~$\ep$.
The Laplace operator $\Delta$ is then approximated by $\mathcal{P}(\ep)$ as follows.

{Let} data points $x_1,\ldots,x_N\in M$ {sampled i.i.d.\ from a smooth density $q$ be given.}
One begins the standard diffusion maps construction \cite{CoLa06} by forming the $N\times N$ symmetric exponential kernel matrix $K(\cdot)$ with $(i,j)$ entries $K_{ij}(\ep)=\exp(-\|x_i-x_j\|^2/(4\ep)).$
The choice of $\ep$ is based on the density of the data points.
We select $\ep$ according to the heuristic of \cite{coifman_graph_2008,berry_variable_2016}, who suggest to choose $\ep$ in a linear region of a plot of $\smash{ \log\big(\sum_{i,j}K_{ij}(\ep)\big) }$ vs~$\log \ep$.
This range of $\ep$ is suitable because the quantity $\smash{ \sum_{i,j}K_{ij}(\ep) }$ scales as approximately $\smash{ \ep^{d/2} }$ and importantly, $\ep$ is not so large that $\smash{ \sum_{i,j}K_{ij}(\ep) }$ saturates at value $N^2$, nor so small that $\sum_{i,j}K_{ij}(\ep)$ takes its minimum value of~$N$.

Having fixed a suitable bandwidth parameter $\ep$, the symmetric matrix $K(\ep)$ undergoes two normalisations. 
The first normalisation aims to factor out bias caused by possible non-uniformity of the distribution~{$q$}~\cite[Proposition~10]{CoLa06}.
For each $i$, one defines the row sum $\smash{ d_i(\ep) = \sum_j K_{ij}(\ep) }$, which is a kernel density estimate of~$q(x_i)$ (up to a constant factor), if $x_i$ is order $\sqrt{\epsilon}$ distant from the boundary of the manifold~$M$.
By normalising $K(\epsilon)$ with these density estimates, we obtain a matrix $\smash{ \widetilde{K}_{ij}(\ep) := K_{ij}(\ep) / (d_i(\ep)d_j(\ep))}$, which satisfies\footnote{Let $k_\epsilon(x,y)$ be the kernel that generates $K_{ij}(\ep)=k_\epsilon(x_i,x_j)$. Let $d_\ep(x)=\int k_\ep(x,y)q(y)\ dy$ be a local estimate of $q(x)$ (up to a constant factor) and $\tilde{k}_\ep(x,y)=k_\ep(x,y)/d_\ep(x)d_\ep(y)$ be the normalised kernel that generates $\widetilde{K}_{ij}(\epsilon)$. The functional analogue of $\sum_j \widetilde{K}_{ij}(\ep)$ is $\int \tilde{k}_\ep(x,y)q(y)\ dy$.
By symmetry of $\tilde{k}_\ep$, without loss we discuss the case with integration over $y$.
Note that
$$\int \tilde{k}_\ep(x,y)q(y)\ dy =\int \underbrace{\left(\frac{k_\ep(x,y)q(y)}{\int k_\ep(x,y)q(y)\ dy}\right)}_{=:\rho_{\ep,x}(y)}\cdot\frac{1}{d_\ep(y)}\ dy.$$
The family of functions $\rho_{\ep,x}$ satisfy $\int \rho_{\ep,x}(y)\ dy=1$ for all $x$. Further, each $\rho_{\ep,x}$ is a 1-dimensional density with support concentrated around $x$. Thus the main integral on the RHS above is a convex combination of $1/d_\ep(y)$ for $y$ close to $x$.  Therefore the approximate value of $\int \tilde{k}_\ep(x,y)q(y)\ dy$ is $1/d_\ep(x)$.}
$\smash{\sum_i \widetilde{K}_{ij}(\ep) \approx 1/d_j(\ep)}$ and $\smash{\sum_j \widetilde{K}_{ij}(\ep) \approx 1/d_i(\ep)}$, except at rows and columns $i$ where $x_i$ is within order $\sqrt\epsilon$ of the boundary of~$M$.
The second normalisation \revision{row}-normalises 
$\smash{\widetilde{K}(\epsilon)}$ to make a column stochastic matrix $\smash{ P_{ij}(\ep) := \widetilde{K}_{ij}(\ep) / \sum_j \widetilde{K}_{ij}(\ep)}$.
One may view this second normalisation as imposing reflecting boundary conditions on the random walk generated by the column stochastic (Markov) matrix~$P(\ep)$. More on this second normalisation is discussed in~\cite{vaughn2019diffusion}.
Using the approximate row and column sums for $\widetilde{K}(\epsilon)$ mentioned above we note that
\[
P_{ij}(\ep)=\tilde{K}_{ij}(\ep) \Big/ \sum_j \tilde{K}_{ij}(\ep) =  K_{ij}(\ep) \Big/ \bigg(d_i(\ep)d_j(\ep)\sum_j \tilde{K}_{ij}(\ep)\bigg) \approx K_{ij}(\ep)/d_j(\ep),
\]
and therefore $P(\ep)$ is approximately the original kernel matrix $K(\ep)$ subjected to a \revision{column} normalisation, except for those rows and columns $i$ where $x_i$ is within order $\sqrt\epsilon$ of the boundary of~$M$.

To set up the statement below we formalise the i.i.d.\ point process.
We independently draw $N$ points in $M$ according to a common distribution~$q$.
We model the $N\to\infty$ limiting process with a sample space $\Omega$ that can be identified with the semi-infinite product $M^\mathbb{N}$ endowed with the product probability measure $\mathfrak{q}^\mathbb{N}$, where $\mathfrak{q}$ is the measure with (Lebesgue) density~$q$.
For a function $f\in C^3(M)$ define the operation $P(\ep)_N f := \smash{ ( \sum_j P_{ij}(\ep) f(x_j) )_{i=1}^N } \in \R^N$.
Each element $\omega=(x_1,x_2,\ldots)\in\Omega$ generates a particular sequence of approximation operators $\{P(\epsilon)_{N,\omega}\}_{N=1}^\infty$, formed from the first $N$ sample points in the sequence~$\omega$.
For $q\in C^3(M)$ it was shown~\cite[Theorem~3]{hein2005graphs} that for each $x$ in the interior of~$M$
\[
\lim_{\ep\to 0}\lim_{N\to\infty} \frac{1}{\ep} (P(\ep)_{N,\omega} f - f)(x) = \Delta f(x)\qquad\mbox{for $\mathfrak{q}^\mathbb{N}$-almost-every $\omega\in\Omega$.}
\]
Thus, $L(\ep) := \frac{1}{\ep}(P(\ep)-\mathrm{Id})$ is the data-based diffusion map approximation of~$\Delta$. 
We note that $\ep$ can be interpreted as an evolution time of a heat flow (a diffusion) on the set of data points.

We wish to select the bandwidth $\ep$ so that the entries in $K(\ep)$ are negligibly small unless $x_i$ are $x_j$ are nearby.
This will mean that  $P(\ep)$ is close to a sparse matrix. 
In practice, a cutoff parameter $r$ is applied to the distance computations $\|x_i - x_j\|$, such that point pairs beyond that distance constitute a zero entry in~$K(\ep)$. We take $r = 2\sqrt{5\ep}$, which by the definition of $K(\epsilon)$ implies that values of $K_{ij}(\epsilon)$ below $0.007$ are set to zero. For the data resolutions we use in our experiments, this choice  produces a sparse $K(\epsilon)$ and therefore a sparse $P(\epsilon)$, which is numerically advantageous when handling very large numbers of data points.

Recall that the construction described so far corresponds to the Laplace operator with homogeneous Neumann boundary conditions, as in
e.g.\ \cite{CoLa06}. The case of homogeneous Dirichlet boundary conditions will be discussed in Section~\ref{ssec:dmapsDirichlet} below, after first discussing the extension to the spacetime domain~$\MM_0$.

\subsection{Application in spacetime}
\label{ssec:iDLdmaps}

We assume that we have $N$ distinct initial conditions $\{x_1,\ldots,x_N\} \subset M$, which generate trajectories of the flow that are sampled at $T$ time instances $0\le t_1 < t_2 < \ldots < t_T \le \tau$, leading to $NT$ spacetime data points
\[
{\mathbb{D}_1} := \left\{ \left(t_i,\phi_{t_i} x_j\right) \mid i=1,\ldots,T,\ j=1,\ldots,N \right\} \subset \MM_1.
\]
Since $\mathbb{D}_1$ acts as a discrete approximation of $\MM_1$, for the sake of exposition we define $\mathbb{D}_0 := \left\{ \left(t_i, x_j\right) \mid i=1,\ldots,T,\ j=1,\ldots,N \right\} \subset \MM_0$ as a discrete approximation of~$\MM_0$.

There are four challenges that we need to address:
\begin{enumerate}
    \item 
It is not yet clear whether we can use diffusion maps to approximate the inflated dynamic Laplacian on ${\mathbb{D}_0}$, because this operator contains the spatial Laplace--Beltrami operators $\Delta_{g_t}$, which are defined with the pullback metric~$g_t$. 
This issue will be solved in a way similar to the FEM approach in~\cite{FrKo23}.
\item By imposing Neumann boundary conditions in space, one has purely temporal modes that are analytically known~\cite{FrKo23}.
Our numerical scheme should ideally preserve the natural decomposition of $L^2(\mathbb{M}_0)$  into temporal and spatial modes, as was the case for the FEM approach \cite{FrKo23}.
\item It is desirable for a numerical scheme to allow the temporal diffusion coefficient $a$ to be varied without having to recompute the entire discrete operator, as in the FEM approach of~\cite{FrKo23}. 
\item If we could apply diffusion maps directly to approximate the inflated dynamic Laplace operator, there are no canonical joint coordinates for space and time, and thus spacetime data points in {$\mathbb{D}_0$} may have very different extents and sampling resolutions.
In contrast to the FEM approach, with diffusion maps we need a bandwidth parameter $\ep$ and it is unlikely that a single bandwidth parameter $\ep$ will be sufficient to accurately capture the different sampling geometries of space and time. One could attempt to impose an anisotropic diffusion on the spacetime data, but this introduces yet another parameter, which we would like to avoid.
\end{enumerate}

To begin to address all of these challenges (and particularly challenges 2 and 4) we note that we can additively decompose the inflated dynamic Laplacian as $\smash{ \Delta_{G_{0,a}} = a^2 \Delta^{(t)} + \Delta^{(x)} }$, where
\[
\Delta^{(t)} = \partial_{tt} \qquad \text{and} \qquad \left( \Delta^{(x)}F \right)(t,\cdot) = \Delta_{g_t}F(t,\cdot).
\]
We will approximate the operators $\cP^{(t)}(\ep) := \exp\big(\epsilon a^2 \Delta^{(t)} \big)$ and $\smash{\cP^{(x)}(\ep) := \exp\big(\epsilon \Delta^{(x)} \big)}$ separately, and then use the Strang splitting~\cite{Str68} approximation, formally
\[
e^{\epsilon(L_1+L_2)} = e^{\epsilon L_1/2}e^{\epsilon L_2}e^{\epsilon L_1/2} + \mathcal{O}(\epsilon^3),
\]
to approximate~$\smash{ \cP_a(\ep) := \exp\big( \ep ( a^2 \Delta^{(t)} + \Delta^{(x)} ) \big) }$.
The resulting matrix, $P_a(\ep)$, to be defined in~\eqref{eq:Strang_Peps} below, will be a $NT \times NT$ matrix approximating the operator $\cP(\ep)$ pointwise in the data points upon multiplication by the vector of function evaluations in the data points from the right:
\[
\left(\cP_a(\ep)F\right)\!\vert_{\mathbb{D}_0} \approx P_a(\ep) (F\vert_{\mathbb{D}_0}).
\]
We write
\begin{equation}
    \label{gfvec}
F\vert_{\mathbb{D}_0} =: f=:\begin{pmatrix}
   f_1\\
   \vdots\\
   f_T
\end{pmatrix} \in \R^{NT}, \mbox{ where }f_i:=\begin{pmatrix}
    F(t_i,x_1)\\
    \vdots \\
    F(t_i,x_N)
\end{pmatrix}\in\mathbb{R}^N.
\end{equation}
Regarding challenge 3 above, we note that for large values of $a$, a diffusion maps approach to estimate $\exp(\epsilon a^2\Delta^{(t)})\approx\cP^{(t)}_{\epsilon}$ will be poor because $\epsilon a^2$ may be too large relative to the dataset.
As the time interval is one-dimensional, and the number of time instances $T$ is usually not exceeding $\sim 100$, we can compute the exponential of a time-Laplacian directly. 
Thus, we compute a finite-difference approximation $L^{(t)} \in \R^{T\times T}$ of the standard Laplace operator on the interval $[0,\tau]$ (with Neumann boundary conditions), and then compute the exponential $\smash{ e^{\epsilon a^2 L^{(t)}/2} }$ directly. 
To be more precise, if the temporal grid has stepsizes $h_i:=t_{i+1}-t_i$, then using the standard second-order central-difference approximation we have for~$g\in\R^T$,
\begin{equation}
    \label{eq:FEM_time}
    \renewcommand{\arraystretch}{1.5}
    \left( L^{(t)} g \right)_i = \left\{
    \begin{array}{ll}
    \frac{1}{h_1^2}( -g_1 + g_2 ), & i=1,\\
    \frac{2}{h_{i-1} + h_i} \left( \frac{g_{i+1} - g_i}{h_i}  - \frac{g_{i} - g_{i-1}}{h_{i-1}} \right), & i=2,\ldots,T-1,\\
    \frac{1}{h_{T-1}^2}( -g_T + g_{T-1} ), & i=T.
    \end{array}
    \right.
\end{equation}
By the enumeration of data points we chose above, the corresponding matrix in spacetime is the Kronecker product~$\smash{P^{(t)}(\ep) := e^{\epsilon a^2 L^{(t)}} \otimes \mathrm{Id}_N }$. 

The spatial diffusion map matrix $P^{(x)}(\ep)$ is a block-diagonal matrix with $T$ blocks of size~$N\times N$.
When constructing the $i^{\rm th}$ block $\smash{P_{t_i}^{(x)}(\epsilon)}$, to avoid the explicit computation of the distances between $x_j$ and $x_{j'}$ with respect to the pullback metric $g_{t_i}$, as mentioned in challenge 1, we instead compute Euclidean distances between $\phi_{t_i}x_j$ and $\phi_{t_i}x_{j'}$; a technique previously used in \cite{FJ15,BERRY2016439,BaKo17,FJ18,FrKo23} for various numerical schemes.
In summary, for the $i^{\mathrm{th}}$ timeslice $\{ \phi_{t_i}x_j \mid j=1,\ldots,N\}$ of the data one computes the associated diffusion map matrix $\smash{ P^{(x)}_{t_i}(\epsilon) }$ as described in Section~\ref{ssec:dmaps}, and that is known from~\cite{BaKo17} to approximate~$\exp\left(\ep \Delta_{g_{t_i}} \right)$ if applied on point evaluations of a function in the points~$\{x_1,\ldots,x_N\}$ at initial time:
\[
\left[\Big( \exp\big(\ep \Delta_{g_{t_i}} \big) F(t_i,\cdot) \Big) (x_j) \right]_{j=1,\ldots,N} \approx P^{(x)}_{t_i}(\epsilon) f_i.
\]

Employing the Strang splitting, the discrete approximation of $\cP_a(\ep)$ we work with is given by
\begin{equation}
\label{eq:Strang_Peps}
P_a(\ep) =
\underbrace{\left( e^{\epsilon a^2 L^{(t)}/2} \otimes \mathrm{Id}_N \right)}_{ = P^{(t)}(\epsilon/2)}
\underbrace{
\begin{pmatrix}
P^{(x)}_{t_1}(\epsilon) & & \\
& \ddots & \\
& & P^{(x)}_{t_T}(\epsilon)
\end{pmatrix}
}_{ = P^{(x)}(\epsilon)}
\underbrace{\left( e^{\epsilon a^2 L^{(t)}/2} \otimes \mathrm{Id}_N \right)}_{ = P^{(t)}(\epsilon/2)}.
\end{equation}

The large $NT \times NT$ matrix $ e^{\epsilon a L^{(t)}/2} \otimes \mathrm{Id}_N $ has up to $NT^2$ non-zero entries, but with a lot of redundancy: effectively a $T\times T$ matrix is copied $N$ times. One can thus save memory by not assembling this matrix explicitly, only evaluating matrix-vector products involving it on the fly. 
Let $\mathbf{F}_{ij}=F(t_i,x_j)\in\mathbb{R}^{T\times N}$. 
Referring to (\ref{gfvec}) for the definition of $f$, then if 
\begin{equation}
    \label{eq:memorytrick}
    \begin{pmatrix}
     \tilde{\mathbf{F}}_1^\top \\
       \vdots \\
       \tilde{\mathbf{F}}_T^\top
   \end{pmatrix}
   := 
   \Big(e^{\epsilon a^2 L^{(t)}/2} \Big)\mathbf{F}\quad \text{ we have } \quad\left( e^{\epsilon a^2 L^{(t)}/2} \otimes \mathrm{Id}_N \right) f = \begin{pmatrix}
    \tilde{\mathbf{F}}_1\\ \vdots \\ \tilde{\mathbf{F}}_T
    \end{pmatrix}.
\end{equation}
This way, at most $T^2$ entries of $ e^{\epsilon a L^{(t)}/2} $ need to be stored in memory, and the computation in \eqref{eq:memorytrick} requires $\cO(NT^2)$ operations.
{For comparison, }note that by judicious choice of the bandwidth $\ep$ the blocks in $P^{(x)}(\ep)$ should all be sparse, hence $P^{(x)}(\ep)$ contains~$\mathcal{O}(NT)$ nonzero entries.
We summarise our construction in Algorithm~\ref{algo:Strang}.
\begin{algorithm}[htb]
    \caption{Splitting method to approximate the dominant spectrum of the inflated dynamic Laplacian}
    \label{algo:Strang}
    \begin{algorithmic}[1]
        \State \textsc{Input}: Trajectory data $\mathbb{D}_1$, temporal diffusion strength $a$, number $K$ of dominant eigenvalues to be computed.
        \State Choose a bandwidth parameter $\ep$ or infer it as in Section~\ref{ssec:dmaps}.
        \For{$i=1,\ldots,T$}
        \State Compute $P^{(x)}_{t_i}(\ep)$ from the data $\{ \phi_{t_i}x_j \mid j=1,\ldots,N\}$, as described in Section~\ref{ssec:dmaps}
        \EndFor
        \State Assemble $P^{(x)}(\ep)$ from the diagonal blocks $P^{(x)}_{t_i}(\ep)$.
        \State Compute $L^{(t)}$ as in \eqref{eq:FEM_time}.
        \State Do one of the following:
        \begin{enumerate}[(a)]
            \item Compute and store $\smash{ P^{(t)}(\ep/2) = e^{\epsilon a^2 L^{(t)}/2} \otimes \mathrm{Id}_N }$.
            \item Store $\smash{ e^{\ep a^2 L^{(t)}/2} }$ and use \eqref{eq:memorytrick} to compute matrix-vector products involving~$P^{(t)}(\ep/2)$ in Step~9.
        \end{enumerate}
        \State Perform power iteration for $P^{(t)}(\ep/2) P^{(x)} P^{(t)}(\ep/2)$ to compute its dominant $K$ eigenpairs~$(\nu_k,F_k)$.
        \State \textsc{Output}: $\left(\frac{1}{\ep}\log \nu_k, F_k \right)$, $k=1,\ldots,K$, as approximate eigenpairs of $\Delta_{G_{0,a}}$.
    \end{algorithmic}
\end{algorithm}

\subsection{Dirichlet boundary conditions}
\label{ssec:dmapsDirichlet}

For the reasons detailed in Section~\ref{ssec:dirichlet} we also consider the inflated dynamic Laplacian with homogeneous Dirichlet boundary conditions in space. This requires a modification of the discrete operator~$P^{(x)}(\ep)$ in~\eqref{eq:Strang_Peps}, because the construction detailed in sections~\ref{ssec:dmaps} and~\ref{ssec:iDLdmaps} yields an approximation of an operator with homogeneous Neumann boundary conditions in space.
Our construction here is reminiscent to the one proposed in~\cite{thiede2019galerkin}.

The main task is to assign boundary points in $\mathbb{D}_0$.
To do this, we first assign boundary points on the $i^{\rm th}$ co-evolved time slice $\{ \phi_{t_i}x_j \mid j=1,\ldots,N\}$ for each $i=1,\ldots,T$.
\revision{In our implementation, w}e use the Matlab \texttt{boundary} command that in turn relies on the \texttt{alphaShape} routine with default parameters.
If we identify $\phi_{t_i}x_j$ as a boundary point in the $i^{\rm th}$ co-evolved time slice, we define the spacetime point $(t_i,x_j)$ as a spatial boundary point for $\mathbb{D}_0$. We denote this boundary choice by~B1; it is motivated by $\phi_t$ being a bijection between $\partial M$ and~$\partial (\phi_t M)$.

The set of boundary points collected according to B1 can be further enlarged by defining the collection $\{(t_l,x_j): l=1,\ldots,T\}$ as  spatial boundary points  if $(t_l,x_j)$ is identified as a boundary using the \texttt{boundary} command for \emph{at least one} time~$t_l$. We denote this boundary choice by~B2. If the spatial approximation of $M$ by $\{x_1,\ldots,x_N\}$ is sufficient to accurately capture the distortion caused by $(\phi_t)_{t\in [0,\tau]}$ then the two procedures B1 and B2 should yield very similar answers. However, if this is not the case (as for the Polar Vortex example in Section~\ref{ssec:examplePV}, which is relatively sparsely sampled relative to the nonlinear dynamics), the second choice B2 performs better, as it also captures  points that move close to the boundary at \textit{some} time~$t_i$. We will thus exclusively use choice~B2 in the sequel.

We construct a vector $(b_{ij})\in \{0,1\}^{NT}$ that encodes the boundary choice B2.
Recalling that we index over all spatial points in timeslice 1 first, then all spatial points in timeslice 2, and so on, we set the entry of $b$ associated to the point $(t_i,x_j)$ to be 0 if $x_j$ is a boundary point for \textit{any} time slice, and to be 1 otherwise. 
We set~$B:= \mathrm{diag}(b) \in \R^{NT\times NT}$.
Finally, we replace in {line 9 of} Algorithm~\ref{algo:Strang} the matrix $P^{(t)}(\ep/2) P^{(x)} P^{(t)}(\ep/2)$ by $B P^{(t)}(\ep/2) B P^{(x)} B P^{(t)}(\ep/2) B$.
The procedure to set up 
this new matrix is summarised in Algorithm~\ref{algo:Dirichlet}.
\begin{algorithm}[htb]
    \caption{Dirichlet boundary conditions in space (choice B2)}
    \label{algo:Dirichlet}
    \begin{algorithmic}[1]
        \For{$i=1,\ldots,T$}
        \State $J_i$ = boundary indices of $\{ \phi_{t_i}x_j \mid j=1,\ldots,N\}$
        \Comment computed with \texttt{boundary}
        \EndFor
        \State Construct {the} diagonal matrix $B\in \{0,1\}^{NT \times NT}$ by
        \[
        B_{k,k} = \left\{
        \begin{array}{ll}
            0, & k = (i-1)N + j \text{ for some } i=1,\ldots,T \text{ and } x_j\in J_{l} \text{ for some } l=1,\ldots,T,\\
            1, & \text{otherwise.}
        \end{array}
        \right.
        \]
        \State Perform power iteration with $B P^{(t)}(\ep/2) B P^{(x)} B P^{(t)}(\ep/2) B$ instead of $P^{(t)}(\ep/2) P^{(x)} P^{(t)}(\ep/2)$ in {line 9 of} Algorithm~\ref{algo:Strang}
    \end{algorithmic}
\end{algorithm}

Addressing the convergence properties of the discretisation~\eqref{eq:Strang_Peps} is beyond the scope of the present study. 
We note that (spectral) convergence studies of diffusion maps on manifolds without boundaries or for Neumann boundary conditions are much better-developed than studies involving Dirichlet boundary conditions; see the introduction in \cite{peoples2021spectral} for an overview. Nevertheless, pointwise (i.e., applied to a fixed function) almost sure $L^2$-convergence of diffusion maps for domains with boundaries has been considered in~\cite{vaughn2019diffusion}. Note that the construction from Section~\ref{ssec:dmaps}, also with the modifications introduced in this section, yields a nonsymmetric $P^{(x)}(\ep)$ in general, thus $P_a(\ep)$ in \eqref{eq:Strang_Peps} is also nonsymmetric despite symmetry of~$P^{(t)}(\ep/2)$.

\section{Numerical examples}
\label{sec:numericalexamples}

\subsection{Semi-material FTCS and modest regime change:  Case study -- a switching double gyre}
\label{sec:switching}

We first illustrate how the temporal diffusion parameter $a$ allows coherent sets to evolve more smoothly by relaxing the materiality constraint.
We use a switching double-gyre on a domain $M=[0,3]\times [0,2]$, where for the first half of the time interval $[0,\tau] = [0,1]$ there are two invariant sets ($[0,1]\times [0,2]$ and $[1,3]\times [0,2]$), which then rapidly switch in the second half of the time interval to invariant sets ($[0,1]\times [0,2]$ and $[2,3]\times [0,2]$).

More explicitly, {inspired by the flow in~\cite[Sec.~6]{shadden-etal},} we consider the stream function $\psi(t,x) = -20\sin(\frac{\pi}{2}x_2) \sin(\pi f(t,x_1))$, 
$x\in\mathbb{R}^2$, giving rise to the ODE $\dot x (t) = (\partial_{x_2}\psi(t,x), \linebreak[3] -\partial_{x_1}\psi(t,x))$. In the notation of \cite{shadden-etal, FP09} we set the speed parameter for the flow to be \revision{$A=20$, which is 200 times that in \cite{shadden-etal} and 80 times that in~\cite{FP09}.} 
At all times, the oscillation frequency parameter $\omega=0$, but there is a one-way drift in the separatrix over the time interval $[0,1]$, which begins slowly, becomes rapid in the interval $[0.4,0.6]$ and then having made the switch in gyre size, settles down again.

We wish to find a simple function $f$ yielding the behaviour described above. The graphs of the functions $f_0(x_1)=-(1/6)x_1^2+(7/6)x_1$ and $f_1(x_1)=(1/6)x_1^2+(1/6)x_1$ are monotonic in $x_1$ and pass through $\{(0,0),(1,1), (3,2)\}$ and $\{(0,0),(2,1), (3,2)\}$, respectively, yielding the pairs of almost-invariant sets mentioned above; shown in Figure~\ref{fig:DGswitch_contour}.
\begin{figure}[htb]
    \centering
    \includegraphics[width=0.49\textwidth]{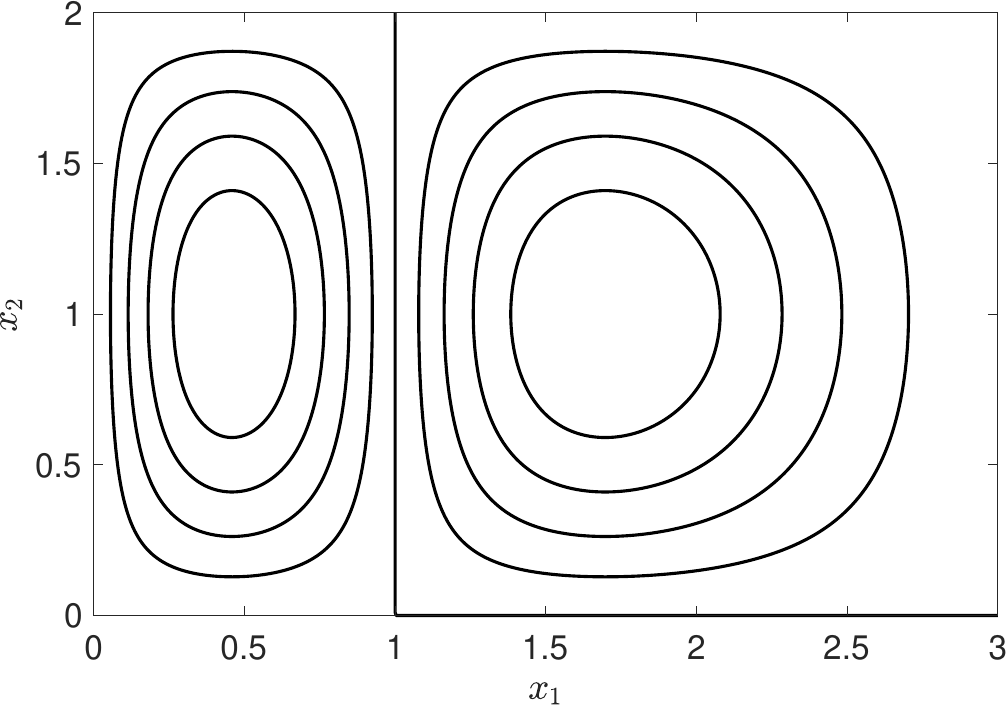}\hfill
    \includegraphics[width=0.49\textwidth]{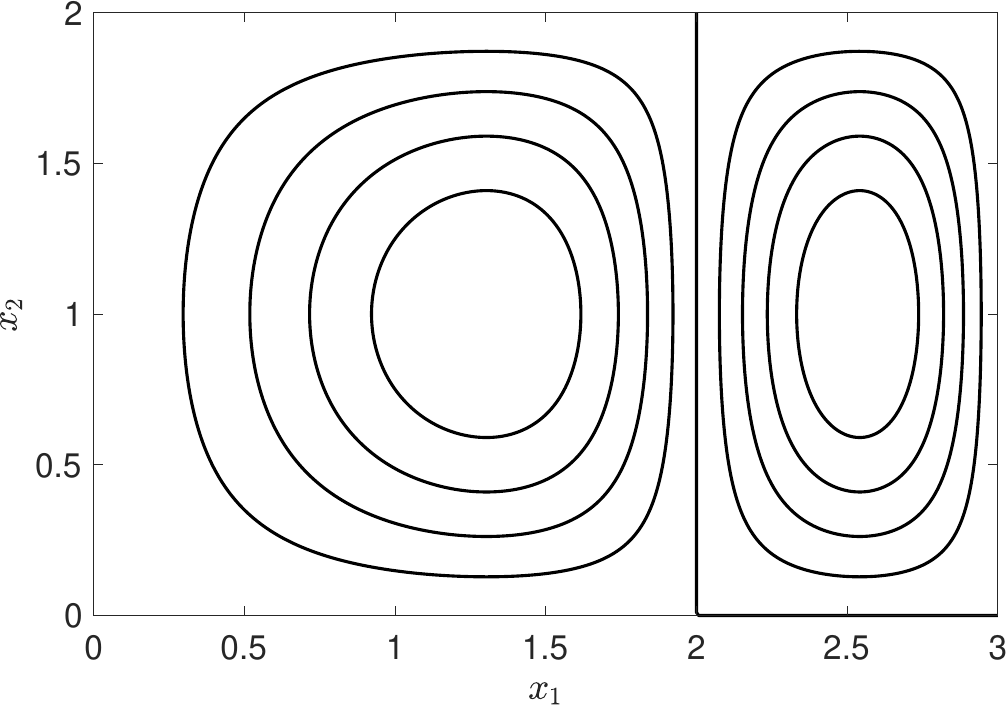}
    \caption{Contour plots of the stream function at starting time $t=0$ (left) and end time $t=1$ (right).}
    \label{fig:DGswitch_contour}
\end{figure}
To build in the switch near $t=1/2$ we interpolate between $f_0$ and $f_1$ by using the smoothed step function~$r(t) = \frac12 (1 + \tanh (40(t-\frac12)) )$ as a  time-dependent interpolation parameter.
One possible interpolation is achieved by setting
\[
\alpha(t) = \frac{(1-2r(t))}{3(r(t)-2)(r(t)+1)}, \quad \beta(t) = \frac{2 - 9\alpha(t)}{3}, \quad f(t,x_1) = \alpha(t) x_1(t)^2 + \beta(t) x_1(t),
\]
leading to an evolution of stream functions that is symmetric upon reflection in the $(t,x_1)$ coordinate plane with respect to~$(t,x_1) = (1/2, 3/2)$.
The evolution of $x(t)$ is then governed by the ODE
\begin{equation}
    \label{eq:switchDG}
    \begin{pmatrix}
        \dot{x}_1(t) \\ \dot{x}_2(t)
    \end{pmatrix}
    =
    20 \begin{pmatrix}
        -\frac{\pi}{2}  \sin(\pi f(t,x_1))\cos (\pi x_2(t)/2)\\
         \pi(2 \alpha(t) x_1(t) + \beta(t))\cos(\pi f(t,x_1))\sin(\pi x_2(t)/2)
    \end{pmatrix}.
\end{equation}

We set the parameter $a$ to be $3a_{\rm min}$ using the heuristic described in~\cite[Section 4.1]{FrKo23} {and briefly repeated in subsection~\ref{subsec:Dirichlet_choose_a}}, which approximately matches the leading nontrivial temporal and spatial eigenvalues of $\Delta_{G_{0,a}}$;  numerically this value is~$a=1$.

We construct a numerical approximation of $\Delta_{G_{0,a}}$ using Algorithm~\ref{algo:Strang}.
In this example, we use $N=1350$ trajectories initialised on a uniform grid of $45\times 30$ particles on $M$, recording the particle positions each $1/100^{\rm th}$ unit of time{, yielding~$T=101$}.
The diffusion maps procedure requires a bandwidth $\epsilon>0$ that takes into account the density of the trajectories.
We choose $\epsilon=0.0115$ according to the heuristic of \cite{coifman_graph_2008,berry_variable_2016}; see Section~\ref{ssec:dmaps} for further details. 
{The spectrum of~$\Delta_{G_{0,a}}$ displayed in Figure~\ref{fig:dg-evals}(right) reveals a large gap after the second spatial eigenvalue, and following the heuristic in Section~\ref{ssec:gap}, we choose $J=2$.
This corresponds to searching for two semi-material coherent sets. Strictly speaking, according to Algorithm \ref{alg:workflow} one should augment the single eigenvector $F^\spat_2$ to obtain a companion vector $\hat{F}_2$ and apply SEBA. This would produce two vectors that separate the blue and red objects shown in Figure \ref{fig:average-vs-inflated-dg-efuns}(right). To keep things simple, in this example we work directly with the single (signed) eigenfunction~$F^\spat_2$.}

The two finite-time coherent sets for this flow are shown in blue and red in
Figure \ref{fig:DLexample} (right).
These sets have a somewhat rough boundary, but are fully material.
On the other hand, the two semi-material finite-time coherent sets shown in Figure \ref{fig:average-vs-inflated-dg-efuns} (right) are not precisely material, but are smoother in space and follow the position of the FTCS over timescales shorter than the full time interval $[0,1]$.
The analogous statements to those in Theorem~\ref{thm:eigbounds} but for the Neumann boundary conditions applied in this example, to be found in~\cite[Theorem 7, Part 1]{FrKo23}, state that the spatial eigenvalues of the inflated dynamic Laplacian dominate the associated eigenvalues of the dynamic Laplace operator. This is clearly the case, as shown in Figure~\ref{fig:dg-evals}.
\begin{figure}[htb]
   \centering
   \includegraphics[width=0.49\textwidth]{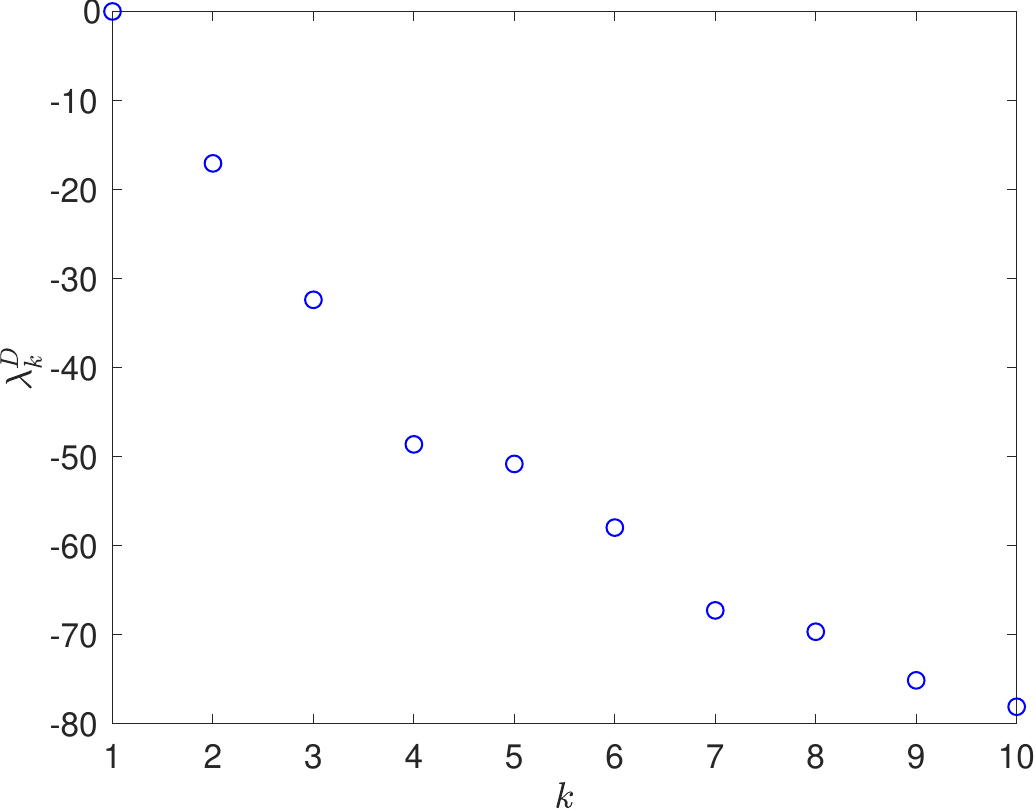}
   \hfill
   \includegraphics[width=0.48\textwidth]{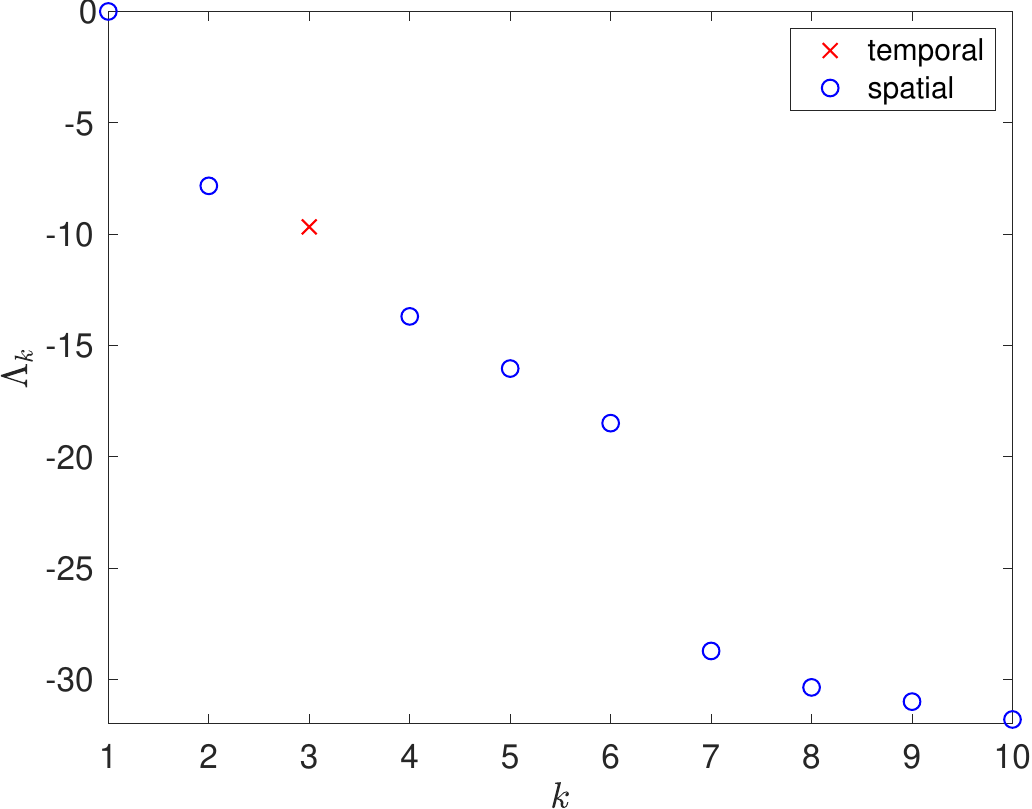}
   \caption{Leading eigenvalues for the standard dynamic Laplacian (left) and the time-inflated dynamic Laplacian (right) for the switching double gyre. Note that $\lambda_k^D\le \Lambda_{k,a}^{\rm spat}$ for $k=1,\ldots,10$, which is consistent with Part 1 of \cite[Theorem 7]{FrKo23}. }
   \label{fig:dg-evals}
\end{figure}

Dominant eigenfunctions of the inflated dynamic Laplacian prefer to have little temporal variation along trajectories, unless variation in the pullback metrics $g_t$ forces them to behave otherwise. Thus, trajectories showcasing a large variation of the spacetime eigenfunction along them are expected to play a central role in the change of dynamical regimes.
{In this example, such trajectories highlight the transfer of fluid  between the two gyres as they change size}.  This is corroborated by Figure~\ref{fig:DGswitch_tempvar}, where we show 
the trajectory of an initial condition $x\in M$ if $t \mapsto F^{\spat}_2(t,x)$ has variance among the top 5\% of variances for all~$x\in M$.
\begin{figure}[htb]
    \centering
    \includegraphics[width = 0.49\textwidth]{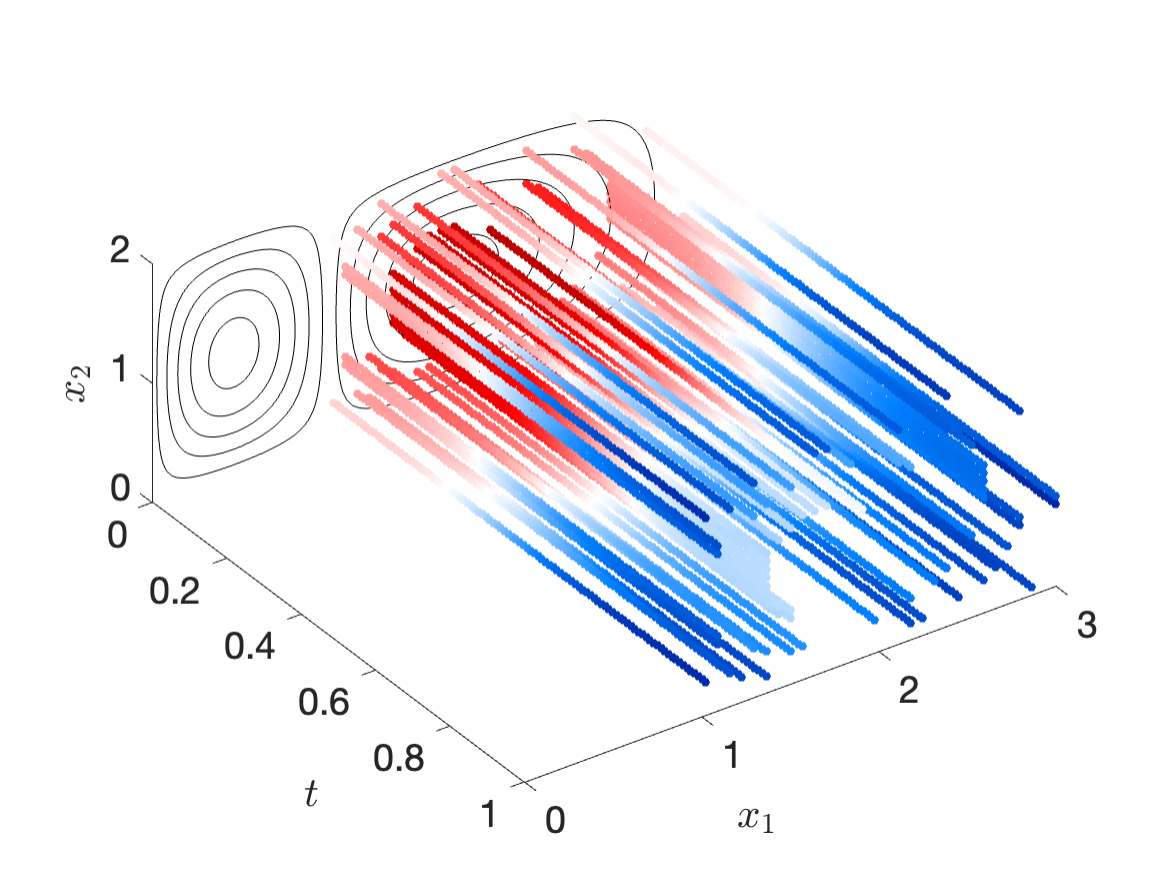}
    \hfill\includegraphics[width = 0.49\textwidth]{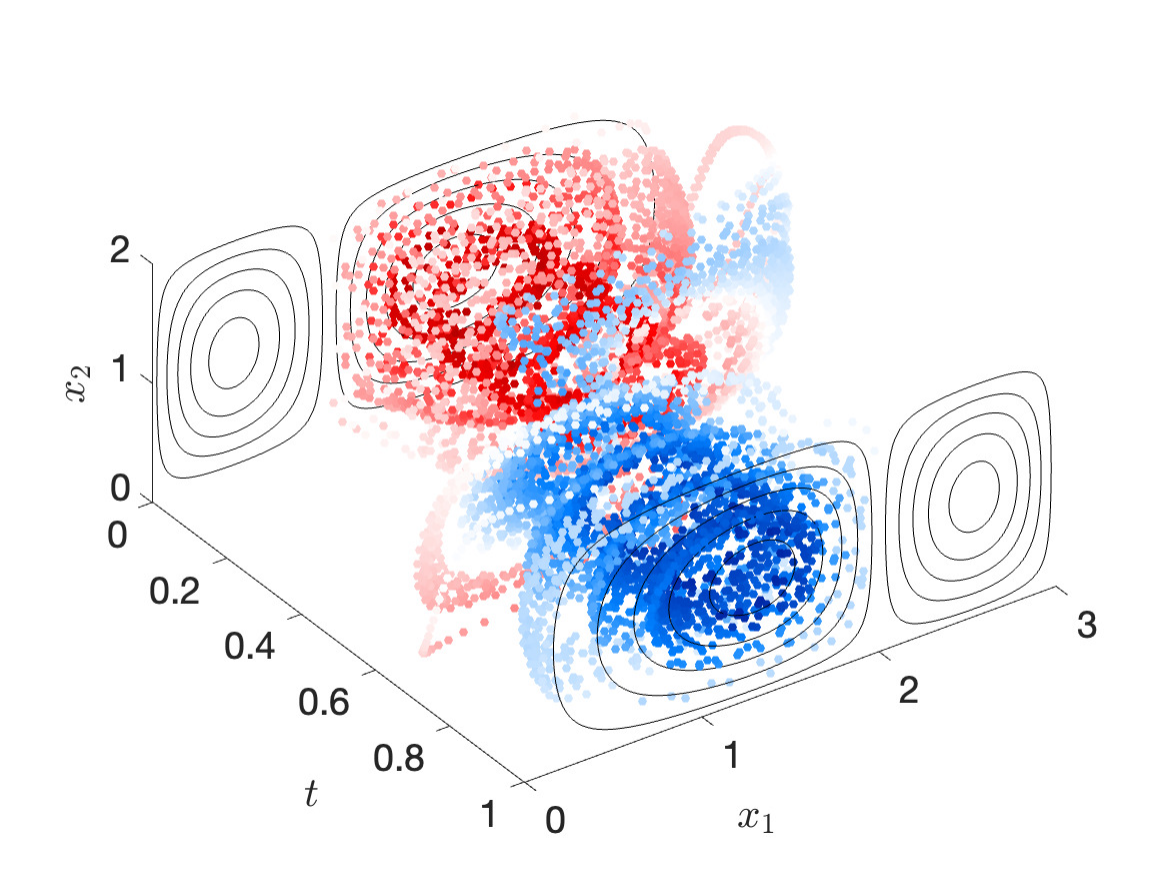}
    \caption{
Trajectories of the switching double gyre system highlighting the top 5\% of largest variance of second spatial eigenfunction of the inflated dynamic Laplacian restricted to them. 
Material trajectories should have constant values of $F$ and we are interested in the ``least material'' trajectories;  thus we plot trajectories that have the highest variance.
Left: $\MM_0$. Right: $\MM_1$. The color indicates the value of~$F^{\spat}_2$ at the associated spacetime point~$(t,x)$. 
Black contours visualize the stream function on the initial and final time slices. 
}
\label{fig:DGswitch_tempvar}
\end{figure}
We note that all these trajectories undergo one sign (colour) change that occurs mostly at times between $0.45$ and $0.55$, roughly when a significant proportion of particles transfer from one gyre to the other.
There is an asymmetry in Figure \ref{fig:DGswitch_tempvar}, whereby 
all the trajectories shown start in the initially larger gyre $[1,3] \times [0,1]$ and end up in the gyre $[0,2] \times [0,1]$ that is larger at final time.
This is because of area-conservation of the flow: the large gyre at the final time has to be ``filled'' from somewhere.  
{In terms of the Rayleigh quotient for~$\Delta_{G_{0,a}}$} it is inexpensive for the initially smaller gyre to feed into the final larger gyre, but the latter needs more mass and that has to partly come from the initially larger gyre. The converse problem (how to fill in the final small gyre) does not exist because some portion of the initial large gyre can simply feed into the final small gyre.

\subsection{Semi-material FTCS and dramatic regime change:  Case study -- destruction and reformation of the polar vortex}
\label{ssec:examplePV}

In this section we study the dynamics of the Antarctic polar vortex over a period of breakup and partial reformation during September to October 2002.
Early dynamical systems approaches to investigate transport and mixing (mostly of ozone) in the polar vortex focused on identifying Lagrangian Coherent Structures as transport barriers, which delineate the vortex~\cite{bowman1993large,joseph2002relation,rypina2007lagrangian,lekien2010computation,beron2010invariant,olascoaga2012brief}.
The Antarctic polar vortex was first diagnosed as a finite-time coherent set in \cite{FSM10} using a historical 14-day reanalysis flow from 1--14 September 2008 on a 475K isentropic surface.
The approximate vortex breakup period we study here has been considered in \cite{lekien2010computation} using finite-time Lyapunov exponents, and by \cite{padberg-gehle_network-based_2017, blachut2020tale,ndouretal21} with coherent-set based analyses based on singular vectors of transfer operators \cite{FSM10}. 

As in \cite{FSM10}, we use velocity data from the ECMWF-Interim data set 
\url{https://ecds.ecmwf.int/datasets}
with $0.5$ degree spatial resolution,
and a temporal resolution of 6 hours.
Because we are interested in vortex break up and reformation we use a longer 60-day flow from 1 September to 30 October September 2002, on a 600K isentropic surface, similar to~\cite{padberg-gehle_network-based_2017}.
Our initial domain $M$ is a disk contained in this surface, which is centred at the south pole and has a radial extent of 6843.75km. In this disk we initialise $N=4367$ trajectories on radial rings equally spaced 187.5km apart. Within each ring, individual initial points are also equally spaced, approximately 187.5km apart.
Following \cite{FSM10} we linearly interpolate the vector field data in space and time, and use a fourth-order Runge--Kutta scheme with a constant step size of 45 minutes to produce the particle positions every 12 hours. By discarding the initial time instance, this results in~$T=120$.
At each such time $t$ we use Matlab's \verb"boundary" command (with default settings) on the $\phi_t(x_i)$, $i=1,\ldots,4367$ to estimate the boundary of~$\phi_t(M)$.
We then apply Dirichlet boundary conditions on the union of these boundary points across time, as described in Section~\ref{ssec:dmapsDirichlet}.

The choice of the diffusion parameter $a$ will be based on the heuristic given in Section~\ref{subsec:Dirichlet_choose_a}. Our initial domain is a disk and the evolution of initial trajectories remains approximately disk-like, so we note that the leading eigenvalue of the disk of radius $r$ for Dirichlet-Laplacian is given by (e.g.\ \cite{berger}) $-j_{0,1}^2/r^2$, where $j_{0,1}\approx 2.4048$ is the first root of the Bessel function~$J_0$.
The largest nonzero eigenvalue for $a^2\Delta$ on an interval $[0,\tau]$ with Neumann boundary conditions is $\lambda_1^{\rm temp} = -a^2\pi^2/\tau^2$.
Equating $-j_{0,1}^2/r^2$ with  $\lambda_1^{\rm temp}$ gives $a_{\min}=2.4048\tau/(r\pi)\approx 0.0268${, which we use in our experiments.}
The value of $\epsilon$ is set to $\epsilon=6835$, based on the heuristic of~\cite{coifman_graph_2008,berry_variable_2016}.

Figure \ref{fig:pv-evals} displays the leading spectra of the  dynamic Laplace operator $\Delta^D$ and the inflated dynamic Laplace operator~$\Delta_{G_{0,a}}$.
\begin{figure}[hbt]
\centering
\includegraphics[width=0.49\textwidth]{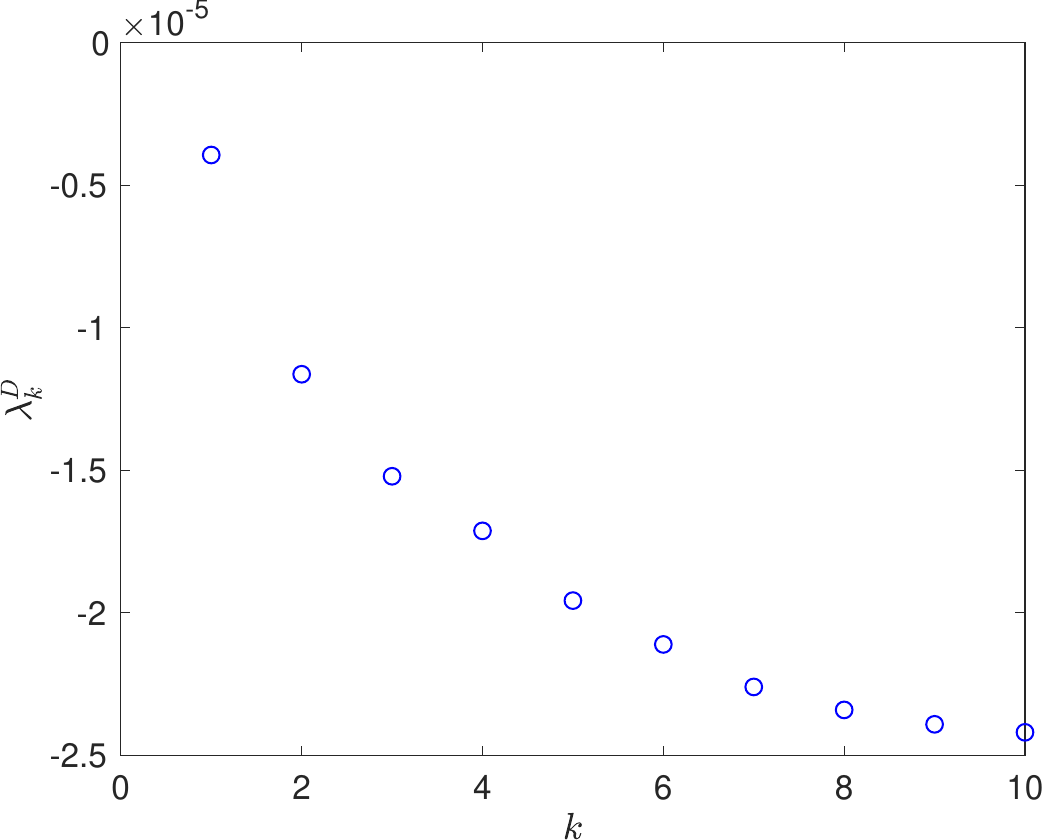}
\includegraphics[width=0.49\textwidth]{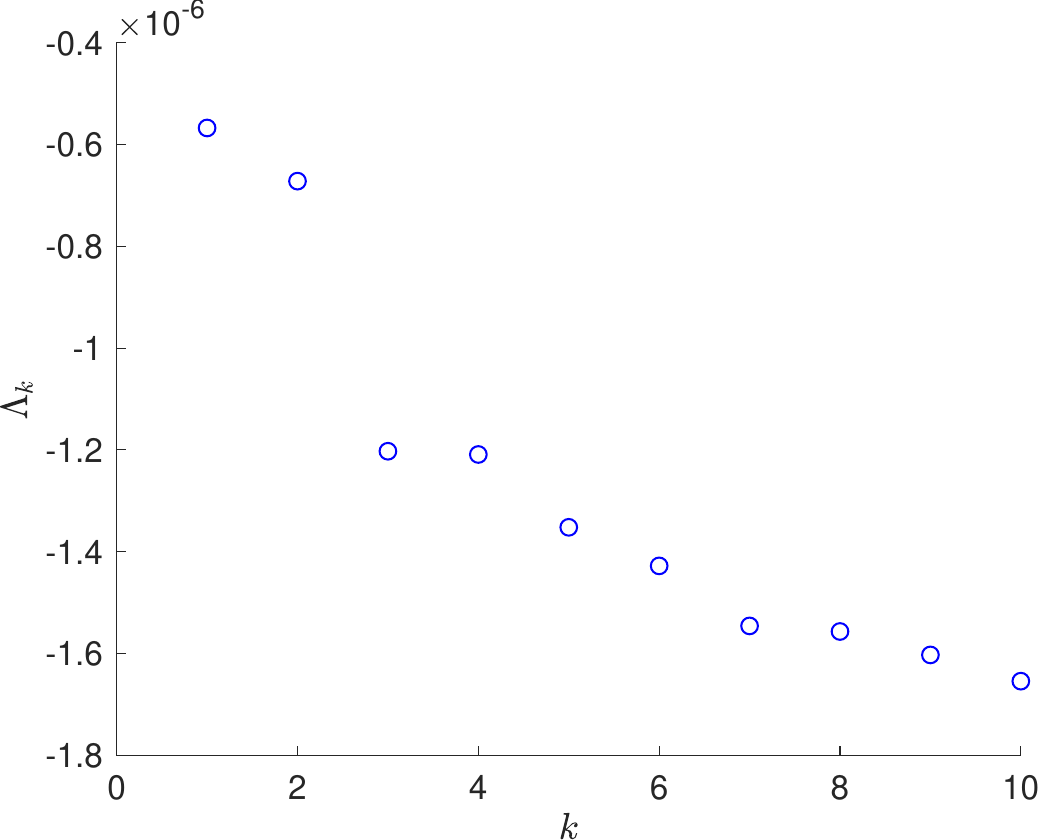}
\caption{Leading eigenvalues for the dynamic Laplacian $\Delta^D$ (left) and the inflated dynamic Laplacian $\Delta_{G_{0,a}}$ (right) for the polar vortex experiment. Note that there are no temporal eigenvalues because of the Dirichlet boundary conditions.
We also see that $\lambda_k^D \le \Lambda_{k,a}$ for $k=1,\ldots,10$, which is consistent with Part~1 of Theorem~\ref{thm:eigbounds}.}
\label{fig:pv-evals}
\end{figure}
Figure~\ref{fig:pv-evals} (right) shows a gap in the spectrum after the leading two eigenvalues {and we therefore use these two eigenvalues to identify our semi-material coherent sets.}
The corresponding two leading eigenfunctions are shown in Figures \ref{fig:PVevec1} and \ref{fig:PVevec2}. 
For further details including trajectories outside the vortex,  we refer the reader to the supplementary movies, where trajectories at the boundary of the domain are highlighted with green circles,  \texttt{mov\_PV\_SpatMode\_1.mp4} and \texttt{mov\_PV\_SpatMode\_2.mp4}.
\begin{figure}[htb]
    \centering
    \includegraphics[width=\textwidth]{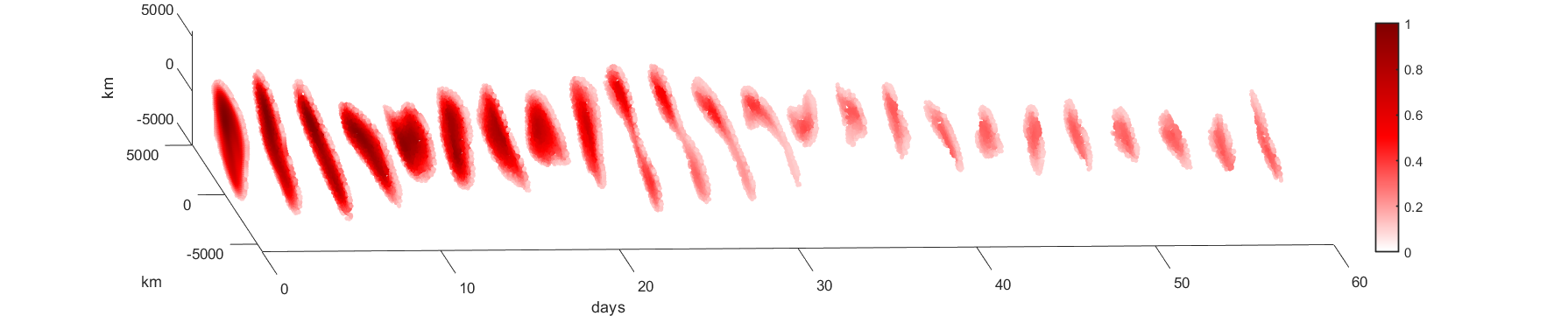}
    \caption{Leading eigenvector of the inflated dynamic Laplacian. A cutoff of 0.1 has been applied to focus on the dominant feature, namely the vortex, which in this eigenvector is strongly highlighted prior to breakup.}
    \label{fig:PVevec1}
\end{figure}
\begin{figure}[htb]
    \centering
    \includegraphics[width=\textwidth]{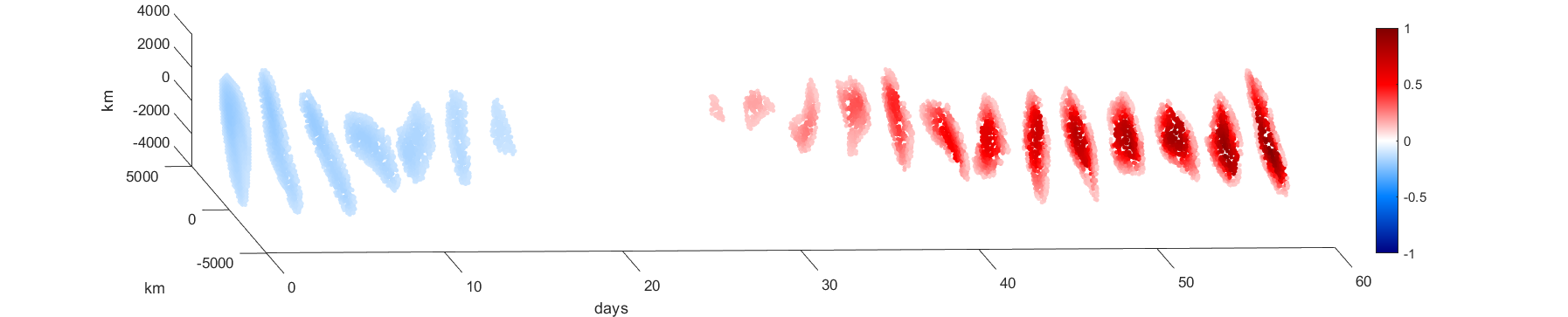}
    \caption{Second eigenvector of the inflated dynamic Laplacian. An absolute-value cutoff of 0.25 has been applied to focus on the dominant features. The second eigenvector strongly highlights the reformed vortex post-breakup.}
    \label{fig:PVevec2}
\end{figure}These two eigenvectors are fed into the SEBA  algorithm \cite{FRS19} to create two sparse basis vectors approximately spanning the top two-dimensional eigenspace.
The superposition of these two sparse basis vectors is displayed in Figure \ref{fig:PVinfl}, where the breakup of the vortex halfway through the 60-day evolution is clearly indicated by the vanishing spacetime eigenfunction values.
In contrast, the leading eigenvector of the dynamic Laplacian in Figure~\ref{fig:PVav} displays an object that remains present over the full 60-day evolution.

On the basis of these results, one may wonder whether there is a vortex ``core'' that is always present over the 60-day evolution, and the ``breakup'' simply corresponds to the ejection of air mass at the periphery of the vortex.
The second eigenfunction of the dynamic Laplacian, shown in Figure \ref{fig:PV-DLevec2}, demonstrates that this is not the case, because blue and white particles are densely embedded within the initial vortex.
\begin{figure}[htb]
    \centering
\includegraphics[width=0.5\textwidth]{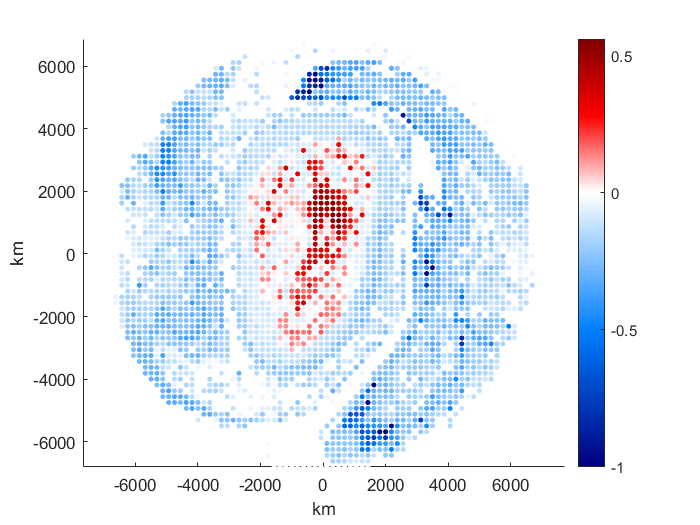}
    \caption{Second eigenvector of dynamic Laplacian on the domain $M$, which corresponds to the initial time slice on 1 September 2002. The central ``vortex'' object is densely perforated with blue/white particles that will later be ejected.}
    \label{fig:PV-DLevec2}
\end{figure}
These blue and white particles will later be ejected;  thus it is not just the periphery that is ejected, but the core of the vortex undergoes ejection and breakup as shown in Figure~\ref{fig:PVinfl}.

As discussed in Section \ref{subsec:Dirichlet_choose_a}, our heuristic choice of $a$ balances a single temporal separation with a single spatial separation, and indeed this is what we see in Figure \ref{fig:PVinfl}.
Increasing the parameter $a$ to $4a_{\rm min}\approx 0.1$ and thus putting greater emphasis on finding semi-material coherent sets that are more material and longer lived, we find a gap in the spectrum of the inflated dynamic Laplace operator~$\Delta_{G_{0,a}}$ after the first eigenvalue. This can be seen in Figure~\ref{fig:PV_evals_multi_a}, where the grey vertical line is at $a_{\min}$ and the eigenvalues computed are for logarithmically spaced $a$ values at steps of a factor~2.
This suggests that there is now a single dominant semi-material coherent set. 
This set would be similar to that indicated by Figure~\ref{fig:PVevec1}, and with further increase in $a$ it would gradually become more material and tend to the one indicated in Figure~\ref{fig:PVav}. 
Thus, consistent with the theory, we see in this example that varying the parameter $a$ controls both the semi-material coherent set's materiality and temporal duration, providing new insights on the dynamics.
\begin{figure}[htb]
    \centering
    \includegraphics[width = 0.5\textwidth]{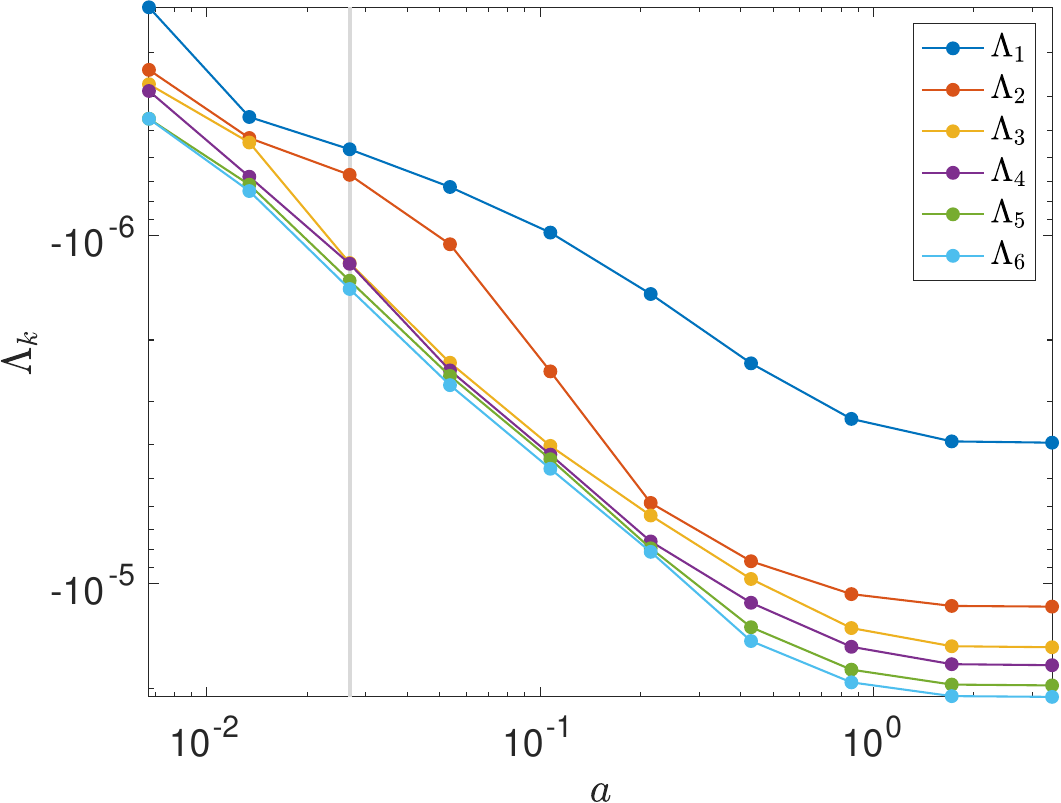}
    \caption{{Eigenvalues of the inflated dynamic Laplacian in the polar vortex example for varying parameter~$a$. The grey vertical line corresponds to~$a_{\min}$. We observe that for $a > 4a_{\min}\approx 0.1$ the gap after the second eigenvalue is now smaller than the gap after the first eigenvalue, suggesting that for increased ``degree of materiality'' there is one dominant coherent set.}}
    \label{fig:PV_evals_multi_a}
\end{figure}

\subsection{Semi-material FTCS and mixed regime changes:  Case study -- a coherent-mixing-coherent fluid flow}


In this section, we study a system undergoing two regime changes over a flow duration from $t=0$ to $t=15$ on a domain $M=[0,2]\times [0,1]$.
We use \revision{a flow reminiscent of}
the standard double-gyre setup of e.g. \cite{shadden-etal, FP09}:
\begin{equation}
    \label{eq:DG}
    \begin{pmatrix}
    \dot{x}_1(t) \\ \dot{x}_2(t)
    \end{pmatrix}
    =
    \begin{pmatrix}
    -\pi A \sin(\pi f(t,x)) \cos(\pi x_2(t)) \\
    \pi A \cos(\pi x_1(t)) \sin(\pi x_2(t)) f_{x_1}(t,x)
    \end{pmatrix},
\end{equation}
where $x(t) \in [0,2] \times [0,1]$, $f(t,x) = \delta \sin(\omega t) x_1(t)^2 + (1-2\delta \sin(\omega t)) x_1(t)$, and $f_{x_1}(t,x) = 2 \delta \sin(\omega t) x_1(t) + 1-2\delta \sin(\omega t)$.
We set $A=0.25$ and $\omega=\pi$; relative to the flow in \cite{FP09} the frequency of separatrix oscillation is halved;  additionally relative to \cite{shadden-etal} the flow speed is 2.5 times faster.
\revision{Note that in the standard double-gyre setup, the term $\cos(\pi x_1(t))$ in \eqref{eq:DG} is replaced by~$\cos(\pi f(t,x))$. With this minor modification we obtain a compressible flow.}
We initially flow for 5 time units with $\delta=0$ (no separatrix oscillation, two steady gyres with invariant sets $[0,1]\times [0,1]$ and $[1,2]\times [0,1]$), then 5 time units with $\delta=0.25$ (separatrix oscillation causes mixing between the outer parts of the gyres), and then another 5 time units with $\delta=0$ (two steady gyres again). 
The higher flow speed and slower oscillation in the central time interval, relative to \cite{shadden-etal, FP09}, creates a greater exchange of fluid between the sets $[0,1]\times [0,1]$ and $[1,2]\times [0,1]$ when $t\in[5,10]$.
Overall, during the time interval $[0,15]$ we initially have perfect coherence (a partition of $M$ into two invariant sets in this example), then mixing in some parts of $M$ with other parts remaining coherent, and finally perfect coherence again, reverting to two invariant sets.

We set the parameter $a=30.0601$ to be $4a_{\rm min}$, where $a_{\rm min}$ is calculated using the heuristic described in \cite[Section~4.1]{FrKo23}, which attempts to match the leading nontrivial temporal and spatial eigenvalues of~$\Delta_{G_{0,a}}$.
In this example, we use $N=800$ trajectories initialised on a uniform grid of $40\times 20$ particles on $M$, recording the particle positions each $0.2$ units of time, leading to $T=76$ time instances.
The diffusion maps bandwidth $\epsilon=0.032$ is again chosen according to the heuristic of~\cite{coifman_graph_2008,berry_variable_2016}.

\subsubsection{Dynamic Laplacian results}
Figure \ref{fig:av-efuns-dg232} shows the leading two nontrivial eigenfunctions of the standard dynamic Laplacian.
\begin{figure}[htb]
    \centering
    \includegraphics[width=\textwidth]{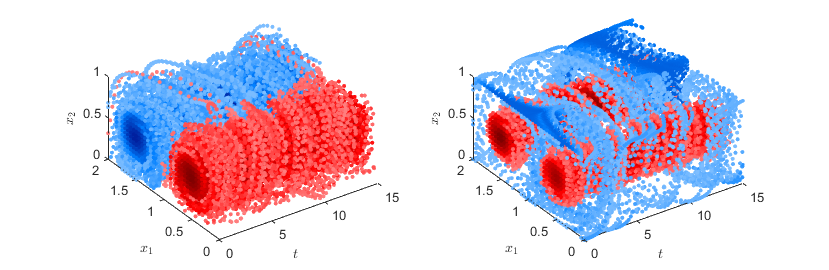}
    \caption{Leading nontrivial eigenfunctions of the dynamic Laplacian for the system (\ref{eq:DG}). Values range from $-1$ to $1$, and values with magnitude below 0.25 have been cut away for clarity.}
    \label{fig:av-efuns-dg232}
\end{figure}
The left panel of Figure \ref{fig:av-efuns-dg232} identifies the central cores of the double gyre flow (in blue and red), which persist as coherent sets \emph{throughout} the time interval $[0,15]$.
The right panel separates the central cores from the remainder of the spatial domain;  again this separation persists through the full flow duration.
If we insert the leading three eigenfunctions of the dynamic Laplacian (the leading trivial constant eigenfunction, in addition to the two in Figure~\ref{fig:av-efuns-dg232}) into the SEBA algorithm \cite{FRS19}, we obtain the three SEBA vectors shown in Figure~\ref{fig:av-seba-dg232}.
\begin{figure}
    \centering
    \includegraphics[width=\textwidth]{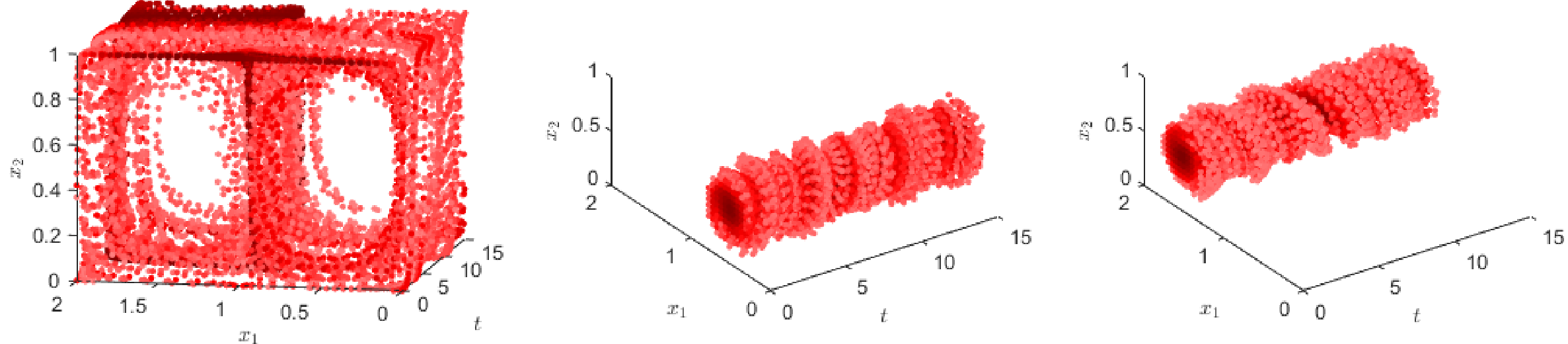}
    \caption{Three SEBA vectors from the leading three eigenfunctions of the dynamic Laplacian for the system (\ref{eq:DG}).}
    \label{fig:av-seba-dg232}
\end{figure}
SEBA separates the two persistent inner cores in Figure~\ref{fig:av-seba-dg232} (centre and right panels), and a persistent outer mixing area (left panel).
As far as a coherent set analysis is concerned, this is correct, but the regime changes at $t=5$ and $t=10$ do not show up because the dynamic Laplacian is designed to find sets that are coherent throughout the full flow duration.

\subsubsection{Inflated dynamic Laplacian results}
{The leading four spatial eigenvalues of the inflated dynamic Laplacian are well spaced apart.
The third eigenvalue is temporal, and after these five eigenvalues there is a long continuum of narrowly spaced spectral points.
We therefore choose $J=4$ and use the leading three nontrivial spatial eigenfunctions $F^\spat_2, F^\spat_3$, and $F^\spat_4$ in Algorithm \ref{alg:workflow}}.
Figure~\ref{fig:dg232evecs} shows these leading three {nontrivial} spatial eigenfunctions of the inflated dynamic Laplacian.
\begin{figure}
    \centering
    \includegraphics[width=\textwidth]{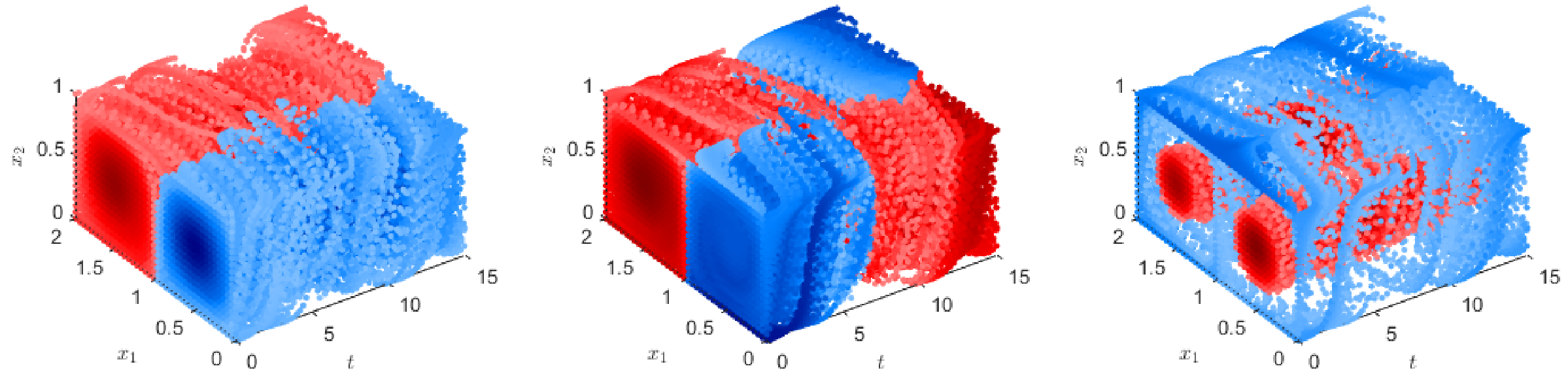}
    \caption{Leading three spatial eigenfunctions of the inflated dynamic Laplacian for system \eqref{eq:DG} shown in the co-evolved frame $\MM_1$. Values range from $-1$ to $1$, and values with magnitude below 0.25 have been cut away for clarity.
    }
\label{fig:dg232evecs}
\end{figure}
The left and right panels of Figure~\ref{fig:dg232evecs} are ``smoother'' versions of the left and right panels in Figure~\ref{fig:av-efuns-dg232};  for example, note that in both the left and right panels of Figure~\ref{fig:dg232evecs} the blue and red sets extend further toward the spatial boundary of $\MM_0$ than those in Figure~\ref{fig:av-efuns-dg232}.
The centre panel of Figure~\ref{fig:dg232evecs} is unusual because the colour does not follow the trajectories;  in particular the sharp transition in time from blue to red and red to blue is very different to the motion of the particles.

To provide greater detail, individual timeslices of the spacetime eigenvectors in Figure~\ref{fig:dg232evecs} are displayed in Figure~\ref{fig:232frames} {and are shown in the supplementary movie \texttt{DG232movie.mp4}}.
\begin{figure}
    \centering
\includegraphics[width=.8\textwidth]{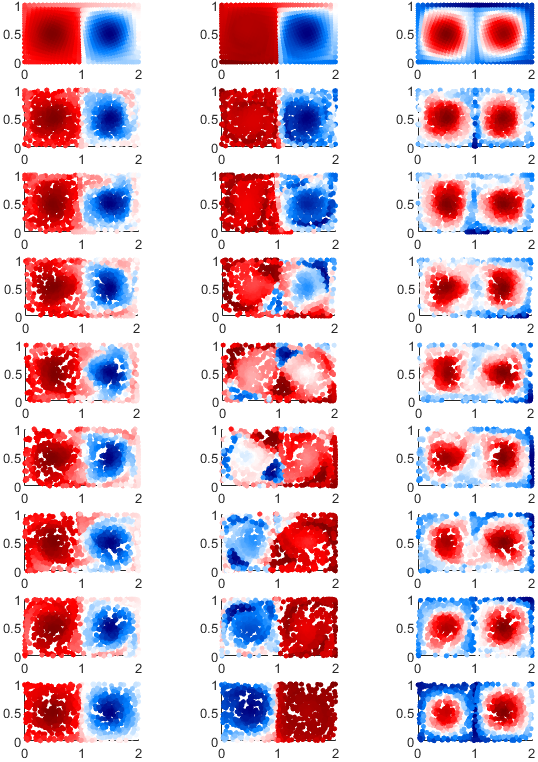}
    \caption{Each row displays a time slice on $\MM_1$ of the leading three spatial eigenfunctions shown in Figure~\ref{fig:dg232evecs}.
    From top to bottom, the nine rows show time slices at times $t = 0.4, 4.4, 5.4, 6.6, 7.6, 8.6, 9.6, 10.6, 14.6$, respectively. We have deliberately displayed more time slices throughout the more complex fluid exchange region between $t=5$ and $t=10$ to focus on the more complicated transitions there.}
 \label{fig:232frames}
\end{figure}
While the left and right columns of Figure~\ref{fig:232frames} show minimal change across time, the centre column has a dramatic change, completely switching the strongly negative (blue) and strongly positive (red) values. 
Such an operation is, in principle, expensive because of the $L^2$ norm of the gradient in (\ref{eq:LB-lamk}), but in fact a price is only paid if the eigenfunction value changes along a path that is not a trajectory in~$\MM_1$.
One can see in the central column of Figure~\ref{fig:232frames} that when the blue/red switching occurs, the colour transformation occurs through ``lobes'' (see e.g.\ \cite{KW90})  and thus follows trajectories with little cost.
On the other hand, in the two gyre cores, there is no way to advect from one side of the domain to the other, and so a direct price is paid by simply changing the eigenfunction value, which is permitted by the time expansion.

Following ideas from \cite[Section~4.3]{FrKo23}, we inspect the $L^2$ norms of the individual time slices of the eigenfunctions.
We expect that lower time-slice norms correspond to temporal regimes of greater mixing.
The transition into a moderate mixing regime in the time interval $[5,10]$ is indeed captured by the $L^2$ norms of the time slice of the eigenfunction in the centre panel of Figure~\ref{fig:dg232evecs};  see Figure~\ref{fig:L2timefibrenorm} (left).
\begin{figure}[htb]
   \centering
\includegraphics[width=1.0\textwidth]{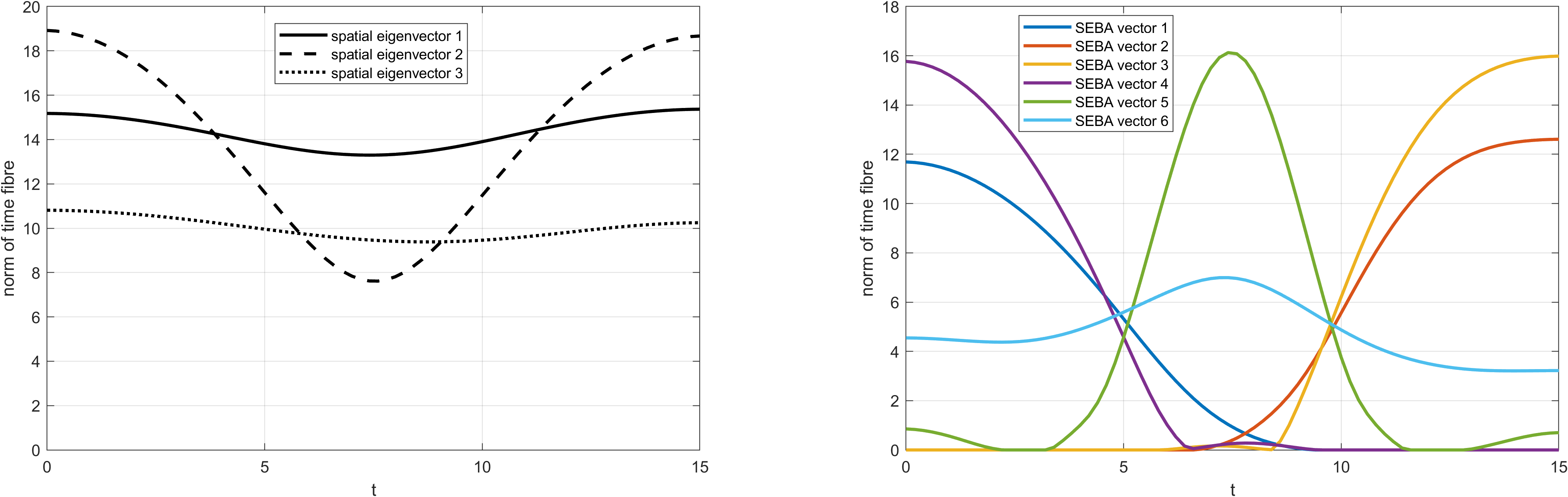}
\caption{
Left: $L^2$ norms of time-fibres of the leading three spatial eigenfunctions shown in Figure \ref{fig:dg232evecs}. Right: $L^2$ norms of time fibres of the six SEBA vectors shown in Figure \ref{fig:dg232seba}.}
\label{fig:L2timefibrenorm}
\end{figure}
Thus, the dynamic regime changes from simple twin gyres to partly mixing double-gyre flow and back to simple twin gyres have been detected by the second nontrivial eigenfunction of the inflated dynamic Laplacian.
In fact, all of the essential macro-dynamical information is encoded in these leading three spatial eigenfunctions. However, because of the orthogonality restrictions of the eigenfunctions, to fully extract all of this macro-information with SEBA, we need to augment the collection of eigenfunctions.
A general theory and numerical strategy for dealing with regime changes is elaborated in Section \ref{sec:efunstruc}.
For the moment, we detail the implementation of the augmentation procedure arising from Section \ref{sec:efunstruc} for this particular flow.

\paragraph{Eigenfunction augmentation.}
For each of the three spacetime eigenfunctions in Figure~\ref{fig:dg232evecs}, we create a companion vector, which is a spacetime vector that is constant on each time slice, with a value given by the $L^2$ norm of the eigenvector on each time slice.
These $L^2$ norms on time slices are precisely what is shown in Figure~\ref{fig:L2timefibrenorm} (left).
As explained in Section \ref{sec:efunstruc}, this augmentation is necessary in the Neumann boundary condition setting to (re)generate degrees of freedom that are restricted by the orthogonality of the inflated dynamic Laplace operator eigenvectors.
We have now doubled our collection of three eigenvectors to a collection of six vectors, and this collection is  input to SEBA~\cite{FRS19} to isolate individual semi-material coherent sets.

Figure \ref{fig:dg232seba} displays the results, with values above 0.1 shown.
\begin{figure}[hbt]
    \centering
\includegraphics[width=\textwidth]{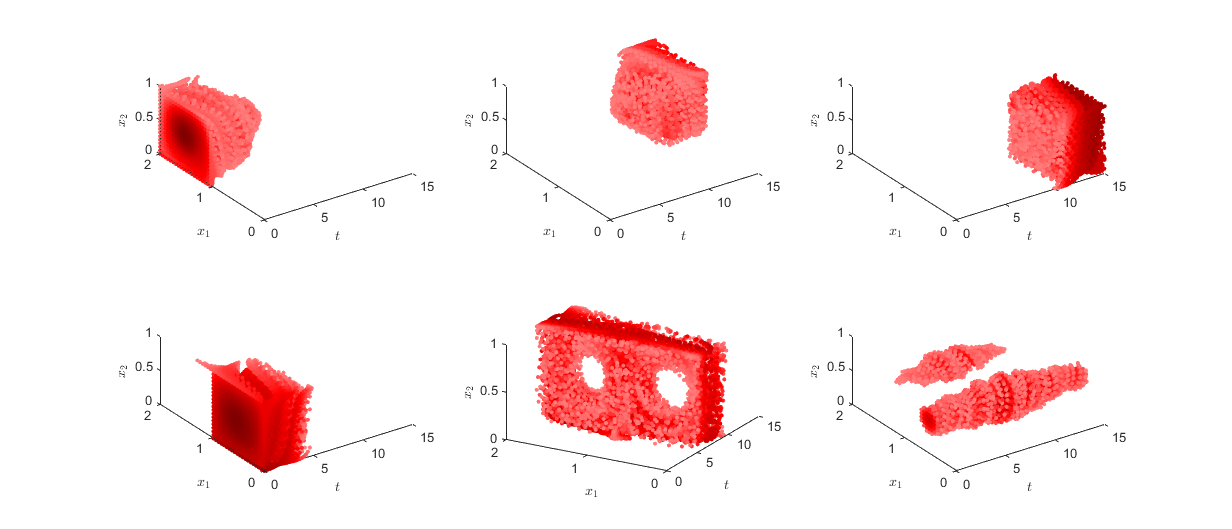}
    \caption{Output of SEBA applied to the three eigenfunctions in Figure \ref{fig:dg232evecs} combined with their three augmented counterparts. All relevant coherent behaviour for the system (\ref{eq:DG}) is identified, as well as the timing of the regime changes in the dynamics. Values above a 0.25 cutoff are shown.}
    \label{fig:dg232seba}
\end{figure}
The corresponding time-fibre norms of the resulting six SEBA vectors are displayed in Figure~\ref{fig:L2timefibrenorm} (right).
Considering the values of the time-slice norms in Figure~\ref{fig:L2timefibrenorm} (right) we expect SEBA vectors numbered 1 and 4 to show objects approximately supported in the first 1/3 of the time evolution, and SEBA vectors numbered 2 and 3 to show objects approximately supported in the final 1/3 of the time evolution. 
SEBA vector 5 is primarily supported in the central 1/3 of the time evolution and SEBA vector 6 is supported throughout the entire flow duration.
These statements are reflected in Figure~\ref{fig:dg232seba}, where we clearly see the ``twin gyre'' invariant sets $[0,1]\times [0,1]$ and $[1,2]\times [0,1]$ over the time intervals $[0,5]$ and $[10,15]$ (upper row and lower left).
The lower central panel of Figure~\ref{fig:dg232seba} identifies the external ``mixing'' region over the time interval $[5,10]$, and the lower right panel identifies the persistent central gyre cores, which remain coherent throughout times in $[0,15]$.
In summary, the six sets in Figure \ref{fig:dg232seba} provide a complete picture of coherent behaviour for the system (\ref{eq:DG}), including correctly identifying the times at which the flow encounters regime switches.

We remark that the analysis in Section \ref{sec:efunstruc} concerns intervals of coherence either abutting one another, or separated by intervals of complete mixing.
The system (\ref{eq:DG}) is more complicated and does not belong to either of these ``idealised'' situations, because we have two coherent intervals with a semi-mixing interval in between them.
Even though the system (\ref{eq:DG}) has no fully mixing regimes, the same eigenfunction augmentation heuristic produces highly informative results from only three leading nontrivial eigenfunctions of the inflated dynamic Laplacian.

\subsection{Computational complexity}

\revision{The simultaneous discretisation of space and time by the overall $NT$ degrees of freedom raises the question whether good resolutions can be achieved by maintaining reasonable computation time. This is indeed the case because our approach constructs a matrix $P_a(\ep) \in \R^{NT \times NT}$ that is in practice very sparse and we only compute some leading eigenvalues with largest magnitude (see the discussion at the very beginning of Section~\ref{sec:dmapsDL}). Thus, essentially all computations are of linear complexity in $NT$, as generally $T \ll N$ because $T$ mesh points resolve a one-dimensional domain, while $N$ mesh points resolve a spatial domain of dimension larger or equal to~2. Memory requirements remain mild.}

\revision{All examples take of the order of ten seconds in Matlab on a laptop with an Apple M2 (2022) chip and 24 GB RAM. More precisely, for the first example with $N=1350$ and $T=101$, the trajectories were computed within 5s, $P^{(x)}$ and $P^{(t)}$ were set up in less than 1s and required a memory of $270$MB and $210$MB, respectively, and the dominant 10 eigenmodes were computed in~3s {using option 8(a) of Algorithm \ref{algo:Strang}.}
The second example with $N=4367$ and $T=120$ produced similar numbers for data generation and matrix assembly, with the addition of 6s for computing the boundary indices, and the eigenvalue computation took 57s, {using option 8(a) in Algorithm~\ref{algo:Strang}.} This longer duration is partly because the full $P^{(t)}$ matrix took almost 1GB~RAM, and 
{could be sped up by using 8(b) instead of 8(a)}.
The timings for the third example, {using option 8(a) in Algorithm \ref{algo:Strang},}  were less than for the first example.
Julia code to carry out Algorithm \ref{algo:Strang}, including the first and third examples, is available at \url{https://github.com/gfroyland/Inflated-Dynamic-Laplacian}. Python code to carry out the algorithm, including the examples, is available at~\url{https://github.com/pkoltai/InflatedDynamicLaplacian}.}

\section{Deeper eigenfunction structure of the inflated dynamic Laplacian}
\label{sec:efunstruc}

In this section we describe the theoretical structure of the spectrum and eigenfunctions of the inflated dynamic Laplacian.
To begin, we recall that under Neumann boundary conditions the {spatial} eigenfunctions $F:\mathbb{M}_0\to \mathbb{R}$ corresponding to nonzero eigenvalues are orthogonal to the constant function (the leading eigenfunction with zero eigenvalue).
It can then be shown that $t\mapsto\int_M F(t,x)\ d\ell(x)$ is a constant function a.e.\ on~$[0,\tau]$;  see the discussion in~\cite[Section~3.4.2]{FrKo23} for more details.
Under Dirichlet boundary conditions, because the leading eigenfunction has no particular structure (apart from taking on high values in space and time where the dynamics is more coherent), the time-fibre integrals of $F$ are not constant in time.

As discussed in~\cite[Section 4.3]{FrKo23} and indicated in Figures \ref{fig:seba-L2-norms} and \ref{fig:L2timefibrenorm}, the $L^2$ time-fibre norms are key indicators of the presence of coherence versus complex phase space mixing.
If $F:\mathbb{M}_0\to \mathbb{R}$ is a spatial eigenfunction of the inflated dynamic Laplacian, we expect low values of $\|F(t,\cdot)\|_{L^2(M)}$ in time intervals of strong mixing.
Exactly how low will depend on the time-diffusion parameter $a$;  smaller $a$ will allow more variation of $\|F(t,\cdot)\|_{L^2(M)}$ in time~$t$.
In the two subsections below, we discuss idealised situations where there are subintervals of time during which several strongly coherent sets are present. In Section~\ref{ss:punc}, these subintervals are separated by further subintervals of time in which strong global mixing occurs, and in Section~\ref{ss:unpunc}, these subintervals of coherence rapidly switch to the next subinterval of coherence.

\subsection{Coherent time intervals punctuated by well-mixing intervals}
\label{ss:punc}

The first of two settings we consider is where the phase space undergoes a number of transitions in the number and/or location of coherent sets, and between each transition there is complete phase space mixing.
More precisely, $[0,\tau]$ may be written as the disjoint union $[t_0,t_0']\cup(t_0',t_1)\cup[t_1,t_1']\cup(t_1',t_2)\cup\cdots\cup(t_{n-1}',t_n)\cup[t_n,t_n']$, with $0=t_0<t_0'<t_1<t_1'\cdots<t_n<t_n'=\tau$.
Within the intervals $[t_i,t_i']$, $i=0,\ldots,n$ we suppose we have $k_i$ strongly coherent sets that partition $M$ and that there is strong mixing within each of these individual coherent sets.
During the intervals $(t_i',t_{i+1})$ we suppose we have complete phase space mixing.

We provide a summary of the expected structure of the spatial eigenfunctions and then present arguments to support the summary statements.
\begin{enumerate}
\item \textit{Dirichlet-in-space boundary conditions:}
We expect to obtain $\sum_{i=0}^n k_i$ leading {spatial} eigenvalues, followed by a large gap to the remainder of the spatial spectrum.
To each interval $[t_i,t'_i]$ there corresponds a $k_i$-dimensional eigenspace $E_i$, where each of the eigenfunctions $F$ in the corresponding spanning  set has support approximately equal to $[t_i,t'_i]$. Generically, one of the eigenfunctions approximately supported on $[t_i,t'_i]$ (the one corresponding to the least negative eigenvalue from the collection of size $k_i$)  is approximately positive and approximately supported on one of the coherent sets in the time interval $[t_i,t'_i]$.  The remaining $k_i-1$ eigenfunctions in the spanning set of $E_i$ have approximate level-set structures that indicate the $k_i$ coherent sets in $[t_i,t'_i]$.
\item \textit{Neumann-in-space boundary conditions:}
We expect to obtain $1+\sum_{i=0}^n (k_i-1)$ leading {spatial} eigenvalues, followed by a large gap to the remainder of the spatial spectrum.
The eigenfunction corresponding to the zero eigenvalue is the constant function;  this is the ``1'' in the above sum.
To each interval $[t_i,t'_i]$ there corresponds a $(k_i-1)$-dimensional eigenspace $E_i$, where each of the eigenfunctions $F$ in the corresponding spanning  set has support approximately equal to $[t_i,t'_i]$.
The $k_i-1$ eigenfunctions in the spanning set of $E_i$ have approximate level-set structures that indicate the $k_i$ coherent sets in $[t_i,t'_i]$.
\end{enumerate}
We now argue why the above enumerations are true and provide some details on the structure of the spatial eigenfunctions.
First we discuss their basic temporal structure. {Consider a dominant, nontrivial (nonconstant) spatial eigenfunction $F$}.
As discussed above (and in more detail in \cite[Section 4.3]{FrKo23}), for a suitably chosen $a$ we expect that $F$ will have values near to zero on the intervals $(t_i',t_{i+1})$, $i=1,\ldots,n$.
Further, in the generic case of distinct eigenvalues, $F$ will have support \emph{only} (approximately) on the time interval~$[t_i,t_i']$. 
To argue why this is the case, consider the variational representation (\ref{eq:LB-lamk}) for an eigenfunction $F$ supported in two distinct time intervals $[t_i,t_i']$ and $[t_j,t'_j]$, $i\neq j$. 
Generically, if we were to separate this eigenfunction into two functions, each supported on only one of the time intervals (and applying minor corrections to maintain smoothness and orthogonality), the Rayleigh quotient  (\ref{eq:LB-lamk}) for the separated function on one interval would be less negative than {for} the other, and in particular, less negative than {for} the original function~$F$.
In this way, functions supported on single intervals are preferred through the variational property.

Next, we consider the spatial structure of the leading eigenfunctions, \textit{within} each of the time intervals of support $[t_i,t_i']$.
Because each time interval of coherence has been isolated from the others by an interval of complete phase space mixing, the spatial structure of eigenfunctions within a single time interval is analogous to the structure one obtains from the dynamic Laplacian \emph{restricted to each time interval of coherence}.
Thus, in the Dirichlet boundary condition case we expect the leading $k_i$ eigenfunctions to have approximate level sets approximately indicating the coherent sets in time interval $[t_i,t'_i]$. 
In the Neumann boundary condition case we expect the leading $k_i-1$ eigenfunctions to have approximate level sets approximately indicating the coherent sets in time interval $[t_i,t'_i]$;  see \cite[Section III.A]{FJ15}.
In both boundary condition settings, the approximate level sets of $F(t,\cdot)$ on slice $t\in[t_i,t_i']$ are such that {the pushforwards} $(\phi_t)_*(F(t,\cdot))$ highlight coherent sets on~$M_t$.

\subsubsection{Eigenfunction augmentation with Neumann boundary conditions}
\label{sssec:Neumann_augment}

As we have noted already, both the dynamic Laplacian $\Delta^D$ and the  inflated dynamic Laplacian $\Delta_{G_{0,a}}$  with Neumann boundary conditions have a single leading eigenfunction that is constant.
In the case of $\Delta^D$ the eigenfunction is constant on $M$, while for $\Delta_{G_{0,a}}$ this eigenfunction is constant on $\mathbb{M}_0$.
This difference in domain has implications for orthogonality of the remaining eigenfunctions.
In fact, the heuristic mentioned above: ``on each coherent time interval inflated dynamic Laplacian eigenfunctions behave analogously to dynamic Laplacian eigenfunctions restricted to that time interval'' is not exactly correct.

To isolate individual coherent sets using eigenfunctions of the dynamic Laplacian, one would usually apply SEBA to a collection of $k$ leading eigenfunctions consisting of the leading (constant) eigenfunction and the remaining $k-1$ eigenfunctions with (space)-integral zero (see {\cite[Figure~9]{FrKo23}} or the examples in~\cite{FRS19}).
Because we have assumed strong coherence and strong mixing within each {of $k$} coherent set{s}, SEBA will return $k$ functions, each approximately supported on one of the coherent sets in~$M$. 
In the inflated dynamic Laplacian case we are currently considering, because we only have a \textit{single} leading eigenfunction that is constant on the \emph{entire spacetime manifold}, we need to augment the collection of functions input to SEBA so that we can capture $\sum_{i=0}^n k_i$ coherent sets from $1+\sum_{i=0}^n (k_i-1)$ eigenfunctions.
In theory, we should remove the leading constant function and augment with $n+1$ functions, each constant on a unique time interval from the collection $[t_i,t_i']$, $i=0,\ldots,n$.
This would yield $\sum_{i=0}^n k_i$ functions to match the expected $\sum_{i=0}^n k_i$ coherent sets.
A drawback with this approach is that we would need to first accurately identify the time intervals $[t_i,t'_i]$.

A simple and performant approach is to instead remove the leading constant function and for \emph{each} of the $\sum_{i=0}^n (k_i-1)$ eigenfunctions $F$, add a companion function $\hat{F}(t,x):=\|F(t,\cdot)\|_{L^2(M)}\mathbf{1}_M(x)$, which is constant in space on each time fibre, taking the constant value given by the $L^2$ norm of the time fibre.
Fixing an interval $[t_i,t_i']$, the corresponding $k_i-1$ companion functions will all have similar support (restricted to approximately $[t_i,t_i']$).
After making this augmentation, we apply SEBA in space-time to all $2\sum_{i=0}^n(k_i-1)$ functions.
We should find $\sum_{i=0}^n k_i$ SEBA functions, each supported on a single coherent set within a single time interval of coherence, along with $\sum_{i=0}^n (k_i-2)$ ``residual'' functions, which are easily identified because they will be constant-valued on each time slice. 

We briefly illustrate this augmentation procedure for the Childress--Soward flow in \cite[Section 7.3]{FrKo23}.
There we expect eight distinct coherent sets in spacetime over two distinct subintervals of time. In the above notation we have $n=2$, $k_1=4$, and~$k_2=4$.
As noted above we have $1+(4-1)+(4-1)=7$ spatial eigenfunctions, which is one too few to capture all eight sets.
Removing the constant eigenfunction, leaves six spatial eigenfunctions, whose timeslice norms across time are shown in Figure~\ref{fig:cmcmnorms}.
\begin{figure}[hbt]
    \centering
\includegraphics[width=\textwidth]{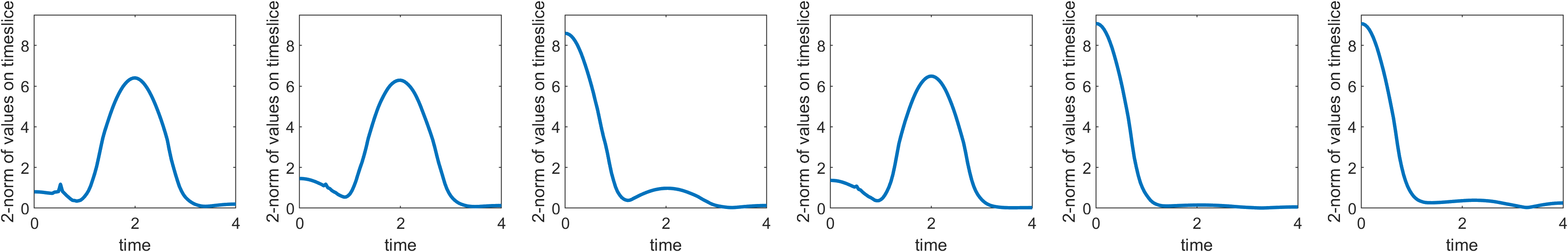}
    \caption{Timeslice norms of the leading six eigenfunctions of the inflated dynamic Laplacian (excluding the leading constant function) {for the Childress--Soward flow}. One clearly sees peaks in norms on the subintervals $[0,0.6]$ and $[1.4,3.2]$, indicating the presence of coherent sets during these times. In this experiment we use an initial grid of $20\times 20$ particles, 101 time instances in the time interval $[0,\tau]$, and temporal diffusion strength of $a=4/\pi$.  See \cite{FrKo23} for details, where more particles and time instances were used.}
\label{fig:cmcmnorms}
\end{figure}
Augmenting as described above yields 12 functions, which exceeds the eight required.
Applying SEBA to these 12, we expect eight SEBA vectors to be supported on the eight coherent sets in spacetime and another four ``residual'' functions, which are constant on each timeslice and therefore easily identifiable by computing variance on each timeslice.
See Figure~\ref{fig:cmcmvars}.
\begin{figure}[hbt]
    \centering
\includegraphics[width=\textwidth]{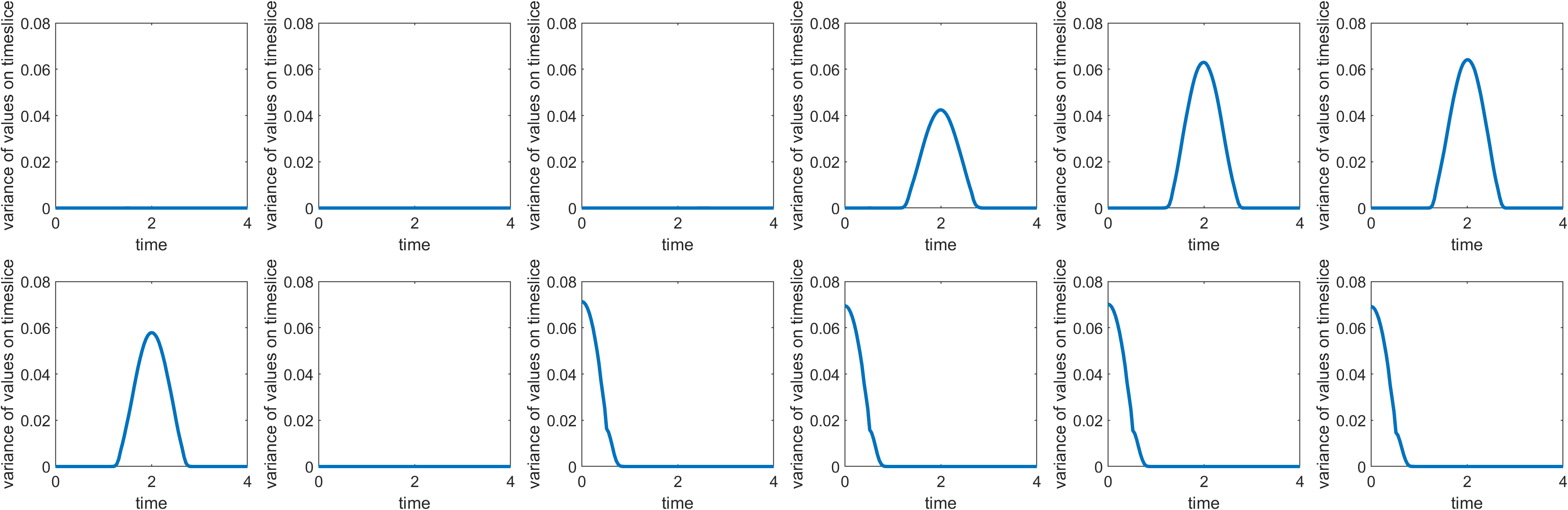}
    \caption{Timeslice variances of SEBA vectors. 
    SEBA is applied to 12 spacetime vectors: the leading six nonconstant spatial eigenfunctions of the inflated dynamic Laplacian (those whose spacetime norms are shown in Figure \ref{fig:cmcmnorms}) and six augmented vectors as described above. The variance clearly separates the ``residual'' vectors with zero timeslice variances (those in positions 1--3 and 8) from the eight SEBA vectors with nontrivial spatial information (those in positions 4--7 and 9--12), which are supported on the eight distinct coherent sets in spacetime.}
\label{fig:cmcmvars}
\end{figure}

\subsection{Distinct coherent time intervals with no temporal separation between them}
\label{ss:unpunc}

We now consider the situation where one or more of the intervals $(t_i',t_{i+1})$ is squeezed out of existence.
That is, $t_i'=t_{i+1}$ and we are in a situation where the interval $[t_i,t_{i+1}]$ has $k_i$ coherent sets and there is a ``switch point'' at around $t_{i+1}$ to an interval $[t_{i+1},t_{i+1}']$ with $k_{i+1}$ coherent sets.
We discuss the Neumann boundary condition setting below;  the Dirichlet boundary condition setting is identical if one replaces $k_i-1$ with $k_i$ in discussions of dimensions of the eigenspaces $E_i$.

Consider the $((k_i-1)+(k_{i+1}-1))$-dimensional direct sum $E_i\oplus E_{i+1}$ spanned by the union of the eigenfunctions supported on one of the two intervals we have abutted against one another.
Prior to squeezing out the complete mixing interval $(t_i,t_{i+1})$ we understand from Section~\ref{ss:punc} how we expect the eigenfunctions to behave.
Each of the $k_i-1$ eigenfunctions $F(t,\cdot)$ in $E_i$, which are supported in $t\in [t_i,t_i']$, will tend to take on similar values within coherent sets in $\{t\}\times M$;  this will be most clear in the pushforward $(\phi_t)_*(F(t,\cdot))$.
A similar situation occurs on the time interval $[t_{i+1},t_{i+1}']$.

When the interval $(t_i,t_{i+1})$ is squeezed out, the $k_i-1$ eigenfunctions and the $k_{i+1}-1$ eigenfunctions will try to merge in a way that requires small temporal derivatives around the time $t_{i+1}$.
Such merging can occur when some of the $k_i$ coherent sets in $\{t\}\times M$, $t\in [t_i,t_i']$ are able to approximately ``match up'' in space to the $k_{i+1}$ coherent sets in $\{t\}\times M$, $t\in [t_{i+1},t_{i+1}']$,
Suppose that $m_i\le \min\{k_i-1,k_{i+1}-1\}$ eigenfunctions are able to approximately merge in this way.
After squeezing out the fully mixing time interval $(t_i,t_{i+1})$ we expect $m_i$ eigenfunctions that are approximately ``material'' in the sense that their values remain relatively unchanged throughout the combined interval $[t_i,t_{i+1}']$.
By perturbation theory throughout the process of squeezing out $(t_i,t_{i+1})$ we expect the  $((k_i-1)+(k_{i+1}-1))$-dimensional direct sum $E_i\oplus E_{i+1}$ to remain relatively unchanged.
Because the pre-squeezed $((k_i-1)$- and $(k_{i+1}-1))$-dimensional eigenspaces were completely unlinked in the time direction, after squeezing out $(t_i,t_{i+1})$ we will have $m_i$ approximately ``material'' eigenfunctions and $k_i+k_{i+1}-2+m_i$ eigenfunctions that are ``nonmaterial'' in the sense that these eigenfunctions exhibit high temporal derivatives around $t_{i+1}$.
These high temporal derivatives signify the termination of one coherent set and the initiation of another.
This effect can be clearly observed in the second eigenfunction shown in Figure~\ref{fig:dg232evecs}.

\section*{Acknowledgements}
The authors thank K.~Padberg-Gehle for preprocessed data for the polar vortex example (Section~\ref{ssec:examplePV}). We also thank K.~Padberg-Gehle and F.J.~Beron-Vera for discussions on this example.

\section*{Funding}

JA was supported by an Australian Research Council (ARC) Disovery Project (DP210100357). The research of GF was partially supported by an Einstein Visiting Fellowship funded by the Einstein Foundation Berlin, and an ARC Discovery Project (DP210100357).
GF is grateful for generous hospitality at the Department of Mathematics, Free University Berlin and the Department of Mathematics, University of Bayreuth, during research visits.
PK has been partially supported by Deutsche Forschungsgemeinschaft (DFG) through grant CRC 1114 ``Scaling Cascades in Complex Systems'', Project Number 235221301, Project A08 ``Characterization and prediction of quasi-stationary atmospheric states''. 

\appendix

\section{Proof of Theorem \ref{thm:eigbounds}}
\label{app:eigenbounds}

\begin{proof}
The first step of the proof establishes standard properties of the eigenvalue problem for $\Delta_{G_{0,a}}$ by classical weak theory, and the second step verifies the claims of the theorem analogously to~\cite[Theorem 3.3]{FrKo23}.

Let $H^1_\dir(\MM_0)$ denote the Sobolev space of weakly differentiable $L^2(\MM_0)$ functions whose derivatives are in $L^2(\MM_0)$ and satisfy homogeneous Dirichlet boundary conditions on $[0,\tau] \times \partial M$ (i.e., ``in space''). This space can be obtained as the closure of the set of smooth functions that vanish on $[0,\tau] \times \partial M$ -- denoted by $C_c^\infty(\mathbb{M}_0)$ above -- with respect to the $H^1(\MM_0)$ norm.

\textbf{Step I.} First we will show the coercivity of the bilinear form associated to~$\Delta_{G_{0,a}}$,
\begin{equation}
    \label{bfg}
B(F,G) := \int_0^{\tau}\!\!\! \int_M a^2 \partial_t F(t,\cdot) \partial_t G(t,\cdot) + \langle \nabla_{g_t} F(t,\cdot), \nabla_{g_t} G(t,\cdot) \rangle_{g_t} \, \frac{d\ell}{a} \, dt,\quad F,G \in H^1_{\dir}(\MM_0).
\end{equation}
Note that for $F \in H^1_\dir(\MM_0)$ one has $F(t,\cdot) \in H^1_0(M)$ for almost every~$t\in [0,\tau]$.
The coercivity will be based on the Poincar\'e--Friedrichs inequality (in space),
\begin{equation}
    \label{eq:Poincare}
    \int_M \langle \nabla f, \nabla f \rangle \, d\ell \ge \tilde{c} \|f\|_{H^1(M)}^2 \quad \text{for } f\in H^1_0(M),
\end{equation}
where $\tilde c>0$ is a constant that does not depend on~$f$. For $f \in H^1(M)$ we have pointwise on $M$ that
\[
\langle \nabla_{g_t} f, \nabla_{g_t} f \rangle_{g_t}  = \langle \nabla f, g_t^{-1} \nabla f \rangle \ge \frac{1}{\lambda_{\max}(g_t)} \| \nabla f \|_2^2,
\]
where $\lambda_{\max}(g_t)$ denotes the largest eigenvalue of the metric tensor~$g_t$. With
\[
\lambda := \min_{t\in[0,\tau]} \min_M \frac{1}{\lambda_{\max}(g_t)},
\]
we obtain from \eqref{eq:Poincare} that
\begin{equation}
    \label{eq:Poincare2}
    \int_M \langle \nabla_{g_t} f, \nabla_{g_t} f \rangle_{g_t} \, d\ell \ge c \|f\|_{H^1(M)}^2 \quad \text{for every } f\in H^1_0(M) \text{ and } t\in[0,\tau],
\end{equation}
with $c = \lambda \tilde{c}$. Since $g_t$ is smooth on $M$ and in $t$, and $\MM_0$ is compact, we obtain that $\lambda>0$ and hence that~$c>0$.  Then, repeatedly applying \eqref{eq:Poincare2} with $f=F(t,\cdot)$ for a.e.\ $t\in [0,\tau]$ we obtain from \eqref{bfg} that
\begin{equation}
    \label{eq:coercivity}
    B(F,F) \ge \min \{a, c/a\} \| F \|_{H^1(\MM_0)}^2,
\end{equation}
the coercivity of~$B$. Since $H^1_\dir(\MM_0)$ is continuously, densely, and compactly embedded in $L^2(\MM_0)$, \cite[Theorem~11.2]{hackbusch2017elliptic} implies that $\Delta_{G_{0,a}}$ has discrete spectrum, all eigenvalues have finite multiplicity, and they only accumulate at~$-\infty$. As $\min\{a,c/a \} > 0$, from \eqref{eq:coercivity} we also obtain that all eigenvalues are strictly negative:~$\Lm_{k,a} < 0$.

\textbf{Step II.}
Once it is established that $\Delta_{G_{0,a}}$ has discrete spectrum, it satisfies the mix-max principle~\cite[Theorem~4.10]{Teschl2009}, and the proof of items 1, 2, and 3 of Theorem \ref{thm:eigbounds} proceeds analogously to that of the respective items 1, 2, and 4 of~\cite[Theorem 3.3]{FrKo23}. The only difference is that due to the homogeneous Dirichlet boundary conditions in space there are no temporal modes and one replaces the notation $H_\spat^1$ and $\Lm_{k,a}^\spat$ of \cite{FrKo23} with $H^1_\dir(\MM_0)$ and $\Lm_{k,a}$, respectively.

We remark that in version 4 in arXiv (\href{https://arxiv.org/abs/2103.16286}{https://arxiv.org/abs/2103.16286}) we have altered the argument around the last set of display equations in the proof of part~4 of \cite[Theorem 3.3]{FrKo23} to make it clearer.
\end{proof}

\bibliographystyle{myalpha}
\bibliography{bibfile}

@article{FJ15,
  title={On fast computation of finite-time coherent sets using radial basis functions},
  author={Froyland, Gary and Junge, Oliver},
  journal={Chaos},
  volume={25},
  number={8},
  year={2015}
}

@article{ABFS22,
  title={Deep {Lagrangian} connectivity in the global ocean inferred from {Argo} floats},
  author={Abernathey, Ryan and Bladwell, Christopher and Froyland, Gary and Sakellariou, Konstantinos},
  journal={Journal of Physical Oceanography},
  volume={52},
  number={5},
  pages={951--963},
  year={2022}
}

@article{DKF25,
  title={Bathymetry imposes a global pattern of cross-front transport in the {Southern Ocean}},
  author={Denes, Michael C and Keating, Shane R and Froyland, Gary},
  journal={Journal of Physical Oceanography},
  volume={55},
  number={3},
  pages={317--338},
  year={2025}
}

@book{chavel_isoperimetry,
  title={Isoperimetric inequalities: differential geometric and analytic perspectives},
  author={Chavel, Isaac},
  volume={145},
  year={2001},
  publisher={Cambridge University Press}
}

@book{chavel_eigenvalues,
  title={Eigenvalues in {R}iemannian geometry},
  author={Chavel, Isaac},
  year={1984},
  publisher={Academic press}
}

@article{Fro15,
  title={Dynamic isoperimetry and the geometry of {L}agrangian coherent structures},
  author={Froyland, Gary},
  journal={Nonlinearity},
  volume={28},
  number={10},
  pages={3587},
  year={2015}
}

@article{berger,
  title={The eigenvalues of the {L}aplacian with {D}irichlet boundary condition in $\mathbb{R}^2$ are almost never minimized by disks},
  author={Berger, Amandine},
  journal={Annals of Global Analysis and Geometry},
  volume={47},
  number={3},
  pages={285--304},
  year={2015}
}

@article{FRS19,
  title={Sparse eigenbasis approximation: Multiple feature extraction across spatiotemporal scales with application to coherent set identification},
  author={Froyland, Gary and Rock, Christopher P and Sakellariou, Konstantinos},
  journal={Communications in Nonlinear Science and Numerical Simulation},
  volume={77},
  pages={81--107},
  year={2019}
}

@article{FJ18,
  title={Robust {FEM}-based extraction of finite-time coherent sets using scattered, sparse, and incomplete trajectories},
  author={Froyland, Gary and Junge, Oliver},
  journal={SIAM Journal on Applied Dynamical Systems},
  volume={17},
  number={2},
  pages={1891--1924},
  year={2018}
}

@article{FP09,
  title={Almost-invariant sets and invariant manifolds—connecting probabilistic and geometric descriptions of coherent structures in flows},
  author={Froyland, Gary and Padberg, Kathrin},
  journal={Physica D},
  volume={238},
  number={16},
  pages={1507--1523},
  year={2009}
}

@article{FrKw17,
  title={A dynamic Laplacian for identifying Lagrangian coherent structures on weighted Riemannian manifolds},
  author={Froyland, Gary and Kwok, Eric},
  journal={Journal of Nonlinear Science},
  volume={30},
  number={5},
  pages={1889--1971},
  year={2020}
}

@article{berry_variable_2016,
	title = {Variable bandwidth diffusion kernels},
	volume = {40},
	issn = {1063-5203},
	url = {https://www.sciencedirect.com/science/article/pii/S1063520315000020},
	doi = {10.1016/j.acha.2015.01.001},
	number = {1},
	urldate = {2022-05-17},
	journal = {Applied and Computational Harmonic Analysis},
	author = {Berry, Tyrus and Harlim, John},
	month = jan,
	year = {2016},
	pages = {68--96}
}

@article{coifman_graph_2008,
	title = {Graph {Laplacian} {Tomography} {From} {Unknown} {Random} {Projections}},
	volume = {17},
	issn = {1941-0042},
	doi = {10.1109/TIP.2008.2002305},
	number = {10},
	journal = {IEEE Transactions on Image Processing},
	author = {Coifman, Ronald R. and Shkolnisky, Yoel and Sigworth, Fred J. and Singer, Amit},
	month = oct,
	year = {2008},
	note = {Conference Name: IEEE Transactions on Image Processing},
	pages = {1891--1899}
}

@article{padberg-gehle_network-based_2017,
	title = {Network-based study of {Lagrangian} transport and mixing},
	volume = {24},
	issn = {1607-7946},
	url = {https://npg.copernicus.org/articles/24/661/2017/},
	doi = {10.5194/npg-24-661-2017},
	language = {en},
	number = {4},
	urldate = {2022-02-08},
	journal = {Nonlinear Processes in Geophysics},
	author = {Padberg-Gehle, Kathrin and Schneide, Christiane},
	month = oct,
	year = {2017},
	pages = {661--671}
}

@article{CoLa06,
author = {Coifman, Ronald R. and Lafon, St{\'{e}}phane},
doi = {10.1016/j.acha.2006.04.006},
isbn = {1063-5203},
issn = {10635203},
journal = {Applied and Computational Harmonic Analysis},
number = {1},
pages = {5--30},
title = {{Diffusion maps}},
url = {http://www.isa.uni-stuttgart.de/Steinwart/Lehre/WiSe-2012/seminar/DM/DiffusionMaps-Coifman-etal.pdf},
volume = {21},
year = {2006}
}

@article{Str68,
  title={On the construction and comparison of difference schemes},
  author={Strang, Gilbert},
  journal={SIAM Journal on Numerical Analysis},
  volume={5},
  number={3},
  pages={506--517},
  year={1968},
  publisher={SIAM}
}

@article{BaKo17,
  title={Understanding the geometry of transport: Diffusion maps for {L}agrangian trajectory data unravel coherent sets},
  author={Banisch, Ralf and Koltai, P{\'e}ter},
  journal={Chaos: An Interdisciplinary Journal of Nonlinear Science},
  volume={27},
  number={3},
  pages={035804},
  year={2017},
  publisher={AIP Publishing}
}

@article{FrKo23,
author = {Froyland, Gary and Koltai, P\'eter},
title = {Detecting the birth and death of finite-time coherent sets},
journal = {Communications on Pure and Applied Mathematics},
volume = {76},
number = {12},
pages = {3642--3684},
year = {2023},
doi = {https://doi.org/10.1002/cpa.22115},
url = {https://onlinelibrary.wiley.com/doi/abs/10.1002/cpa.22115},
eprint = {https://onlinelibrary.wiley.com/doi/pdf/10.1002/cpa.22115}
}

@book{Teschl2009,
  title={Mathematical Methods in Quantum Mechanics: With Applications to {S}chr{\"o}dinger Operators},
  author={Teschl, Gerald},
  series={Graduate Studies in Mathematics},
  volume={99},
  year={2009},
  publisher={American Mathematical Society}
}

@book{hackbusch2017elliptic,
  title={Elliptic differential equations: theory and numerical treatment},
  author={Hackbusch, Wolfgang},
  volume={18},
  year={2017},
  publisher={Springer}
}

@article{vaughn2019diffusion,
  title={Diffusion maps for embedded manifolds with boundary with applications to {PDE}s},
  author={Vaughn, Ryan and Berry, Tyrus and Antil, Harbir},
  journal={arXiv preprint arXiv:1912.01391},
  year={2019}
}

@article{peoples2021spectral,
  title={Spectral convergence of symmetrized graph {L}aplacian on manifolds with boundary},
  author={Peoples, J Wilson and Harlim, John},
  journal={Foundations of Data Science},
  volume={8},
  pages={119--167},
  year={2026},
  publisher={American Institute of Mathematical Sciences}
}

@article{thiede2019galerkin,
  title={Galerkin approximation of dynamical quantities using trajectory data},
  author={Thiede, Erik H and Giannakis, Dimitrios and Dinner, Aaron R and Weare, Jonathan},
  journal={The Journal of Chemical Physics},
  volume={150},
  number={24},
  year={2019},
  publisher={AIP Publishing}
}

@inproceedings{hein2005graphs,
  title={From graphs to manifolds--weak and strong pointwise consistency of graph {L}aplacians},
  author={Hein, Matthias and Audibert, Jean-Yves and Von Luxburg, Ulrike},
  booktitle={International Conference on Computational Learning Theory},
  pages={470--485},
  year={2005},
  organization={Springer}
}

@article{blachut2020tale,
  title={A tale of two vortices: How numerical ergodic theory and transfer operators reveal fundamental changes to coherent structures in non-autonomous dynamical systems},
  author={Blachut, Chantelle and Gonz{\'a}lez-Tokman, Cecilia},
  journal={Journal of Computational Dynamics},
  volume={7},
  number={2},
  pages={369--399},
  year={2020},
  publisher={Journal of Computational Dynamics}
}

@article{Blachut2023,
	author = {Blachut, Chantelle and Gonz{\'a}lez-Tokman, Cecilia and Hern{\'a}ndez-Due{n}as, Gerardo},
	date = {2023/05/12},
	doi = {10.1007/s00332-023-09911-3},
	id = {Blachut2023},
	isbn = {1432-1467},
	journal = {Journal of Nonlinear Science},
	number = {4},
	pages = {55},
	title = {A Patch in Time Saves Nine: Methods for the Identification of Localised Dynamical Behaviour and Lifespans of Coherent Structures},
	url = {https://doi.org/10.1007/s00332-023-09911-3},
	volume = {33},
	year = {2023}
}

@article{pojehaller98,
  title={Finite time transport in aperiodic flows},
  author={Haller, George and Poje, AC},
  journal={Physica D: Nonlinear Phenomena},
  volume={119},
  number={3-4},
  pages={352--380},
  year={1998}
}

@article{MMP84,
  title={Transport in {H}amiltonian systems},
  author={MacKay, RS and Meiss, JD and Percival, IC},
  journal={Physica D: Nonlinear Phenomena},
  volume={13},
  number={1-2},
  pages={55--81},
  year={1984}
}

@ARTICLE{romkedar90,
  author = {V Rom-Kedar and A Leonard and S Wiggins},
  title = {An analytical study of transport, mixing and chaos in an unsteady vortical flow},
  journal = {Journal of Fluid Mechanics},
  year = {1990},
  volume = {214},
  pages={347 - 394}
}

@ARTICLE{pierrehumbert91,
  author = {R T Pierrehumbert},
  title = {Chaotic mixing of tracer and vorticity by modulated travelling {R}ossby waves},
  journal = {Geophysical \& Astrophysical Fluid Dynamics},
  year = {1991},
  volume = {84(1-4)},
  pages={285-319}
}

@ARTICLE{pierrehumbertyang93,
  author = {R T Pierrehumbert and H Yang},
  title = {Global chaotic mixing on isentropic surfaces},
  journal = {Journal of the atmospheric sciences},
  year = {1993},
  volume = {50(15)},
  pages={2462-2480}
}

@article{F05,
  title={Statistically optimal almost-invariant sets},
  author={Froyland, Gary},
  journal={Physica D: Nonlinear Phenomena},
  volume={200},
  number={3-4},
  pages={205--219},
  year={2005}
}

@article{F13,
  title={An analytic framework for identifying finite-time coherent sets in time-dependent dynamical systems},
  author={Froyland, Gary},
  journal={Physica D: Nonlinear Phenomena},
  volume={250},
  pages={1--19},
  year={2013}
}

@article{DJ99,
  title={On the approximation of complicated dynamical behavior},
  author={Dellnitz, Michael and Junge, Oliver},
  journal={SIAM Journal on Numerical Analysis},
  volume={36},
  number={2},
  pages={491--515},
  year={1999}
}

@article{DFK22,
  title={Persistence and material coherence of a mesoscale ocean eddy},
  author={Denes, Michael C and Froyland, Gary and Keating, Shane R},
  journal={Physical Review Fluids},
  volume={7},
  number={3},
  pages={034501},
  year={2022}
}

@article{KW90,
  title={Transport in two-dimensional maps},
  author={Rom-Kedar, Vered and Wiggins, Stephen},
  journal={Archive for Rational Mechanics and Analysis},
  volume={109},
  pages={239--298},
  year={1990},
  publisher={Springer-Verlag}
}

@article{BadzaFroyland24,
  title={Identifying the onset and decay of quasi-stationary families of almost-invariant sets with an application to atmospheric blocking events},
  author={Badza, Aleksandar and Froyland, Gary},
  journal={Chaos: An Interdisciplinary Journal of Nonlinear Science},
  volume={34},
  number={12},
  year={2024}
}

@article{Badzaetal26,
  title={Convective heat transfer patterns in plane-layer Rayleigh-B{\'e}nard convection revealed by transient almost-invariant sets},
  author={Badza, Aleksandar and Froyland, Gary and Samuel, Roshan J and Schumacher, J{\"o}rg},
  journal={Physical Review Research},
  volume={8},
  number={2},
  pages={023040},
  year={2026}
}

@ARTICLE{shadden-etal,
  author = {Shawn C Shadden and Francois Lekien and Jerrold E Marsden},
  title = {Definition and properties of {L}agrangian coherent structures from finite-time {L}yapunov exponents in two-dimensional aperiodic flows},
  journal = {Physica D},
  year = {2005},
  volume = {212(3)},
  pages ={271-304}
}

@article{macmillan20,
  title={Detection of evolving {L}agrangian coherent structures: A multiple object tracking approach},
  author={MacMillan, Theodore and Ouellette, Nicholas T and Richter, David H},
  journal={Physical Review Fluids},
  volume={5},
  number={12},
  pages={124401},
  year={2020}
}

@article{AnKaBV20,
  title={Genesis, evolution, and apocalypse of Loop Current rings},
  author={Andrade-Canto, F and Karrasch, D and Beron-Vera, FJ},
  journal={Physics of Fluids},
  volume={32},
  number={11},
  pages={116603},
  year={2020},
  publisher={AIP Publishing LLC}
}

@ARTICLE{FSM10,
  author = {Froyland, Gary and Santitissadeekorn, Naratip and Monahan, Adam},
  title = {Transport in time-dependent dynamical systems: Finite-time coherent
	sets},
  journal = {Chaos},
  year = {2010},
  volume = {20(4)},
  pages = {043116}
}

@article{ndouretal21,
  title={Spectral Early-Warning Signals for Sudden Changes in Time-Dependent Flow Patterns},
  author={Ndour, Moussa and Padberg-Gehle, Kathrin and Rasmussen, Martin},
  journal={Fluids},
  volume={6},
  number={2},
  pages={49},
  year={2021},
  publisher={Multidisciplinary Digital Publishing Institute}
}

@article{FrEtAl12,
  title={Three-dimensional characterization and tracking of an {A}gulhas {R}ing},
  author={Froyland, Gary and Horenkamp, Christian and Rossi, Vincent and Santitissadeekorn, Naratip and Gupta, Alex Sen},
  journal={Ocean Modelling},
  volume={52},
  pages={69--75},
  year={2012},
  publisher={Elsevier}
}

@article{KoCiSch16,
	Author = {P\'eter Koltai and Giovanni Ciccotti and Christof Sch{\"u}tte},
	Journal = {The Journal of Chemical Physics},
	Note = {Editors' Choice article},
	Number = {17},
	Pages = {174103},
	Publisher = {AIP Publishing},
	Title = {On metastability and {M}arkov state models for non-stationary molecular dynamics},
	Volume = {145},
	Year = {2016}
}

@article{ElA21,
author = {El Aouni, Anass },
title = {A hybrid identification and tracking of {L}agrangian mesoscale eddies},
journal = {Physics of Fluids},
volume = {33},
number = {3},
pages = {036604},
year = {2021},
doi = {10.1063/5.0038761},
URL = {
        https://doi.org/10.1063/5.0038761
},
eprint = {
        https://doi.org/10.1063/5.0038761
}
}

@article{budivsic2012geometry,
  title={Geometry of the ergodic quotient reveals coherent structures in flows},
  author={Budi{\v{s}}i{\'c}, Marko and Mezi{\'c}, Igor},
  journal={Physica D: Nonlinear Phenomena},
  volume={241},
  number={15},
  pages={1255--1269},
  year={2012},
  publisher={Elsevier}
}

@article{AlPe15,
author = {Allshouse, Michael R. and Peacock, Thomas},
doi = {10.1063/1.4922968},
issn = {1054-1500},
journal = {Chaos: An Interdisciplinary Journal of Nonlinear Science},
number = {9},
pages = {097617},
title = {{Lagrangian based methods for coherent structure detection}},
url = {http://scitation.aip.org/content/aip/journal/chaos/25/9/10.1063/1.4922968},
volume = {25},
year = {2015}
}

@article{ma2014differential,
  title={Differential geometry perspective of shape coherence and curvature evolution by finite-time nonhyperbolic splitting},
  author={Ma, Tian and Bollt, Erik M},
  journal={SIAM Journal on Applied Dynamical Systems},
  volume={13},
  number={3},
  pages={1106--1136},
  year={2014},
  publisher={SIAM}
}

@article{haller2013coherent,
  title={Coherent {L}agrangian vortices: the black holes of turbulence},
  author={Haller, G and Beron-Vera, FJ},
  journal={Journal of Fluid Mechanics},
  volume={731},
  year={2013},
  publisher={Cambridge University Press}
}

@article{rypina2011investigating,
  title={Investigating the connection between complexity of isolated trajectories and {L}agrangian coherent structures},
  author={Rypina, Irina I and Scott, SE and Pratt, Lawrence J and Brown, Michael G},
  journal={Nonlinear Processes in Geophysics},
  volume={18},
  number={6},
  pages={977--987},
  year={2011},
  publisher={Copernicus GmbH}
}

@article{allshouse2012detecting,
  title={Detecting coherent structures using braids},
  author={Allshouse, Michael R and Thiffeault, Jean-Luc},
  journal={Physica D: Nonlinear Phenomena},
  volume={241},
  number={2},
  pages={95--105},
  year={2012},
  publisher={Elsevier}
}

@article{lekien2010computation,
  title={The computation of finite-time {L}yapunov exponents on unstructured meshes and for non-{E}uclidean manifolds},
  author={Lekien, Francois and Ross, Shane D},
  journal={Chaos: An Interdisciplinary Journal of Nonlinear Science},
  volume={20},
  number={1},
  year={2010},
  publisher={AIP Publishing}
}

@article{joseph2002relation,
  title={Relation between kinematic boundaries, stirring, and barriers for the {A}ntarctic polar vortex},
  author={Joseph, Binson and Legras, Bernard},
  journal={Journal of the Atmospheric Sciences},
  volume={59},
  number={7},
  pages={1198--1212},
  year={2002}
}

@article{bowman1993large,
  title={Large-scale isentropic mixing properties of the {A}ntarctic polar vortex from analyzed winds},
  author={Bowman, Kenneth P},
  journal={Journal of Geophysical Research: Atmospheres},
  volume={98},
  number={D12},
  pages={23013--23027},
  year={1993},
  publisher={Wiley Online Library}
}

@article{rypina2007lagrangian,
  title={On the {L}agrangian dynamics of atmospheric zonal jets and the permeability of the stratospheric polar vortex},
  author={Rypina, II and Brown, Michael G and Beron-Vera, Francisco J and Ko{\c{c}}ak, Huseyin and Olascoaga, Maria J and Udovydchenkov, IA},
  journal={Journal of the Atmospheric Sciences},
  volume={64},
  number={10},
  pages={3595--3610},
  year={2007}
}

@article{olascoaga2012brief,
  title={Brief communication ``{S}tratospheric winds, transport barriers and the 2011 {A}rctic ozone hole''},
  author={Olascoaga, MJ and Brown, MG and Beron-Vera, FJ and Ko{\c{c}}ak, H},
  journal={Nonlinear Processes in Geophysics},
  volume={19},
  number={6},
  pages={687--692},
  year={2012},
  publisher={Copernicus Publications G{\"o}ttingen, Germany}
}

@article{beron2010invariant,
  title={Invariant-tori-like {L}agrangian coherent structures in geophysical flows},
  author={Beron-Vera, Francisco J and Olascoaga, Mar{\'\i}a J and Brown, Michael G and Ko{\c{c}}ak, Huseyin and Rypina, Irina I},
  journal={Chaos: An Interdisciplinary Journal of Nonlinear Science},
  volume={20},
  number={1},
  year={2010},
  publisher={AIP Publishing}
}

@Article{FrJuKo13,
	author = {Gary Froyland and Oliver Junge and P{\'e}ter Koltai},
	title = {Estimating long term behavior of flows without trajectory integration: the infinitesimal generator approach},
	journal = {SIAM J. Numer. Anal.},
	year = {2013},
	volume = {51},
	number = {1},
	pages = {223--247}
}

@article{BERRY2016439,
title = {Local kernels and the geometric structure of data},
journal = {Applied and Computational Harmonic Analysis},
volume = {40},
number = {3},
pages = {439-469},
year = {2016},
issn = {1063-5203},
doi = {https://doi.org/10.1016/j.acha.2015.03.002},
url = {https://www.sciencedirect.com/science/article/pii/S106352031500024X},
author = {Tyrus Berry and Timothy Sauer}
}

@Article{SchoeEtAl25,
AUTHOR = {Schoeller, H. and Chemnitz, R. and Koltai, P. and Engel, M. and Pfahl, S.},
TITLE = {Assessing {L}agrangian coherence in atmospheric blocking},
JOURNAL = {Nonlinear Processes in Geophysics},
VOLUME = {32},
YEAR = {2025},
NUMBER = {1},
PAGES = {51--73},
URL = {https://npg.copernicus.org/articles/32/51/2025/},
DOI = {10.5194/npg-32-51-2025}
}

@article{vieweg2024lagrangian,
  title={Lagrangian studies of coherent sets and heat transport in constant heat flux-driven turbulent {R}ayleigh--{B}{\'e}nard convection},
  author={Vieweg, Philipp P and Kl{\"u}nker, Anna and Schumacher, J{\"o}rg and Padberg-Gehle, Kathrin},
  journal={European Journal of Mechanics-B/Fluids},
  volume={103},
  pages={69--85},
  year={2024},
  publisher={Elsevier}
}

@Article{KWNS18,
AUTHOR = {Koltai, P\'eter and Wu, Hao and No\'e, Frank and Sch\"utte, Christof},
TITLE = {Optimal Data-Driven Estimation of Generalized {M}arkov State Models for Non-Equilibrium Dynamics},
JOURNAL = {Computation},
VOLUME = {6},
YEAR = {2018},
NUMBER = {1},
URL = {http://www.mdpi.com/2079-3197/6/1/22},
ISSN = {2079-3197},
DOI = {10.3390/computation6010022}
}

\end{document}